\documentclass[10pt]{article}
\def\commentflag{3}
\def\ct{3}

\usepackage{titlesec}
\usepackage{amssymb,amsmath,amsfonts,amsthm,enumerate,color,mathrsfs,amssymb,extarrows,empheq,cite}
\usepackage[shortlabels]{enumitem}
\usepackage[toc,title,titletoc,header]{appendix}

\usepackage{esint}

\usepackage{etoolbox}
\patchcmd{\thebibliography}
  {\settowidth}
  {\setlength{\itemsep}{0pt plus 0.1pt}\settowidth}
  {}{}
\apptocmd{\thebibliography}
  {\small}
  {}{}

\usepackage{epsfig}
\usepackage[hidelinks]{hyperref}
\usepackage{graphicx}
\usepackage{indentfirst}
\usepackage{multicol}
\usepackage{booktabs} 
\usepackage{float} 
\usepackage{subfigure}

\usepackage[usenames,dvipsnames]{xcolor}

\usepackage[hidelinks]{hyperref}
\hypersetup{colorlinks = true,
linkcolor = MidnightBlue, anchorcolor =red,
citecolor = ForestGreen,
filecolor = red,urlcolor = red,
            pdfauthor=author}
\usepackage{cleveref}

\usepackage{soul,todonotes}

\newcommand{\mb}[1]{{\color{black} #1}}

\theoremstyle{plain}

\newtheorem{theorem}{Theorem}[section]
\newtheorem{lemma}[theorem]{Lemma}
\newtheorem{proposition}[theorem]{Proposition}
\newtheorem{corollary}[theorem]{Corollary}

\numberwithin{equation}{section}

\theoremstyle{remark}
\newtheorem{remark}[theorem]{Remark}

\theoremstyle{definition}
\newtheorem{definition}[theorem]{Definition}

\newcommand{\norm}[1]{\lVert #1 \rVert}

\def\q {\quad}
\def \qq {\qquad}
\def \l{\langle}
\def \r{\rangle}

\def\bb{\begin{equation}
  \left\{\ 
   \begin{aligned} }
\def\ee{   \end{aligned}
  \right.
  \end{equation}}

\def\mm{ \left[
 \begin{matrix}}
\def\nn{\end{matrix} \right] } 

\def\p{\partial}
\def \dd{\cdot}
\def \t{\times}

\def\n{\nu}
\def \w {\widetilde}
\def \h{\hat}

\def\d{\delta}
\def \vp{\varphi}
\def \na{\nabla}
\def \ep{\varepsilon}
\def \lad{\lambda}
\def \si {\sigma}

\def \Lad{\Lambda}

\def\re{\mathfrak{Re}}
\def\im{\mathfrak{Im}}

\def \curl{\nabla \times}

\def \ccurl{{\rm  curl}}
\def \ddiv {{\rm div}}

\def \sdiv {\sdiv_S}

\def \ms {\mathscr}
 
\def \mbb {\mathbb}

\def \td{\mathcal{T}_D^{\alpha,\ww}}

\def \kd{\mathcal{K}_D^{\alpha,\ww}}

\def \R{\mathbb{R}}

\def \C{\mathbb{C}}

\def \S{\mathbb{S}}

\def \P{\mathbb{P}}

\def \TT{\mathcal{T}}

\def \np{\mathcal{K}}

\def \M{\mathcal{M}}
\def \A{\mathcal{A}}
\def \L{\mathcal{L}}
\def \gr{\mathcal{G}}

\def \ww{\omega}

\def \di{{\rm{\bf d}}}

\def\D{\mathcal{D}}

\def \pd {\mathbb{P}_{{\rm d}}} 
\def \pw {\mathbb{P}_{{\rm w}}}

\def \hzz {{\bf H}_0(\ddiv 0, D)}

\def \ei {{\rm {\bf p}}^i}

\def \po {{\rm {\bf p}}}

\newcommand{\tb}[1]{\mathcal{T}_{D,#1}}
\newcommand{\kb}[1]{\mathcal{K}_{D,#1}}

\newcommand {\mc}[1]{\mathcal{#1}}

\def \sss {\scriptscriptstyle}

\DeclareMathOperator{\dom}{dom}
\DeclareMathOperator{\tr}{tr}

\newenvironment{myitemize}
{ \begin{itemize}
    \setlength{\itemsep}{0pt}
    \setlength{\parskip}{0pt}
    \setlength{\parsep}{0pt}     }
{ \end{itemize}                  }

\title{Fano resonances in all-dielectric electromagnetic metasurfaces}

\begin{document}
\author{
Habib Ammari\footnote{Department of Mathematics, ETH Z\"{u}rich, R\"{a}mistrasse 101, CH-8092 Z\"{u}rich, Switzerland. 
The work of this author is partially supported by the
Swiss National Science Foundation (SNSF) grant 200021–200307.
(habib.ammari@math.ethz.ch).}
\and Bowen Li\footnote{Department of Mathematics, Duke University, Durham, NC 27708, USA. (bowen.li@math.duke.edu).}
\and Hongjie Li\footnote{Department of Mathematics, The Chinese University of Hong Kong, Shatin, N.T., Hong Kong. (hjli@math.cuhk.edu.hk).}
\and Jun Zou\footnote{Department of Mathematics, The Chinese University of Hong Kong, Shatin, N.T., Hong Kong.
The work of this author was
substantially supported by Hong Kong RGC General Research Fund (projects 14306921 and 14306719) and NSFC/Hong Kong RGC Joint Research Scheme 2022/23 (project N\_CUHK465/22). 
(zou@math.cuhk.edu.hk).}
}
\date{}
\maketitle

\begin{abstract}
We are interested in the resonant electromagnetic (EM) scattering by all-dielectric metasurfaces made of a two-dimensional lattice of nanoparticles with high refractive indices. In [Ammari et al.,  Trans. AMS, 376 (2023), 39-90], it has been shown that a single high-index nanoresonator can couple with the incident wave and exhibit a strong magnetic dipole response. Recent physics experiments reveal that when the particles are arranged in certain periodic configurations, they may have different anomalous scattering effects in the macroscopic scale, compared to the single-particle case. In this work, we shall develop a rigorous mathematical framework for analyzing the resonant behaviors of all-dielectric metasurfaces. We start with the characterization of subwavelength scattering resonances in such a periodic setting and their asymptotic expansions in terms of the refractive index of the nanoparticles. Then we show that real resonances always exist below the essential spectrum of the periodic Maxwell operator and that they are the simple poles of the scattering resolvent with the exponentially decaying resonant modes. 
By using group theory, we discuss the implications of the symmetry of the metasurface on the subwavelength band functions and their associated eigenfunctions. For the symmetric metasurfaces with the normal incidence, we use a variational method to show the existence of embedded eigenvalues (i.e., real subwavelength resonances embedded in the continuous radiation spectrum). \mb{Furthermore, we break the configuration symmetry either by introducing a small deformation of particles or by slightly deviating from the normal incidence and prove that Fano-type reflection and transmission anomalies can arise in both of these scenarios.}

\medskip

\noindent \textbf{Keywords:} bound states in the continuum, Fano resonances, all-dielectric metasurface, scattering resonances, high-contrast nanoparticles

\medskip

\noindent  \textbf{Mathematics Subject Classification (MSC2000):} 35C20, 35J20, 35P20, 78M30
\end{abstract}

\section{Introduction}
The study of electromagnetic metamaterials has been a very topical subject in the field of nanophotonics over the last few years, due to their great potential in wavefront control, field concentration, and light focusing at subwavelength scales \cite{Li2019, zhang2016advances,chen2016review}. The plasmonic nanoresonators (e.g., metallic nanoparticles) are one of the widely used building blocks for novel optical metamaterials \cite{sarid2010modern}; see the monograph  \cite{ammari2018mathematical} and the references therein for mathematical treatment.  However, the practical applications of the plasmonic metamaterials are severely limited by the heat dissipation induced by the imaginary part of the electric permittivity in the visible light range. This motivates physicists and engineers to search for alternatives to the plasmonic elements. 
Recently, dielectric nanoresonators (e.g., silicon nanoparticles) have received considerable attention for their low intrinsic losses and high electric permittivity. It was experimentally demonstrated in \cite{evlyukhin2012demonstration,kuznetsov2016optically,garcia2011strong,luk2015optimum} that a single silicon nanoparticle can resonate in the optical realm with a high Q-factor (i.e., the so-called dielectric resonances), and the excited electric and magnetic dipole moments can have comparable order of magnitudes. Such properties are desirable and attractive in many applications in imaging science, material science, and wireless communications, and lead to a superior performance of the all-dielectric metamaterials over the lossy plasmonic devices \cite{staude2017metamaterial,jahani2016all,kuznetsov2016optically,krasnok2012all,ammari2020superresolution}. 
For the case of spherical nanoparticles, the dielectric resonance can be well understood by the Mie scattering theory \cite{tzarouchis2018light}, while for the TM or TE polarization case, the quasi-static dielectric resonance is equivalent to the eigenvalue problem for the Newtonian potential \cite{ammari2019subwavelength}; see also \cite{meklachi2018asymptotic} for the discussion on the nonlinear Kerr-type material. The complete mathematical theories for the cases of a single dielectric nanoparticle and a cluster of dielectric nanoparticles with the full Maxwell's equations have been achieved in \cite{ammari2020mathematical} and \cite{cao2022electromagnetic}, respectively. 

Metasurfaces are two-dimensional composite materials with thicknesses much smaller than the operating wavelength, which enable the complete control of the phase, amplitude, and polarization of the scattered wave \cite{ sun2019electromagnetic,huidobro2016graphene,leroy2015superabsorption}. We refer the readers to \cite{LLZ0527, ammari2017mathematical,ammari2020meta,ammari2020equivalent} for recent mathematical results on plasmonic metasurfaces and bubble meta-screens. Thanks to the unique optical properties of the resonant dielectric nanoparticles, great effort has also been made to develop all-dielectric metasurfaces for realizing high-efficiency light manipulation. The current work is mainly motivated by \cite{koshelev2018asymmetric,li2019symmetry}, where the authors designed a class of dielectric metasurfaces with broken in-plane symmetry that can support the bound states in the continuum and allow the Fano-type transmission and reflection. Fano resonance is a special resonant scattering phenomenon corresponding to the asymmetric resonance peaks in transmission spectra, initially investigated by Ugo Fano \cite{fano1961effects} in quantum mechanics. Such phenomenon has been extensively explored for various photonic periodic structures since the 1990s \cite{limonov2017fano,shipman2010resonant}. 
The basic physical origin of the Fano resonance lies in the interference between the continuous state in the background and the discrete resonant state from the scattering system, as pointed out by Fano in his original work \cite{fano1961effects}. 
Mathematically, it is closely related to the resonances embedded in the continuous spectrum of the underlying operators of the 
interested models and the associated bound states in the continuum (BICs), which are generally nonrobust with respect to the perturbations. In more detail, for the ideal system where the embedded eigenvalues exist, the BICs are the resonant states localized to the nanoresonators and decoupled from the far field. When the system is slightly perturbed, the embedded eigenvalue will be shifted into the lower complex half-plane with a small imaginary part, and the bound states become leaky and will interact with the broad background radiation, which gives rise to the sharp transmission peaks and dips. This mechanism has been utilized in many nanophotonic materials, e.g., the metallic slab with arrays of subwavelength holes \cite{hsu2016bound,zhen2014topological}, the plasmonic metasurface \cite{luk2010fano}, and the arrays of dielectric spheres and rods  \cite{bulgakov2019high,sidorenko2021observation,bulgakov2018optical}.

The existence of embedded eigenvalues (equivalently, BICs) for the EM grating of periodic conductors was first rigorously proved by Bonnet-Bendhia and Starling \cite{bonnet1994guided} under some symmetry assumptions. Their arguments are based on a variational formulation of the problem with some ideas borrowed from the earlier work by N\'{e}d\'{e}lec and Starling \cite{nedelec1991integral}. 
Later, Shipman and Volkov \cite{shipman2007guided} characterized the robust and nonrobust bound states for the homogeneous periodic slabs and gave sufficient conditions for the nonexistence of guided modes. Moreover, in \cite{shipman2010resonant,shipman2005resonant,shipman2012total,shipman2013resonant}, Shipman et al. considered the asymptotics of the embedded eigenvalues with respect to the perturbation of the Bloch wave vector (which breaks the symmetry of the grating) and analyzed the Fano-type transmission anomalies by the analytic perturbation theory. Recently, Lin et al. \cite{lin2020mathematical} quantitatively investigated the Fano resonance phenomenon and the field enhancement in the setting of a metallic grating with subwavelength slits by using layer potential techniques. 
In their subsequent work \cite{lin2021fano}, they focused on the case of the perfectly conducting grating with strongly coupled slits and analyzed the BICs and the corresponding Fano resonances. A similar configuration for the bubble meta-screen was discussed in \cite{ammari2021bound}, where Ammari et al. used the capacitance matrix to explicitly approximate the scattering matrix and justify the Fano-type anomalies. 
In all of these works \cite{lin2020mathematical,lin2021fano,ammari2021bound}, the authors showed that there exist two classes of subwavelength resonances, where one of them has a much smaller imaginary part than the other one. The discrete resonant state corresponds to the resonance with the larger imaginary part and gives a broad resonance peak, while the continuum state with the smaller imaginary part is perturbed from the bound state and has a narrow resonance peak. It is exactly the hybridization of these two states that gives rise to a Fano transmission phenomenon. However, the rigorous analysis for the full Maxwell system has not been well developed yet. We are aiming to fill this gap.

In this work, we will establish a unified mathematical framework for EM resonant scattering by asymmetric dielectric metasurfaces of nanoparticles with high refractive indices. For such structures, the BICs and the Fano resonances have been experimentally reported in \cite{koshelev2018asymmetric}. Our analysis is based on both the variational method and the Lippmann-Schwinger equation for the full Maxwell's equations. As in the case of the single dielectric nanoparticles \cite{ammari2020mathematical}, the dielectric resonances are defined as the complex poles of the scattering resolvent of the Maxwell operator, which are located in the lower complex plane $\{\ww \in \C\,;\ \im \ww \le 0\}$. We shall be particularly interested in the real subwavelength dielectric resonances. 
For convenience, we will use the words \emph{real resonances} and \emph{eigenvalues} interchangeably in what follows, if there is no confusion. In Section \ref{sec:general},  we start with the notion of the subwavelength scattering resonances in the periodic setting. We prove that the real resonances are the simple poles of the resolvent of the volume integral operator for the EM scattering problem and that the associated resonant mode exponentially decays in the far field and is decoupled from the incoming wave; see Propositions \ref{lem:expdecay} and \ref{prop:real}. Hence, we can relate the real scattering resonances to the eigenvalue problem of the unbounded Maxwell operator on the $L^2$-space. We proceed to characterize the essential spectrum of the periodic Maxwell operator, which turns out to be the continuous radiation spectrum, and show that there is always a discrete eigenvalue below the infimum of the essential spectrum; see Theorem \ref{thm:ess_discre}. However, the resonant states for the eigenvalue below the essential spectrum do not interact with the background radiation, and thus would not induce any Fano-type resonance phenomenon. 

As discussed above, the desired bound states in the continuum correspond to the real scattering resonances embedded in the continuous essential spectrum, the existence of which is a difficult task and usually depends on the additional structure of the configuration, e.g., 
the symmetry assumption. In this work, we mainly focus on the in-plane inversion symmetry \eqref{assp:sym} motivated by the physical setup \cite{koshelev2018asymmetric,li2019symmetry}. In Section \ref{sec:sym_band_multi}, we discuss the general symmetry properties of the dielectric metasurfaces by using the representation theory of the symmetry group and its consequences on the resonances and the resonant modes. With these preparations, in Section \ref{sec:exist_real}, we manage to show the existence of the real subwavelength resonances which are embedded in the continuous spectrum in the high contrast regime (i.e., the refractive indices of the nanoparticles are large enough), under the normal incidence and a mild assumption on the limiting problem; see Theorem \ref{main_real_2}. In doing so, we adopt the regularized variational formulation for the EM scattering problem \cite{bao1997variational,bao2000scattering} and follow the analysis framework of \cite{bonnet1994guided,shipman2010resonant} for the acoustic case. It is worth emphasizing that our arguments are much more involved with the main difficulty coming from the infinite-dimensional kernel of the curl operator. In addition, to justify that the real resonance is in the subwavelength regime for the large contrast, 
it is necessary to investigate the limiting eigenvalue problem, for which we exploit the techniques from \cite{simon1978canonical, hempel2000spectral} on the convergence of the monotone sesquilinear form. 

It is known that the bound states associated with such embedded eigenvalues are unstable with respect to the perturbation of the metasurface that destroys the symmetry \eqref{assp:sym}. In Section \ref{sec:fano_scattering}, we slightly perturb the normal incidence and the symmetric dielectric metasurface by a deformation field as in \cite{koshelev2018asymmetric,li2019symmetry}. In this case, the BICs become the quasi-modes with finite Q-factors. We shall show that for the perturbed asymmetric metasurface, the reflection energy could present a sharp asymmetric shape of Fano-type. For this, we first derive the asymptotic expansions of the subwavelength resonances with respect to the high contrast in Theorem \ref{asym:resonance}. Then, we calculate the shape derivative of the leading-order approximation of the subwavelength resonance; see Proposition \ref{lem:shapedev}. This allows us to find the asymptotics of the resonances for the perturbed metasurface in Corollary \ref{coro:appro_reso}, where we assume that the eigenspace of the leading-order operator is spanned by a symmetric field and an antisymmetric one for simplicity. To quantify the expected Fano resonance phenomenon, we calculate the quasi-static approximation of the scattered polarization vectors that uniformly holds for the incident frequencies near the resonances, which readily implies the asymptotic expansion of the reflection and transmission energy; see Corollary \ref{coro:approx_energy}. Thanks to the energy conservation property, we only focus on the reflection energy and demonstrate that it can be effectively described by the classical Fano formula \cite{koshelev2018asymmetric,shipman2010resonant}, meaning that there is a sharp asymmetric line shape in the reflection spectrum. 

\vspace{- 1 mm}
\subsection*{Notations}
\vspace{- 1 mm}
\begin{myitemize}
\item  We write $(x' ,x_3)$ for $x = (x_1, x_2, x_3) \in \R^3$
with $x' = (x_1, x_2) \in \R^2 $ and $ x_3 \in \R$. 
\item We \mb{employ} standard asymptotic notations: for $f$ in a normed vector space, we denote $f = O(\epsilon)$ if $\norm{f}\le C |\epsilon|$ holds for a constant $C > 0$ independent of $\epsilon$, and $f = o(\epsilon)$ if $\norm{f}/|\epsilon| \to 0$ as $\epsilon \to 0$. Moreover, for real $a,b >0$, we write $a \ll b$ if $a \le C b$ for a sufficiently small constant $C$ independent of $a,b$, and $a \gg b$ if $a \ge C b$ for large enough $C$ independent of $a,b$.
\item Let $\Lambda$ be a two-dimensional lattice generated by the vectors $v_1, v_2 \in \R^2$:
\begin{equation*}
  \Lambda := \big\{\gamma \in \R^2\,;\ \gamma = n_1 v_1 + n_2 v_2\,,\ n_i \in \mathbb{Z}\big\}\,,
\end{equation*} 
with the fundamental cell $Y$:
\begin{equation*}
Y := \big\{\gamma \in \R^2\,;\  \gamma = c_1 v_1 + c_2 v_2\,,\ c_i \in [-1/2,1/2]\big\}\,.
\end{equation*}
For simplicity, the area of $Y$ is assumed to be one, i.e., $|Y| = 1$.  We also define some fundamental cells in $\R^3$ by
\begin{align*}
    Y_\infty := Y \t \R\,,\q\  Y_h := Y \t [-h, h]\,,\q \text{for}\  h > 0\,, 
\end{align*} 
for later use. Moreover, we introduce the reciprocal lattice $\Lambda^*$ by
\begin{align} \label{def:dual_lattice}
\Lambda^* := \big\{q \in \R^2\,;\ q \dd \gamma \in 2 \pi \mathbb{Z} \q \text{for any}\ \gamma \in \Lambda\big\}\,,
\end{align} 
and the first Brillouin zone of $\Lambda^*$ by
\begin{align} \label{def:bri_zone}
    \mc{B} := \{\alpha \in \R^2\,;\ |\alpha| \le |\alpha - q|\q \forall q \in \Lambda^*\}\,.
\end{align} 
\item A function (or vector field) $f$ on $\R^3$ is $\alpha$-quasi-periodic \mb{($\alpha \in \R^2$) in variables $x_1$ and $x_2$} if 
\begin{align} \label{def:quasi-period}
    f(x + (\gamma, 0)) = e^{i\alpha \dd \gamma}f(x)\,,\q  \gamma \in \Lambda\,.
\end{align}
\item We use the standard Sobolev spaces on a three-dimensional domain $D \subset \R^3$ or a two-dimensional surface $\Sigma \subset \R^3$. The bold typeface is used to indicate the spaces of vector fields, e.g., ${\bf L}^2(D)$ is the space of $L^2$-integrable vector fields on $D$. We use the subscript $loc$ to denote the spaces of locally integrable functions (fields), and
the subscript $\alpha$ to denote the Sobolev spaces of $\alpha$-quasi-periodic functions (fields), e.g.,
\begin{align*}
    {\bf H}_{\alpha,loc}(\ccurl, Y_\infty): = \{\vp \in {\bf L}^2_{loc}(Y_\infty)\,;\ \ccurl \vp \in {\bf L}^2_{loc}(Y_\infty)\,,\ \vp \ \text{satisfies}\ \eqref{def:quasi-period}\}\,.
\end{align*}
In the case of $\alpha = 0$, we have the spaces of periodic functions and shall use the subscript $p$ to avoid confusion, e.g., the space ${\bf H}_p^1(Y_h)$ consists of periodic $H^1$-fields. 
\mb{In addition, we use} the subscript $T$ to denote the tangential vector fields on the surface $\Sigma$, e.g., ${\bf H}_T^{\sss -1/2}(\ccurl, \Sigma):= \{\vp \in {\bf H}^{\sss -1/2}_{T}(\Sigma)\,;\  \ccurl_{\Sigma} \vp \in H^{\sss -1/2}(\Sigma)\}$, where $\ccurl_{\Sigma}$ is the \mb{surface scalar curl}. We refer the readers to  \cite{nedelec2001acoustic} for the definition of surface differential operators. In particular, for the reader's convenience, we recall the space of divergence-free vector fields ${\bf H}(\ddiv0, D) = \{\vp \in {\bf L}^2(D)\,;\ \ddiv \vp = 0\}$ and its subspace \mb{of fields with vanishing normal traces}:
$${\bf H}_0(\ddiv0, D) = \{\vp \in {\bf H}(\ddiv0, D)\,;\ \n \dd \vp = 0\ \text{on}\ \p D\}\,.$$ 
We also need the weighted $L^2$-space: for a nonnegative function $\ep \in L^\infty(Y_\infty)$, $${\bf L}_{\alpha,\ep}^2(Y_\infty): = \{f\,; \ f\ \text{is a measurable vector field satisfying \eqref{def:quasi-period} and}\  \ep|f|^2 \in L^1(Y_\infty)\}\,.$$
Similarly, we define 
\begin{align} \label{def:diverzerospace}
    {\bf H}_\alpha(\ep; \ddiv0, Y_\infty): = \{f \in {\bf L}_{\alpha,\ep}^2(Y_\infty)\,;\ \ddiv(\ep f) = 0\}\,.
\end{align}
The tangential component trace for a $H(\ccurl)$-vector field is defined by
\begin{align} \label{eq:trace_tang}
    \pi_t(u):= (\n \t u)\t \n\,.
\end{align}
\item For two Banach spaces $X$ and $Y$, we denote by $\L(X,Y)$ the bounded linear operators from $X$ to $Y$, or simply by $\L(X)$ if $Y = X$. We write $\norm{\dd}_X$ for the norm on the space $X$, or simply, write $\norm{\dd}$, when no confusion is caused. In particular, for notation simplicity, we denote by $(\dd,\dd)_D$ the $L^2$-inner product on $D$, and by $(\dd,\dd)_{\ep,D}$ the weighted $L^2$-inner product on $D$, i.e., $(f,g)_{\ep, D} = \int_D \ep \bar{f} g\ dx$, while the standard complex dual pairing on the surface $\Sigma$ is denoted by $\l \dd,\dd \r_\Sigma$. 

\end{myitemize}

\section{Scattering resonances of all-dielectric metasurfaces} \label{sec:general}
In this section, we consider the resonant electromagnetic (EM) scattering by periodically distributed dielectric nanoparticles with high refractive indices.
We will first define the scattering resonances as the poles of the associated scattering resolvent and introduce the band dispersion functions. Then in Section \ref{sec:band}, we show that the real resonances are the simple poles of the resolvent and the corresponding resonant modes are exponentially decaying in the far field. We also prove that the real resonances always exist below the light line when the refractive index of the nanoparticle is larger than the background. In Section \ref{sec:sym_band_multi}, we provide some general results on the symmetry properties of the band function and the multiplicity of the scattering resonance based on the symmetry group of the metasurface. 

We begin with the formulation of the periodic electromagnetic scattering problem.  Suppose that the nanoparticles, \mb{denoted as $D$,} are contained inside $Y_h$ for some $h > 0$, where $D$ is a bounded open set with a smooth boundary $\p D$, \mb{a unit outer normal vector $\nu$,}
and a characteristic size of order one. The all-dielectric metasurface is defined as the collection of nanoparticles that are periodically distributed on the lattice $\Lambda$, which is denoted by  
\begin{align} \label{def:D}
\D := \bigcup_{\gamma \in \Lambda} \big(D + (\gamma,0)\big)\,. 
\end{align}
Moreover, we assume that the nanoparticles are non-magnetic (i.e., the magnetic permeability $\mu = 1$ on $\R^3$) and have the electric permittivity of the form:
\begin{align} \label{assp:high_constrast}
\ep := 1 + \tau \chi_{\mc{D}}\,, \q \tau \gg 1\,,\q \text{on} \ \R^3\,,
\end{align}
where $\chi_{\mc{D}}$ is the characteristic function of the set $\mc{D}$, and $\tau \in \R$ is the contrast parameter. \mb{From  the standard Floquet-Bloch theory \cite{kuchment2001mathematics,kuchment2016overview,kirsch2018limiting},
the EM scattering problem by the lattice of the nanoresonators $\mc{D}$ can be modeled by the following family of $\alpha$-quasi-periodic scattering problems. Letting $(E^i,H^i)$ be an  $\alpha$-quasi-periodic incident field satisfying the homogeneous Maxwell's equations, we consider 
\begin{equation}\label{model}
  \left\{  \begin{array}{ll}
  \curl E = i \ww  H       & \text{in} \ \R^3\backslash \partial \D\,,  \\
  \curl H = - i \ww \ep  E       & \text{in} \ \R^3\backslash \partial \D\,, \\
 \left[\n \t E \right]= \left[\n \t H\right] =  0 &  \text{on} \ \p \D\,, \\ 
(E,H) & \text{satisfies the $\alpha$-quasi-periodic condition \eqref{def:quasi-period}}\,,   \\
(E - E^i,H-H^i) & \text{satisfies the outgoing radiation condition as}\ x_3 \to \pm \infty\,,
\end{array} \right. 
\end{equation}
where the radiation condition is given} by the Rayleigh-Bloch expansion \cite{kirsch2018limiting,bonnet1994guided}:
\begin{align} \label{eq:RBexpansion}
  E(x',x_3) = 
    \left\{
    \begin{aligned}
  & \sum_{q \in \Lambda^*} E_q(h,\alpha) e^{i (\alpha + q) \dd x'} e^{i  \beta_q (x_3 - h)} \,, \q & & \text{for} \ x_3 > h\,, \\
  & \sum_{q \in \Lambda^*} E_q(-h, \alpha) e^{i (\alpha + q) \dd x'} e^{- i \beta_q (x_3 + h)}\,,\q  & & \text{for} \ x_3 < - h\,,
    \end{aligned}
    \right.
\end{align}
where for $q \in \Lambda^*$, 
\begin{align} \label{def:bq}
\beta_q := \sqrt{\ww^2- |\alpha + q|^2}\,,    
\end{align}
and
 $E_q (\pm h, \alpha)$ is the Fourier coefficient of $e^{-i\alpha \dd x'}E$ on $x_3 = \pm h$:
\begin{align} \label{def:coeff}
    E_q(\pm h,\alpha) := \frac{1}{(2 \pi)^2}\int_Y e^{-i\alpha \dd x'} E(x',\pm h) e^{-iq\dd x'} d x'\,.
\end{align}
Here and in what follows, the square root $\sqrt{\dd}$ is chosen with  positive real or positive imaginary parts, and we always assume that for a fixed $\alpha \in \mc{B}$, there holds 
\begin{equation*}
\ww^2 \notin Z(\alpha) := \big\{|\alpha + q|^2\,;\ q \in \Lambda^*\big\}\,,
\end{equation*}
which guarantees that $\beta_q$ is well defined and $\beta_q \neq 0$.

For our analysis of the scattering problem \eqref{model}, let us first introduce the quasi-periodic Green's function by 
    \begin{align} \label{eq:qpgreen}
        G^{\alpha, \ww}(x) 
      & = \frac{1}{2}\sum_{q \in \Lambda^*} i \beta_q^{-1} e^{i(\alpha+ q)\dd x'} e^{i \beta_q |x_3 |}\,,
    \end{align} 
which is the fundamental solution to the equation:
\begin{equation} \label{eq:qpeq}
    -(\Delta + \ww^2) G^{\alpha,\ww}(x) = \sum_{\gamma \in \Lambda}\d(x - (\gamma,0)) e^{i \gamma \dd \alpha}\,,
\end{equation}
satisfying the outgoing condition \eqref{eq:RBexpansion}. Then, we define the vector potentials: 
\begin{align}
    \kd[\varphi] = \int_D G^{\alpha,\ww}(x,y)\vp(y)\, dy:\q {\bf L}^2(D) \rightarrow {\bf H}^2_{\alpha,loc}(Y_\infty)\,, \label{def:vectorpotential}
    \end{align}
and 
\begin{equation} \label{def:electricpotential}
    \td[\varphi] = (\ww^2 + \nabla \ddiv) \kd[\varphi]:\q {\bf L}^2(D) \rightarrow {\bf H}_{\alpha,loc}(\ccurl,Y_\infty)\,.
\end{equation}
By definition, it is easy to check that $\TT^{\alpha,\ww}_{D} [\vp]$ is $\alpha$-quasi-periodic and satisfies  
the equation:
\begin{align} \label{eq:soltd}
    (\curl \curl - \ww^2)\TT^{\alpha,\ww}_{D} [\vp] = \ww^2 \vp \chi_D\q \text{on}\ \R^3\,,
\end{align} 
and the Rayleigh-Bloch expansion \eqref{eq:RBexpansion}. 
Then the Lippmann-Schwinger equation for \eqref{model} readily follows: 
\begin{equation} \label{eq:Lippmann-Schwinger}
    E^s := E - E^i = \tau \td[E]\q \text{on}\ \R^3\,.
\end{equation}

We recall the Helmholtz decomposition for $L^2$-vector fields \cite{amrouche1998vector,girault2012finite}:
\begin{align} \label{eq:helm}
    {\bf L}^2(D) = \na H_0^1(D) \oplus {\bf H}(\ddiv0,D) = \na H_0^1(D) \oplus {\bf H}_0(\ddiv0,D) \oplus W\,,
\end{align}
where $W$ is the space of the gradients of harmonic $H^1$-functions. It is known \cite{costabel2012essential,costabel2010volume,ammari2020mathematical} that \mb{for $\alpha \in \C^2$ and $\ww \in \C$} with $\ww^2 \notin Z(\alpha)$, $\na H_0^1(D)$ and ${\bf H}(\ddiv0,D)$ are invariant subspaces of $\TT_D^{\alpha,\ww}$ with
$
      \TT_D^{\alpha,\ww}[\na \phi] = - \na [\phi \chi_D]
$ for $\phi \in H_0^1(D)$.  Moreover, the spectrum $\sigma(\td)$ is a disjoint union of the essential spectrum $\sigma_{ess}(\td) = \{-1,0,-\frac{1}{2}\}$ and the discrete spectrum $\sigma_{disc}(\td)$. We then define
\begin{align} \label{def:lipschi_operator}
    \A_\tau(\alpha,\ww) := \tau^{-1} - \td: {\bf H}(\ddiv0,D) \to {\bf H}(\ddiv 0, D)\,,
\end{align}
which is an operator-valued analytic function on the set
$\{(\tau,\alpha,\ww)\in \C^4\,;\ |\tau| \gg 1\,,\ \ww^2 \notin Z(\alpha)\}$. By the multidimensional analytic Fredholm theorem \cite{kuchment2012floquet,stessin2011analyticity,MichaelTaylor2010s}, we have the following basic lemma. 

\begin{lemma} \label{lem:spec_td}
The set of poles of the resolvent $(\A_\tau(\alpha,\ww))^{-1}$:
\begin{align} \label{eq:resonances}
   \Omega(\A_\tau(\alpha,\ww)):=  \big\{(\tau,\alpha,\ww)\in \C^4\,;\ |\tau| \gg 1\,,\ \ww^2 \notin Z(\alpha)\,,\ \A_\tau(\alpha,\ww)\ \text{is not invertible}\big\}\,,
\end{align}
is either empty, or an analytic subset of $\C^4$ of codimension $1$ (that is, it is locally given by the zeros of an analytic function in the variables $\tau$, $\alpha$, and $\ww$). 
\end{lemma}

It is straightforward to follow \cite[Section 2]{ammari2020mathematical} to introduce the scattering resolvent for the problem \eqref{model} and show that its poles are the same as those of $\A_\tau(\alpha,\dd)^{-1}$, which we call the \emph{scattering resonances} of the dielectric metasurface. Then by Lemma \ref{lem:spec_td}, \mb{we see that} the slice of the set $\Omega(\A_\tau(\alpha,\ww))$ \eqref{eq:resonances} at $(\alpha,\tau)$, denoted by $\Omega(\A_\tau(\alpha,\ww))|_{(\alpha,\tau)}$, gives the discrete set of resonant frequencies $\ww_j$ for the problem \eqref{model}:
\begin{align*}
   \{\ww_j(\alpha,\tau)\}_{j \ge 0} =  \Omega(\A_\tau(\alpha,\ww))|_{(\alpha,\tau)}\,.
\end{align*}
For each $j$, the resonance $\ww_j(\alpha,\tau)$, as a function in $\alpha \in \mc{B}$, is Lipschitz continuous and known as the band dispersion function \cite{kuchment2016overview}. In this work, we are mainly interested in the resonances in the subwavelength regime, i.e., those $\ww_j \ll 1$. In analogy with the analysis in \cite[Section 3]{ammari2020mathematical}, we have the following asymptotic result for the subwavelength band function $\ww_j(\alpha,\tau)$. Let $\pd$ and $\pw$ be the orthogonal projections for the decomposition \eqref{eq:helm}:
\begin{align} \label{def:proj_helm}
    \pd: {\bf L}^2(D) \to {\bf H}_0(\ddiv0,D)\,,\q \pw: {\bf L}^2(D) \to W\,.
\end{align}
\begin{theorem} \label{thm:asym_to}
In the high contrast regime \eqref{assp:high_constrast}, the subwavelength resonances for the scattering problem \eqref{model} exist with the asymptotic form: as $\tau \to \infty$, 
\begin{align*}
    \ww_j(\alpha,\tau) = \frac{1}{\sqrt{\tau \lad_j^\alpha}} + O(\tau^{-1})\,,
\end{align*}
where the remainder term is complex, and 
$\lad_j^\alpha$ is an eigenvalue of the compact self-adjoint operator $\pd \mc{K}_D^{\alpha,0} \pd$. 
\end{theorem}

\subsection{Real resonances and Bound states} \label{sec:band}
It is clear that for a fixed $\alpha \in \mc{B}$, at each resonance $\ww_j(\alpha,\tau)$, there are nonzero solutions $(E,H)$ to \eqref{model} with $(E^i,H^i)=0$, which are called \emph{Bloch modes}. We will first prove in Proposition \ref{lem:expdecay} that when the resonant frequency
is real, the associated Bloch modes exponentially decay and hence do not couple to the far field. The argument is standard and based on the variational method.  
It follows that in this case, the resonant modes are bounded with respect to the
$L^2$-norm and usually referred to as the \emph{bound states}.

We \mb{start with the introduction} of the EM Dirichlet to Neumann (DtN) operators $\ms{T}$. \mb{We denote}
\begin{align*}
    \Sigma_{\pm h} := Y \t \{\pm h\}  = \{
 x \in \R^3\,;\ x' \in Y\,,\ x_3 = \pm h
 \}\,, \q \text{for}\ h > 0\,.
\end{align*}
and 
\begin{align*}
    \Sigma := \Sigma_h \cup \Sigma_{-h}\,.
\end{align*}

\begin{definition} \label{def:dtn_em}
Let $\vp \in{\bf H}_{\alpha,T}^{-1/2}(\ccurl,\Sigma_{\pm h})$ be an $\alpha$-quasi-periodic tangent field on $\Sigma_{\pm h}$ with the Fourier expansion: 
\begin{equation*}
    \vp(x',\pm h) := \sum_{q \in \Lambda^*} (\vp_q, 0)  e^{i(\alpha + q) \dd x'}\,, \q \vp_q := (\vp_{1,q},\vp_{2,q}) \in \R^2 \,.
 \end{equation*}
The DtN operator $\mathscr{T}:{\bf H}_{\alpha,T}^{-1/2}(\ccurl,\Sigma_{\pm h}) \to {\bf H}_{\alpha,T}^{-1/2}(\ddiv, \Sigma_{\pm h})
$ is defined by 
\begin{equation}\label{def:dtn_exp}
    \ms{T} \vp = \sum_{q \in \Lambda^*} \big((\ms{T} \vp)_q, 0 \big) e^{i(\alpha + q)\dd x'}\,,
\end{equation}
where the coefficients $(\ms{T} \vp)_q \in \R^2$, $q \in \Lambda^*$, are given by
\begin{align*}
(\ms{T} \vp)_q := \frac{1}{\omega\beta_q} \big\{\beta_q^2 \vp_q + ((\alpha + q)\dd \vp_q)(\alpha + q)\big\}\,.
\end{align*}
\end{definition}

With the help of the operator $\ms{T}$, one can readily obtain the variational formulation for the resonance problem: find $(\ww, E) \in \C \t  {\bf H}_\alpha(\ccurl, Y_h)$ with $\ww^2 \notin Z(\alpha)$ such that 
\begin{align} \label{eq:desa}
    B_{\alpha,\tau,\ww} (\vp, E) := (\ccurl \vp, \ccurl E)_{Y_h} - \ww^2 (\vp, \ep E)_{Y_h} - i \ww \l \pi_t(\vp),\ms{T} \pi_t(E)\r_{\Sigma} = 0\,,\q \forall \vp \in {\bf H}_\alpha(\ccurl,Y_h)\,,
\end{align}
where $\l\dd,\dd \r_\Sigma = \l \dd, \dd \r_{\Sigma_h} + \l \dd ,\dd \r_{\Sigma_{-h}}$, and the trace $\pi_t(\dd)$ is given in \eqref{eq:trace_tang}. We note from Definition \ref{def:dtn_em} that 
\begin{align}  \label{eq:cap_1}
    \l \vp, \ms{T} \psi\r_{\Sigma_{\pm h}} = \sum_{q \in \Lambda^*} \frac{1}{\ww \beta_q} \big\{ \beta_q^2 \psi_q \dd \overline{\vp_q} + ((\alpha + q)\dd \psi_q)(\alpha + q)\dd \overline{\vp_q} \big\}\,, \q \forall \psi, \vp \in{\bf H}_{\alpha,T}^{-1/2}(\ccurl,\Sigma_{\pm h})\,,
\end{align}
which implies
\begin{align} \label{eq:cap_2}
    \l\psi, \ms{T}_{\pm} \psi\r_{\Sigma_{\pm h}} 
    = &  \sum_{q \in \Lambda^*} \frac{1}{\ww \beta_q} \big\{ \beta_q^2 |\psi_q|^2 + |(\alpha + q) \dd \psi_q|^2  \big\}\,.
\end{align}
\mb{Given $\alpha \in \mc{B}$ and $\ww \in \R$,} we introduce the following subsets of the reciprocal lattice $\Lad^*$:
\begin{align} \label{eq:lattice_decp}
    \Lad^*_{\mc{P}} := \big\{q\in \Lad^*\,;\ \ww^2 > |\alpha + q|^2\big\}\,,\q   \Lad^*_{\mc{E}} = \big\{q\in \Lad^*\,;\ \ww^2 < |\alpha + q|^2\big\}\,,
\end{align}
and define the associated operators $\ms{T}^{\mc{P}}$ and $\ms{T}^{\mc{E}}$ as in \eqref{def:dtn_exp} but with the sums over $\Lad^*_{\mc{P}}$ and $ \Lad^*_{\mc{E}}$, respectively. By definition, we have that $\Lad^*_{\mc{P}}$ is a finite set with $\beta_q > 0$ for $q \in \Lad^*_{\mc{P}}$,  while for $q \in \Lad^*_{\mc{E}}$, $\beta_q$ is purely imaginary. Then, recalling the Rayleigh-Bloch expansion, it is clear that $\Lad^*_{\mc{P}}$ corresponds to the propagating modes, and $\Lad^*_{\mc{E}}$ is for the exponentially decaying modes. Note that equations $\ww^2 = |\alpha + q|^2$ with $q \in \Lad^*$ separate the $(\ww,\alpha)$ plane into regions where the wave has different radiation behaviors as $|x_3| \to \infty$. For $\alpha \in \mc{B}$, by definition \eqref{def:bri_zone}, we have that the minimal $|\alpha + q|$ is given by $|\alpha|$, which is the so-called \emph{light line} \cite{bulgakov2017bound,bulgakov2017propagating}, namely,
\begin{align} \label{eq:alphaetc}
|\alpha| = \inf_{q \in \Lambda^*}|\alpha + q|\,.
\end{align}
The following lemma is a simple consequence of \eqref{eq:cap_1} and \eqref{eq:cap_2}. 

\begin{lemma} \label{lem:capacity}
   Let $\ww \neq 0 \in \R$ and $\alpha \in \mc{B}$ with $\ww^2 \notin Z(\alpha)$.   For $\psi, \vp \in{\bf H}_{\alpha,T}^{-1/2}(\ccurl,\Sigma_{\pm h})$, it holds that 
\begin{align*}
    \l\vp, -i \ms{T}^{\mc{E}} \psi\r_{\Sigma_{\pm h}} = \l-i \ms{T}^{\mc{E}} \vp, \psi\r_{\Sigma_{\pm h}}\,,\q \l\ms{T}^{\mc{P}} \vp, \psi\r_{\Sigma_{\pm h}} =  \l\vp,  \ms{T}^{\mc{P}} \psi\r_{\Sigma_{\pm h}}\,.
    \end{align*}
Moreover, we have 
\begin{align*}
     \l\psi, \ms{T}^{\mc{P}} \psi\r_{\Sigma_{\pm h}} \ge 0\,.
\end{align*}
\end{lemma}

\begin{proposition} \label{lem:expdecay}
Let $\ww$ be a resonant frequency for $\alpha \in \mc{B}$ with $\ww^2 \notin Z(\alpha)$, and $E \in {\bf H}_{\alpha,loc}(\ccurl, Y_\infty)$ be the associated Bloch mode with the expansion \eqref{eq:RBexpansion}. Then $\ww$ is real if and only if $E$ exponentially decays as $x_3 \to \infty$. In this case, we have $E \in {\bf H}_{\alpha}(\ccurl, Y_\infty)$.
\end{proposition}

\begin{proof}
Suppose that $(\ww,E)$ solves the variational problem \eqref{eq:desa}. It implies
\begin{align*}
    \im\, B_{\alpha,\tau,\ww} (E, E) =  - \im \, i \ww \l \pi_t(E),\ms{T} \pi_t(E)\r_{\Sigma} = - \re \,\ww \l \pi_t(E),\ms{T} \pi_t(E)\r_{\Sigma}  = 0\,.
\end{align*}
We hence have $ \re \l \pi_t(E),\ms{T} \pi_t(E)\r_{\Sigma} = \l \pi_t(E),\ms{T}^{\mc{P}} \pi_t(E)\r_{\Sigma} = 0$ by formulas \eqref{def:bq} and \eqref{eq:cap_2} and the definition of $\ms{T}^{\mc{P}}$. 
Then it follows from Lemma \ref{lem:capacity} that the coefficients $(\pi_t(E))_q$ for $q \in \Lambda_\mc{P}^*$ on $\Sigma_{\pm h}$ vanish. Recalling the Rayleigh-Bloch expansion \eqref{eq:RBexpansion} and $\ddiv E = 0$ on $Y_\infty \backslash Y_h$, we obtain the exponential decay of the resonant mode $E$ when $x_3 \to \infty$. Conversely, suppose that the Bloch mode $E$ exponentially decays in the far field. Then, we have that for any suitably large $h$, the restriction of $E$ to $Y_h$ satisfies \eqref{eq:desa} with $\l \pi_t(E),\ms{T} \pi_t(E)\r_{\Sigma} \to 0$ as $h \to \infty$, by formulas \eqref{def:coeff} and \eqref{eq:cap_2}, and the exponential decay of $E$. Thus, there holds
$ (\ccurl E, \ccurl E)_{Y_\infty} - \ww^2(E, \ep E)_{Y_\infty} = 0
$, which readily gives $E \in {\bf H}_\alpha(\ccurl, Y_\infty)$ and $\ww \in \R$. 
\end{proof}

We proceed to prove that the real resonances are the simple poles of the scattering resolvent, which means that for a real resonance, its algebraic multiplicity equals to its geometric multiplicity \cite{dyatlov2019mathematical,gohberg1990classes,gokhberg1971operator}. For the reader's convenience, \mb{let us first recall some related concepts}, following \cite{gokhberg1971operator}. Suppose that $\A(\lad)$ is an operator-valued analytic function from a neighborhood $N_{\ep}(\lad_0)$ of $\lad_0 \in \C$ to $\mc{L}(\mc{H})$, where $\mc{H}$ is a Banach space. $\lad_0$ is a pole of $\A(\lad)^{-1}$ means that there exists an analytic function $\phi(\lad): N_{\ep}(\lad_0) \to \mc{H}$, called a root function for $\A(\lad)$ at $\lad_0$, such that $\phi(\lad_0) \neq 0$ and $\A(\lad_0) \phi(\lad_0) = 0$. It is clear that $\phi(\lad_0)$ is an eigenvector of $\A(\lad_0)$, and we define the geometric multiplicity of $\lad_0$ by $\dim \ker(\A(\lad_0))$. By analyticity, we have $\A(\lad)\phi(\lad) = (\lad - \lad_0)^m \psi(\lad)$ for some $m > 0$ with $\psi(\lad_0) \neq 0$ and call the number $m$ the multiplicity of the root function $\phi(\lad)$. We define the rank of an eigenvector $\phi_0$ by ${\rm rank}(\phi_0) := \{m(\phi)\,;\ \phi(\lad)\ \text{is a root function with}\ \phi(\lad_0) = \phi_0\}$. Let $\{\phi_0^j\}$ be an orthogonal set of eigenvectors that span $\ker(\A(\lad_0))$ with the property: ${\rm rank}(\phi_0^j)$ is the maximum of the ranks of all $\phi$ in the orthogonal complement of ${\rm span}\{\phi_0^1\,\ldots, \phi_0^{j-1}\} $ in $\ker (\A(\lad_0))$. We then define the algebraic multiplicity of $\lad_0$ by $\sum_{j} {\rm rank}(\phi_0^j)$. 

\begin{proposition} \label{prop:real}
For a fixed $\alpha \in \mc{B}$, the real resonances are the simple poles of $\mc{A}_\tau(\alpha,\ww)^{-1}$.
\end{proposition}

\begin{proof}
Let $\ww_0$ be a real resonance, i.e., a real pole of $\mc{A}_\tau(\alpha,\ww)^{-1}$. By definition, it suffices to prove that for any eigenfunction $E_0 \in \ker(\mc{A}_\tau(\alpha,\ww_0))$, its rank 
equals to one. For this, without loss of generality, we assume that $E_0$ is well defined on $Y_\infty$ and exponentially decays as $x_3 \to \infty$, by considering $\td[E_0]$ and Proposition \ref{lem:expdecay}. Let $E(\ww)$ be an associated root function with $E_0 = E(\ww_0)$. We  write 
\begin{align} \label{auxeq_simple_pole_2}
\A_\tau(\alpha,\ww)[E(\ww)] = (\ww - \ww_0) \psi(\ww)\,.
\end{align}
We will show that $\psi(\ww_0) \neq 0$, which gives ${\rm rank}(E_0) = 1$ and thus completes the proof. Noting that $E_0 \in {\bf H}_\alpha(\ccurl, Y_\infty)$ exponentially decays, by integration by parts, it holds that, for $\ww \in \C$ near $\ww_0$,
\begin{align*}
    0 &= (\td[\vp], \mb{\curl \curl E_0 - \ww_0^2 \ep E_0})_{Y_\infty} \\
    & = (\curl \curl T_D^{\alpha,\ww} [\vp], E_0)_{Y_\infty} -  \ww_0^2 (\TT_D^{\alpha,\ww}[\vp], \ep E_0)_{Y_\infty} \\
& = (\ww^2 \td[\vp] + \ww^2 \vp \chi_D, E_0)_{Y_\infty} - \ww_0^2 (\TT_D^{\alpha,\ww}[\vp], (1 + \tau \chi_D) E_0)_{Y_\infty}\,,\q \forall \vp \in {\bf H}(\ddiv0,D)\,,
\end{align*}
where the last step follows from \eqref{eq:soltd}. Then we obtain 
\begin{align} \label{auxeq_simpole}
    (\ww^2 \vp \chi_D, E_0)_{Y_\infty} - \tau \ww_0^2(\td[\vp], \chi_D E_0)_{Y_\infty} = ((\ww_0^2 - \ww^2)\td[\vp], E_0)_{Y_\infty}\,.
\end{align}
Taking $\vp = E(\ww)$ in the above formula \eqref{auxeq_simpole} and using \eqref{auxeq_simple_pole_2} gives 
\begin{align*}
    ((\ww^2 - \ww_0^2) E(\ww), E_0)_{D} +  \tau \ww_0^2((\ww - \ww_0) \psi(\ww),  E_0)_{D} = ((\ww_0^2 - \ww^2)\td[E(\ww)], E_0)_{Y_\infty}\,,
\end{align*}
which further implies, by eliminating the factor $\ww - \ww_0$ and letting $\ww = \ww_0$,  
\begin{align} \label{auxeq_sim_3}
    2 \ww_0  (E(\ww_0), E_0)_{D} +  \tau \ww_0^2( \psi(\ww_0),  E_0)_{D} & = - 2 \ww_0 (\mc{T}_D^{\alpha,\ww_0}[E(\ww_0)], E_0)_{Y_\infty} \notag \\
    & = - 2 \tau^{-1} \ww_0 (E(\ww_0), E_0)_{Y_\infty}\,.
\end{align}
If $\psi(\ww_0) = 0$, we obtain from \eqref{auxeq_sim_3} that 
$
    \tau \norm{E_0}_D^2 = - \norm{E_0}_{Y_\infty}^2
$. It follows that $E_0 = 0$, which is a contradiction. Hence, $\psi(\ww_0) \neq 0$ and the proof is complete. 
\end{proof}

An immediate and important follow-up question is the existence of real resonances, which will be the main focus of the remaining section and Section \ref{sec:exist_real}. We next prove that there always exist real resonances below the light line $|\alpha|$, by extending Theorem 4.1 and Theorem 4.4 of \cite{bonnet1994guided} for the scalar case; see Theorem \ref{thm:ess_discre}. 
For this purpose, we define the Maxwell operator by $\M_\ep(\alpha): = \ep^{-1}(\curl)^2$ with the domain $\dom(\M_{\ep}(\alpha))$ given by the space $ {\bf C}^\infty_{c,\alpha}(Y_\infty)$ of quasi-periodic smooth fields with compact support, 
which is a densely defined positive unbounded operator on ${\bf L}^2_{\alpha,\ep}(Y_\infty)$. We claim that by Friedrichs extension \cite{simon1975methods}, the operator $\M_\ep$ admits a self-adjoint extension. Indeed, we consider the following sesquilinear unbounded form on ${\bf L}^2_{\alpha,\ep}(Y_\infty)$:
\begin{align*}
    a_{\alpha,\ep}(\vp, E)  := & \, (\vp,\M_\ep(\alpha)E)_{\ep, Y_\infty} + (\vp,E)_{\ep, Y_\infty}\\
    = &\,  (\curl \vp, \curl E)_{Y_\infty} + (\vp, \ep E)_{Y_\infty}\,,\q  \forall E, \vp \in {\bf C}^\infty_{c,\alpha}(Y_\infty)\,.
\end{align*}
Clearly, $a_{\alpha,\ep}(\dd,\dd)$ defines an inner product on $\dom (\M_\ep(\alpha))$, and 
 the completed space $\overline{\dom(\M_\ep(\alpha))}^{a_{\alpha,\ep}}$ is nothing else than ${\bf H}_\alpha(\ccurl, Y_\infty)$, and hence $a_{\alpha,\ep}(\dd,\dd)$ is ${\bf H}_\alpha(\ccurl, Y_\infty)$-elliptic. Then, by Riesz representation, we define the operator
$\mc{N}_{\alpha,\ep}:  \dom(\mc{N}_{\alpha,\ep}) \to {\bf L}_{\alpha,\ep}^2(Y_\infty) 
$ satisfying that for $E \in \dom(\mc{N}_{\alpha,\ep})$ and $\vp \in {\bf H}_\alpha(\ccurl,Y_\infty)$, 
\begin{align} \label{def:Nep} 
a_{\alpha,\ep}(\vp, E) = \big(\vp, \mc{N}_{\alpha,\ep}(E)\big)_{\ep, Y_\infty}\,,
\end{align}
where the domain of $\mc{N}_{\alpha,\ep}$ is given by
\begin{align*}
    \dom(\mc{N}_{\alpha,\ep}): = \big\{E \in {\bf H}_\alpha(\ccurl,Y_\infty)\,;\ \vp \to a_{\alpha,\ep}(\vp, E)\ \text{is continuous on}\ {\bf H}_\alpha(\ccurl, Y_\infty)\ \text{w.r.t.}\,{\bf L}^2_{\alpha}\text{-topology}\big\}\,.
\end{align*}
Note that $(\curl \vp, \curl E)_{Y_\infty}$ is continuous in $\vp \in {\bf H}_\alpha(\ccurl, Y_\infty)$ w.r.t.\,${\bf L}^2_{\alpha}$-topology if and only if there exists $\psi \in {\bf L}^2_\alpha(Y_\infty)$ such that $(\curl \vp, \curl E)_{Y_\infty} = (\vp, \psi)_{Y_\infty}$. We hence have
\begin{align} \label{def:mep}
  \dom(\mc{N}_{\alpha,\ep}) = \big\{E \in {\bf L}^2_\alpha(Y_\infty)\,;\ \ccurl\,E\in {\bf L}^2_\alpha(Y_\infty)\,,\  \ccurl\ccurl\,E \in {\bf L}^2_\alpha(Y_\infty) \big\}\,,
\end{align}
and the Friedrichs extension allows us to conclude the following lemma.

\begin{lemma} \label{lem:symmaxwell}
The operator 
$\mc{N}_{\alpha,\ep}: \dom(\mc{N}_{\alpha,\ep}) \to {\bf L}^2_{\alpha,\ep}(Y_\infty)$ defined in \eqref{def:Nep} is 
bijective, positive, and self-adjoint on ${\bf L}^2_{\alpha,\ep}(Y_\infty)$ with $\mc{N}_{\alpha,\ep}^{-1} \in \mc{L}({\bf L}^2_{\alpha,\ep}(Y_\infty))$. In particular, the positive operator $\mc{N}_{\alpha,\ep}-1$ is the self-adjoint extension of the Maxwell operator $\M_\ep(\alpha)$.
\end{lemma}

Recalling Proposition \ref{lem:expdecay} above, it is easy to see that $\ww \neq 0$ is a real resonance 
if and only if $\ww^2$ is an eigenvalue of $\M_\ep(\alpha)$ with eigenfunctions in ${\bf H}_\alpha(\ep; \ddiv 0, Y_\infty)$ (cf.\,\eqref{def:diverzerospace}). Therefore, for the existence of real resonances, it suffices to investigate the spectral properties of the self-adjoint operator $\M_\ep(\alpha)$ on ${\bf H}_\alpha(\ep; \ddiv 0, Y_\infty)$:
\begin{align} \label{def:res_maxwell}
 \M_\ep(\alpha):\dom(\M_\ep(\alpha)) \to {\bf H}_\alpha(\ep; \ddiv 0, Y_\infty)\,,\q \text{with}\ \dom(\M_\ep(\alpha)) := \dom(\mc{N}_{\alpha,\ep})\cap {\bf H}_\alpha(\ep; \ddiv 0, Y_\infty)\,.
\end{align}
We define, for $n \ge 1$,  
\begin{align} \label{minmaxeig}
    \mu_n(\tau) := &\sup_{\{\vp_i\}_{i = 1}^{n-1} \subset \dom(\mc{M}_\ep(\alpha))}\inf_{\substack{\vp\in \dom(\M_\ep(\alpha))\\ \norm{\vp}_{\ep,Y_\infty} = 1\,,\   \vp \perp \vp_i}} ( \vp, \M_\ep(\alpha) \vp)_{\ep, Y_\infty} \notag \\
    = & \sup_{\vp_i} \inf_\vp \Big\{\frac{( \curl \vp, \curl \vp)_{Y_\infty}}{( (1 + \tau \chi_D) \vp, \vp)_{Y_\infty}} \,;\ \vp,\, \vp_i \in  \dom(\mc{M}_\ep(\alpha))\,, \ \vp \perp \vp_i\,,\   1 \le i \le n - 1  \Big\}\,,
\end{align}
which is a decreasing function in $\tau > 0$. By the min-max principle for self-adjoint operators \cite[Theorem XIII.1]{simon1978methods}, we have that for each $n$, either $\mu_n = \inf\{\lad\,;\ \lad \in \sigma_{ess}(\mc{M}_\ep(\alpha))\}$ and in this case $\mu_m = \mu_n$ for any $m \ge n$, or $\mu_n$ is the $n$th eigenvalue of $\mc{M}_\ep(\alpha)$ counting multiplicity with $\mu_n < \inf\{\lad\,;\ \lad \in \sigma_{ess}(\mc{M}_\ep(\alpha))\}$. 

In Theorem \ref{thm:ess_discre} below, we characterize the essential spectrum of the Maxwell operator
$\sigma_{ess}(\M_\ep(\alpha))$ and show the existence of discrete eigenvalues of $\M_\ep(\alpha)$ below $|\alpha|^2$, equivalently, the real resonances exist below the light line $|\alpha|$. For ease of exposition, we recall the Weyl's criterion for the essential spectrum \cite[Theorem 7.2]{hislop2012introduction}: $\lad \in \sigma_{ess}(\M_\ep(\alpha))$ if and only if there exists a \emph{Weyl sequence} $\{\vp_k\}_{k \ge 1} \subset \dom(\M_\ep(\alpha))$ associated with $\lad$, which satisfies $\norm{\vp_k}_{\ep, Y_\infty} = 1$ and that when $k \to \infty$, $\vp_k \to 0$ weakly and $\norm{(\M_\ep(\alpha) - \lad) \vp_k}_{\ep, Y_\infty} \to 0$.

\begin{theorem} \label{thm:ess_discre} Let $\mc{M}_{\ep}(\alpha)$, $\alpha \in \mc{B}$,  be the self-adjoint positive operator defined in \eqref{def:res_maxwell}. Then it holds that 
 \begin{enumerate}
     \item $\sigma_{ess}(\M_\ep(\alpha)) = [|\alpha|^2,+\infty)$.
     \item $\mu_1(\tau) < |\alpha|^2$ for any $\tau > 0$ and $\alpha \neq 0$, that is, 
     there exists a discrete eigenvalue of $\M_\ep(\alpha)$ below $|\alpha|^2$.
 \end{enumerate}
\end{theorem}

\begin{proof} For the first statement, let $\lad \ge |\alpha|^2$ and $\psi$ be a $C^\infty$-compactly supported function on $Y_\infty$  satisfying $\psi(x_3) = 0$ for $|x_3| \le h$ and $\int_\R \psi(x_3)^2 \, d x_3 = 1$, where $h$ is suitably large such that $D$ is included in $Y_h$. We consider the following sequence of vector fields:
\begin{align*}
   \w{E}_n = \frac{1}{\sqrt{n}} \psi\big(\frac{x_3}{n}\big) p e^{id  \dd x}  \q \text{with} \ p, d = (\alpha,\sqrt{\lad - |\alpha|^2}) \in \R^3,\ |p| = 1 \ \text{and} \ p \dd d = 0\,. 
\end{align*}
It is easy to see that $\norm{\w{E}_n}_{Y_\infty} = 1$ and $\w{E}_n$ weakly converges to zero in ${\bf L}^2_{\alpha,\ep}(Y_\infty)$. Moreover, by a direct computation and vector calculus identities, we find 
\begin{align*}
      (\mc{M}_\ep - \lad)[\w{E}_n] & = \na (\ddiv \w{E}_n) - \Delta \w{E}_n - \lad \w{E}_n \\
      & = \na \Big(\frac{1}{n \sqrt{n}} \psi'\big(\frac{x_3}{n}\big)p_3 e^{i d \dd x} \Big) - \frac{1}{\sqrt{n}}p \Big(\frac{1}{n^2} \psi''\big(\frac{x_3}{n}\big)e^{i d \dd x} + 2 \frac{i d_3}{n}\psi'\big(\frac{x_3}{n}\big) e^{i d \dd x} \Big)\,,
\end{align*}
which tends to zero as $n \to \infty$. To obtain a Weyl sequence, we add to $\w{E}_n$ a corrective gradient field $\na p_n \in {\bf L}^2_\alpha(Y_\infty)$ such that $\ddiv(\w{E}_n + \ep \na p_n) = 0$, which is equivalent to solving the elliptic equation $- \ddiv(\ep \na p_n) = \ddiv \w{E}_n$. By Lemma \ref{lem:app1}, it is uniquely solvable in the weighted Sobolev space $H_\alpha^{1,-1}(Y_\infty)/\R$ and there holds 
\begin{align*}
\norm{\na p_n}_{Y_\infty} \lesssim \norm{\ddiv \w{E}_n}_{Y_\infty} = \frac{|p_3|}{n} \big(\int_\R |\psi'|^2 \ dx\big)^{1/2} \to 0\,, \q \text{as}\ n \to \infty\,.  
\end{align*}
We then define $E_n := (\w{E}_n + \na p_n)/\norm{\w{E}_n + \na p_n}_{Y_\infty}$. It is straightforward to check that $\{E_n\}\subset \dom(\mc{M}_\ep(\alpha))$ is a Weyl sequence for $\lad$, which implies $[|\alpha|^2, + \infty) \subset \si_{ess}(\mc{M}_\ep(\alpha))$.

To show the reverse direction, we will prove the following two claims: 
\begin{align} \label{first_sta_2}
\sigma_{ess}(\M) \subset [|\alpha|^2, + \infty) \q \text{and} \q \sigma_{ess}(\M_\ep) \subset \sigma_{ess}(\M)\,,    
\end{align}
where $\mc{M}$ is defined as $\mc{M}_\ep$ with $\ep = 1$ (i.e., $\tau = 0$). For the first one, suppose that $\{E_n\}$ is a Weyl sequence associated with $\lad \in \si_{ess}(\mc{M})$, which readily implies
\begin{align} \label{essen_auxeq_1}
    ((\mc{M} - \lad)[E_n], E_n)_{Y_\infty} = \norm{\curl E_n}^2_{Y_\infty} - \lad \to 0\,, \q \text{as} \ n \to \infty\,.
\end{align}
Note that $\ddiv E_n = 0$ and $E_n$, as a quasi-periodic function, admits the expansion: $E_n(x) = \sum_{q \in \Lambda^*} E_{n,q}(x_3) e^{i(\alpha + q)\dd x}$. We can estimate, thanks to \eqref{eq:alphaetc}, 
\begin{align*}
    \int_{Y_\infty} |\curl E_n|^2 + |\ddiv E_n|^2\ dx = \int_{Y_\infty} |\na E_n|^2\ dx \ge |\alpha|^2\,, 
\end{align*}
which, by \eqref{essen_auxeq_1}, gives $\lad \ge |\alpha|^2$.

We next show the second claim. Letting $\lad \in \si_{ess}(\mc{M}_\ep)$, it suffices to construct a Weyl sequence associated with $\lad$ but for the operator $\mc{M}$. By abuse of notation, suppose that $\{\w{E}_n\}$ is a Weyl sequence corresponding to $\lad \in \si_{ess}(\mc{M}_\ep)$. 
We define a smooth cutoff function $\theta$ on $\R$ such that $\theta(x_3) = 0$ on $(-h,h)$ and $\theta(x_3) = 1$ on $\R\backslash (-2h,2h)$. It is clear that the sequence $ \theta \w{E}_n$ weakly converges to zero in ${\bf L}^2_\alpha(Y_\infty)$ and so does $(1 - \theta) \w{E}_n$. Noting that $1 - \theta$ is a $C^\infty$ function compactly supported in $Y_{2h}$, we can check that $(1 - \theta) \w{E}_n$ is a bounded sequence in the space:
\begin{align*}
    V = \{\psi \in {\bf L}^2_\alpha(Y_{2h})\,;\ \ccurl \psi \in {\bf L}^2_\alpha(Y_{2h}),\ \ddiv(\ep\psi) \in {\bf L}^2_\alpha(Y_{2h}),\  \n \t \psi = 0\ \text{on}\ \Sigma_{\pm 2h}\}\,,
\end{align*}
which is compactly embedded into the space ${\bf L}^2_\alpha(Y_{2h})$, by the classical result of Weber \cite{weber1980local} (see also \cite[Appendix B]{hazard1996solution}). As a consequence, we have, by the weak convergence of $(1 - \theta) \w{E}_n$,
\begin{align*}
\norm{(1 - \theta) \w{E}_n}_{\ep, Y_\infty} \to 0\,,\q \text{as}\ n \to \infty\,.  
\end{align*}
It follows from $\norm{\w{E}_n}_{\ep, Y_\infty} = 1$ that $\norm{\theta \w{E}_n}_{\ep,Y_\infty} \to 1$. We next check 
\begin{align} \label{midclaim}
    (\M - \lad) [\theta\w{E}_n] \to 0\,, \q \text{as}\ n \to \infty\,.
\end{align}
Indeed, $\p_j \theta$, for $j = 1,2,3$, is $C^\infty$-smooth on $\R$ and compactly supported in the closure of $Y_{2h} \backslash Y_h$. Then, a direct computation gives
\begin{align} \label{eqq1}
\curl \curl (\theta \w{E}_n)  & = \curl (\na \theta \t \w{E}_n + \theta \curl \w{E}_n) \notag \\
& = \na \theta \ddiv \w{E}_n - \w{E}_n \Delta \theta + (\w{E}_n \dd \na)(\na \theta) - (\na \theta\dd \na )\w{E}_n + \na \theta \t (\curl \w{E}_n) + \theta \curl \curl \w{E}_n\,.
\end{align}
Since $\w{E}_n$ is a Weyl sequence for $\mc{M}_\ep$, we have $\na \theta \ddiv \w{E}_n = 0$ and as $n \to \infty$, 
\begin{align} \label{eqq2}
     \norm{\theta \curl \curl \w{E}_n - \lad \theta \w{E}_n}_{Y_\infty} = \norm{\theta (\mc{M}_\ep - \lad)[\w{E}_n]}_{Y_\infty} \to 0\,.
\end{align}
Moreover, by the standard interior regularity for Maxwell's equation, we see that $\w{E}_n$ is bounded in ${\bf H}_\alpha^2(Y_{2h}\backslash Y_h)$. Therefore, it follows that
\begin{align} \label{auxeq_ess_2}
     \norm{- \w{E}_n \Delta \theta + (\w{E}_n \dd \na)(\na \theta) - (\na \theta\dd \na )\w{E}_n + \na \theta \t (\curl \w{E}_n)}_{Y_\infty} \to 0\,, \q \text{as}\ n \to \infty\,, 
\end{align}
thanks to the compact embedding from ${\bf H}_\alpha^1(Y_{2h}\backslash Y_h)$ to ${\bf L}_\alpha^2(Y_{2h}\backslash Y_h)$. Collecting \eqref{eqq1}--\eqref{auxeq_ess_2} gives \eqref{midclaim}. 
Again, to obtain the desired Weyl sequence, we need to consider a corrective $L^2$-field $\na p_n$ defined by the equation $$ - \Delta p_n = \ddiv (\theta \w{E}_n)\,.$$ We compute $\ddiv (\theta \w{E}_n) = \na \theta \dd \w{E}_n$ and find $\norm{\ddiv (\theta \w{E}_n)}_{Y_\infty} \to 0$, similarly to \eqref{auxeq_ess_2}, which, by Lemma \ref{lem:app1}, also implies $\norm{\na p_n}_{Y_\infty} \to 0$. Then, it is easy to see that $\lad \in \sigma_{ess}(\mc{M})$ and
$E_n := (\theta\w{E}_n + \na p_n)/\norm{\theta\w{E}_n + \na p_n}_{Y_\infty}$ is an associated Weyl sequence. We hence have proved \eqref{first_sta_2} and showed $\si_{ess}(\mc{M}_\ep(\alpha)) = [|\alpha|^2, + \infty)$.

For the second statement, we consider $\w{E} := f(x_3) p e^{i\alpha\dd x}$ with $p := (\alpha_2,-\alpha_1,0)$ and 
\begin{align*}
  f(x_3) := 
\begin{cases}
      1 \,, \q &  \text{if} \ |x_3| < h\,, \\
 2 - h^{-1}|x_3|\,,\q  &  \text{if} \ h < |x_3| < 2 h\,,\\
   0 \,, \q &  \text{if} \ |x_3| > 2 h\,,
\end{cases}
\end{align*}
where $h > 0$ is large enough so that $D \subset Y_h$.  Clearly, there holds $\ddiv\w{E} = 0$. We then define $p \in H_\alpha^{1,-1}(Y_\infty)/\R$, for a given $\ep$ of the form \eqref{assp:high_constrast}, by $\ddiv(\ep(\na p + \w{E})) = 0$, and let $E: = \w{E} + \na p$, which belongs to the domain $\dom(\mc{M}_\ep(\alpha))$. We note that $\ddiv(\ep \w{E}) = \tau \ddiv(\chi_D \w{E})$ is independent of $h$, hence so is the function $p$. 
We proceed to compute 
\begin{align*}
    \curl \w{E} = (f'(x_3) \alpha_1, f'(x_3) \alpha_2, - i |\alpha|^2 f(x_3)) e^{i \alpha \dd x}
\end{align*}
and 
\begin{align*}
    \frac{(\curl E, \curl E)_{Y_\infty}}{(\ep E, E)_{Y_\infty}} &=  \frac{(\curl \w{E}, \curl \w{E})_{Y_\infty}}{(\ep E, E)_{Y_\infty}} =  \frac{(\curl \w{E}, \curl \w{E})_{Y_\infty} - |\alpha|^2 (\ep E, E)_{Y_\infty}}{(\ep E, E)_{Y_\infty}} + |\alpha|^2\,,
\end{align*}
which yields 
\begin{align} 
    &(\curl \w{E}, \curl \w{E})_{Y_\infty} - |\alpha|^2 (\ep E, E)_{Y_\infty}  \notag \\
    = & \int_\R |\alpha|^2(f'(x_3)^2 + |\alpha|^2 f(x_3)^2) \ d x_3 - |\alpha|^2 \norm{E}^2_{\ep, Y_\infty} \notag \\
    = & 2 |\alpha|^2 h^{-1} - |\alpha|^2 \big(\norm{E}^2_{\ep, Y_\infty} - \norm{\w{E}}^2_{Y_\infty}\big)\,.
    \label{eq:aux1}
\end{align}
By Lemma \ref{auxlem1}, the difference $ \norm{E}^2_{ \ep, Y_\infty} - \norm{\w{E}}^2_{Y_\infty}$ is strictly positive and independent of $h$, and
\begin{align*}
    \norm{E}^2_{ \ep, Y_\infty} - \norm{\w{E}}^2_{Y_\infty} \ge \norm{\na p}_{Y_\infty} \norm{\w{E}}_D\,.
\end{align*}
Therefore, for large enough $h$, we conclude from \eqref{eq:aux1} that 
\begin{align*}
    (\curl E, \curl E)_{Y_\infty} - |\alpha|^2 (\ep E, E)_{Y_\infty} < 0\,,
\end{align*}
which, by the min-max representation \eqref{minmaxeig}, means $\mu_1(\tau) < |\alpha|^2$. 
\end{proof}

\subsection{Symmetry of band functions and multiplicity} \label{sec:sym_band_multi}

In this section, we consider dielectric metasurfaces with certain symmetries, and derive some consequences on the symmetry properties of 
the band functions $\ww_j(\alpha)$ and Bloch modes, and the multiplicity of scattering resonances. These results will be useful in Section \ref{sec:exist_real}. 

We first note from \eqref{assp:high_constrast} and \eqref{model} that the symmetry of the scattering problem is completely determined by the symmetry group of the metasurface $\D$ (i.e., the set of isometries that carry $\D$ to itself). Here, an isometry of $\R^d$ is a distance-preserving map from $\R^d$ to itself, which can be uniquely written as the composition of a translation and an orthogonal linear map; see \cite[Corollary 6.2.7]{artin2011algebra}. We denote by $M_d$ the set of all isometries of $\R^d$. Since this work only considers a single layer of dielectric nanoparticles, for the sake of simplicity, we will not discuss the symmetry of $\D$ along the $z$-direction. Thus, the symmetry group of $\D$ reduces to the well-investigated plane symmetry group. To be precise, let $g$ be a two-dimensional isometry. By abuse of notation, we still denote by 
$g$ its three-dimensional extension: $g x = (g x',x_3)$ for $x \in \R^3$. Then we
 define the symmetry group of $\D$ as follows:
\begin{align*}
  \mc{G} := \{g\,;\  g\  \text{is a two-dimensional isometry and}\ g(\D) = \D\}\,, 
\end{align*}
which is a discrete group, according to \cite[Definition 6.5.1]{artin2011algebra} and the structure \eqref{def:D} of $\mc{D}$. Moreover, there are two important subgroups associated with $\mc{G}$: the translation group $\mc{L}$ and the point group $\bar{\mc{G}}$. The group $\mc{L}$ is simply defined as the subgroup of $\mc{G}$ of the translations. Recalling that an isometry $g$ of $\R^2$ is uniquely decomposed as $g = t_a \bar{g}$, 
where $\bar{g} \in O_2$ and $t_a$ is a translation: $t_a x = x + a$ for $x \in \R^2$, we can introduce a homomorphism $\pi: M_2 \to O_2$ by $\pi(g) = \bar{g}$. Here $O_d$ denotes the group of orthogonal matrices on $\R^d$.

The point group $\bar{\mc{G}}$ is then defined as the image of $\mc{G}$ under the map $\pi$, which, by \cite[Proposition 6.5.10]{artin2011algebra}, is a finite group. We remark that the point group $\bar{\gr}$ is generally not a subgroup of $\gr$. For ease of analysis,  
in what follows, we consider the following mild assumption on the translation group $\mc{L}$: 
\begin{align} \label{assp:lattice}
\text{the lattice}\ \Lambda \ \text{gives all the translation vectors in} \ \mc{L}\,.
\end{align}
By \cite[Proposition 6.5.11]{artin2011algebra}, along with the assumption \eqref{assp:lattice} and the definition \eqref{def:dual_lattice}, we have 
\begin{equation} \label{eq:rela_point_trans}
\bar{g}(\Lambda) = \Lambda \q \text{and} \q \bar{g}(\Lambda^*) = \Lambda^*\,,
\end{equation}
that is, the group $\bar{\mc{G}}$ is contained in the symmetry group of $\Lambda$ and $\Lambda^*$. One may also observe 
that the symmetry group $\gr$ implicitly enforces some symmetry conditions on the shape of $D$ in the sense that up to a translation, it holds that  $\bar{g} \in \bar{\gr}$ maps $D$ to $D$, i.e., $\bar{g}(D) = D + (b, 0)$ for some $b \in \R^2$.

We next consider the invariance of the operator $\td$ under the action of the symmetry group $\gr$ of $\D$. We introduce the quasi-periodic extension of the function space ${\bf L}^2(D)$: 
\begin{align*}
{\bf L}^2_{\alpha}(D) := \Big\{
\sum_{\gamma \in \Lad} (f\chi_D)(x - (\gamma,0)) e^{i \alpha\dd \gamma}\,,\ f \in {\bf L}^2(D)\Big\} \subset {\bf L}^2_\alpha(Y_\infty)\,. 
\end{align*}
By definition of $\td$ and ${\bf L}^2_\alpha(D)$, we readily see 
\begin{align} \label{auxeqq_1}
    \td\big[\vp|_D\big] = \TT_{D + (\gamma,0)}^{\alpha,\ww} \big[\vp|_{D + (\gamma,0)}\big]\,, \q \forall \gamma \in \Lad\,, \ \vp \in {\bf L}^2_\alpha(D)\,.
\end{align} 
Then, we define the group action $\mc{O}$ of $\mc{G}$ on the vector fields $\vp$ by  
\begin{equation} \label{eq:quasiregular}
    \mc{O}_g [\vp] := \bar{g} \vp (g^{-1} x)\,,\q \text{for}\   g \in \mc{G}\,.
\end{equation}
Note that $g^{-1}(x + \gamma) = g^{-1}x + \bar{g}^{-1}\gamma$ for any $x, \gamma \in \R^2$, as $g$ is an invertible affine map. Thus, by the relation \eqref{eq:rela_point_trans},  it holds that, for vector fields $\vp$ with the quasi-periodicity $\alpha \in \mc{B}$, 
\begin{align}\label{eq:transog}
    \mc{O}_g[\vp](x + \gamma) = \bar{g} \vp(g^{-1} x+ \bar{g}^{-1}\gamma) =  \mc{O}_g[\vp](x) e^{i \alpha \dd \bar{g}^{-1}\gamma} =  \mc{O}_g[\vp](x) e^{i \bar{g}[\alpha] \dd \gamma}\,.
\end{align}
This means that $\mc{O}_g$ is an isometric isomorphism from ${\bf L}^2_\alpha(Y_\infty)$ to ${\bf L}^2_{\bar{g}[\alpha]}(Y_\infty)$. With the above preparation, we have the following result. 
\begin{lemma} \label{lem:symtd}
    For $\vp \in {\bf L}^2_\alpha(D)$, it holds that 
    \begin{align} \label{eq:symtd}
       \TT_D^{\bar{g}[\alpha],\ww}[\mc{O}_g \vp] = \mc{O}_g\td[\vp]\,.
    \end{align}
\end{lemma}
\begin{proof}
By definition,  we first have 
    \begin{align*}
        \TT_D^{\bar{g}[\alpha],\ww}[\mc{O}_g \vp] & = (\ww^2 + \na \ddiv)\int_{D} G^{\bar{g}[\alpha],\ww}(x,y) 
        \bar{g} \vp (g^{-1} y)\ dy  \\
      &  = (\ww^2 + \na \ddiv)\bar{g} \int_{g^{-1}(D)} G^{\bar{g}[\alpha],\ww}(x,g y) 
      \vp (y)\ dy\,.
    \end{align*}
We then calculate, by writing $g = t_a \bar{g}$ and noting that $\bar{g}$ maps $\Lad^*$ to itself (cf.\,\eqref{eq:rela_point_trans}), 
    \begin{align*}
        G^{\bar{g}[\alpha],\ww}(x,g y) & = \frac{1}{2}\sum_{q \in \Lad^*} i \beta_q^{-1} e^{i (\bar{g}[\alpha + q])\dd (x' - g y')} e^{i \beta_q |x_3 - y_3|}\\
& =  \frac{1}{2}\sum_{q \in \Lad^*} i \beta_q^{-1} e^{i (\alpha + q)\dd (\bar{g}^{-1} (x' - a) - y')} e^{i \beta_q |x_3 - y_3|} \\
& = G^{\alpha,\ww}(g^{-1}x, y)\,.
    \end{align*} 
    We also find, by direct computation, for a smooth vector field $\vp$, 
    \begin{align*}
        \na \ddiv \mc{O}_g \vp =  \na (\ddiv \vp)(g^{-1} x) = \bar{g} (\na \ddiv \vp)(g^{-1} x) = \mc{O}_g [\na \ddiv\vp]\,.
     \end{align*}
Hence, we can conclude, from the above formulas,
     \begin{align*}
        \TT_D^{\bar{g}[\alpha],\ww}[\mc{O}_g \vp] & = (\ww^2 + \na \ddiv)\bar{g} \int_{g^{-1}(D)} G^{\alpha,\ww}(g^{-1}x, y)
      \vp (y) dy \\
      & = \mc{O}_g \TT^{\alpha,\ww}_{g^{-1}(D)}[\vp] = \mc{O}_g \TT^{\alpha,\ww}_{D}[\vp]\,,
    \end{align*} 
where the last equality follows from $g^{-1}(D) = D + (\gamma, 0)$ for some $\gamma \in \Lad$ and \eqref{auxeqq_1}. 
\end{proof}

As a consequence of Lemma \ref{lem:symtd}, if $(\ww_*, \vp_*)\in \C \t {\bf L}_\alpha^2(D)$ is an eigenpair of $\A(\alpha,\ww)$, then we have 
\begin{align} \label{eq:inva_eig}
\A(\bar{g}[\alpha],\ww_*)[\mc{O}_g \vp_*] = \mc{O}_g\A(\alpha,\ww_*)[\vp_*] = 0\,,
\end{align}
and equivalently, the following corollary holds (note from the definition \eqref{def:bri_zone} that $\bar{g}$ maps $\mc{B}$ to itself for any $\bar{g} \in \bar{G}$). 

\begin{corollary} \label{coro:symband}
The band dispersion function $\ww_j(\alpha)$ for each $j \ge 0$ has the following symmetry: 
  \begin{align*}
      \ww_j(\bar{g}\alpha) = \ww_j(\alpha)\,,\q \forall \bar{g}\in \bar{\gr}\,, \ \alpha \in \mc{B}\,.
  \end{align*}
Moreover, suppose that $\vp \in {\bf L}_\alpha^2(D)$ is an eigenfunction of $\A(\alpha,\ww)$ for some resonance $\ww$ with the quasi-periodicity $\alpha$. Then $\mc{O}_g \vp$ is an eigenfunction for the same resonance but with the quasi-periodicity  $\bar{g}\alpha$.
\end{corollary}

This motivates us to introduce the notion of high symmetry points in the Brillouin zone $\mc{B}$ \cite{kaxiras2019quantum}. 

\begin{definition} \label{def:highsympoint}
We say that $\alpha \in \mc{B}$ is a high symmetry point, if the subgroup $\bar{\mathcal{G}}_\alpha = \{\bar{g} \in \bar{\mathcal{G}}\,;\ \bar{g} \alpha - \alpha \in \Lad^*\}$ of the point group $\bar{\mc{G}}$ is nontrivial. 
\end{definition} 

We shall discuss the geometric multiplicity of a 
resonance $\ww(\alpha)$ at the high symmetry points $\alpha_* \in \mc{B}$ by using the group representation theory. To do so, we first recall some preliminaries following \cite{hewitt2013abstract}. It has been seen that $\bar{\mc{G}}_{\alpha_*}$ is a finite compact subgroup of $O_2$. We define its representation on a Hilbert space $\mc{H}$ by a group homomorphism $\pi: \bar{\mc{G}}_{\alpha_*} \to GL(\mc{H})$, where $GL(\mc{H})$ is the group of invertible continuous linear operators on $\mc{H}$. The representation $ \pi$ is unitary if $\pi_g$ is a unitary operator for each $g$. Moreover,
we say that $\pi$ is finite dimensional if $\mc{H}$ is finite dimensional, and that a subspace $\mc{H}' \subset \mc{H}$ is invariant for $\pi$ if $\pi_g(\mc{H}') \subset \mc{H}'$ for any $g \in \bar{\mc{G}}_{\alpha_*}$. It is natural to define the associated subrepresentation $\pi|_{\mc{H}'}: \bar{\mc{G}}_{\alpha_*} \to GL(\mc{H}')$ by $ (\pi|_{\mc{H}'})_g := \pi_g|_{\mc{H}'}$. If the only invariant subspaces are $\{0\}$ and $\mc{H}$, then the representation $\pi$ is said to be irreducible.  

To proceed, we note from \eqref{eq:transog} that
the map $\mc{O}_g$ in  \eqref{eq:quasiregular} actually defines a unitary representation $\mc{O}$ of $\bar{\mc{G}}_{\alpha_*}$ on ${\bf L}_{\alpha_*}^2(D)$ at a high symmetry point $\alpha_*$.  Let $\hat{\mc{G}}_{\alpha_*}$ be the set of (equivalent classes of) irreducible representations of $\bar{\mc{G}}_{\alpha_*}$. It is known that any irreducible representation of a compact group is finite-dimensional \cite{knapp2001representation}. Hence, we can denote by $d_\si$ the dimension of a representation $\si \in \hat{\mc{G}}_{\alpha_*}$ and define its character $\chi_\si$ by $\chi_\si(g) = \tr \pi_g$. Then for each $\si \in \hat{\mc{G}}_{\alpha_*}$, we introduce the projection $P_\si: {\bf L}_{\alpha_*}^2(D) \to {\bf L}_{\alpha_*}^2(D)$ by 
\begin{align} \label{eq:proj}
P_\si[\vp] = \frac{1}{|\bar{\mc{G}}_{\alpha_*}|}  \sum_{g \in \bar{\mc{G}}_{\alpha_*}}  d_\si \overline{\chi_\si(g)}\mc{O}_g[\vp]\,,
\end{align}
where $|\bar{\mc{G}}_{\alpha_*}|$ is the order of the finite group $\bar{\mc{G}}_{\alpha_*}$. 
The following lemma is from \cite[Theorem (27.44)]{hewitt2013abstract}. 

\begin{lemma} \label{lem:group_rep}
The projections $P_\si$ in \eqref{eq:proj}
are orthogonal and give an orthogonal decomposition of ${\bf L}_{\alpha_*}^2(D)$:
\begin{align} \label{def:proj}
    {\bf L}_{\alpha_*}^2(D) = \oplus_{\si \in \hat{\mc{G}}_{\alpha_*}} M_\si\,,\q  M_\si := P_\si {\bf L}_{\alpha_*}^2(D)\,.
\end{align}
 For each $\si$, $M_\si$ is either zero or a direct orthogonal sum of invariant subspaces $M_{j,\si}$  ($1 \le j \le m_\si \le \infty$) of dimension $d_\si$  with $\mc{O}|_{M_{j,\si}} \in \si$. Moreover, $M_\si$ is the smallest closed subspace of ${\bf L}_{\alpha_*}^2(D)$ containing all the invariant subspaces of ${\bf L}_{\alpha_*}^2(D)$ on which $\mc{O}$ is in the equivalent class of the irreducible representation $\si$. 
\end{lemma}

We are now ready to give some results on the multiplicity of resonances at $\alpha_*$. By Lemma \ref{lem:symtd} and \eqref{eq:proj}, we see that the projection $P_\si$ commutes with the operator $\mc{A}(\alpha_*,\ww)$ for any $\ww$. It follows that $M_\si$ is an invariant subspace for $\mc{A}(\alpha_*,\ww)$, and the following result holds. 

\begin{proposition}
Let $\alpha_* \in \mc{B}$ be a high symmetry point and let $M_\si$ be defined in \eqref{def:proj}.
For any resonance $\ww(\alpha_*)$ of the operator $\mc{A}(\alpha_*,\ww)$ restricted on the subspace $M_\si$, its geometric multiplicity is a multiple of the dimension $d_\si$ of the irreducible representation $\si$, i.e., $\dim \mc{A}(\alpha_*,\ww)|_{M_\si} = c d_\si$ for some positive integer $c$.
\end{proposition}

\begin{proof}
We consider the eigenspace $V^{\ww(\alpha_*)}: = \{E \in M_\si\,;\  \mc{A}(\alpha_*,\ww)|_{M_\si}[E] = 0\}$ associated with the resonance $\ww(\alpha_*)$, which is of finite dimension. Note from \eqref{eq:inva_eig} that the subspace $V^{\ww(\alpha_*)}$ is invariant for the representation $\mc{O}$. Then, by Lemma \ref{lem:group_rep} above, we have the decomposition: $V^{\ww(\alpha_*)} = \oplus_{j = 1}^c V^{\ww(\alpha_*)}_j$ for a positive integer $c$, where each space $V^{\ww(\alpha_*)}_j$ satisfies $\mc{O}|_{V^{\ww(\alpha_*)}_j} \in \si$. It readily yields $\dim \mc{A}(\alpha_*,\ww)|_{M_\si} = c d_\si$ and completes the proof. 
\end{proof}

\section{Existence of embedded eigenvalues} \label{sec:exist_real}
We have seen in Theorem \ref{thm:ess_discre} that the square of the light line $|\alpha| = \inf_{q \in \Lad^*} |\alpha + q|$ is nothing else than the infimum of the essential spectrum, and there exists a real scattering resonance below $|\alpha|$. In this section, \mb{we focus on
the high symmetry point $\alpha = 0$ with $\bar{\gr}_0 = \bar{\gr}$ by Definition \ref{def:highsympoint}, which corresponds to the case of the normal incident wave}. We further assume that the dielectric permittivity $\ep$ in \eqref{assp:high_constrast} has the following symmetry:
\begin{align} \label{assp:sym}
    \ep(x,y,z)|_{Y_h} = \ep(-x, y,z)|_{Y_h}\,,\q  \ep(x,y,z)|_{Y_h} = \ep(x, -y,z)|_{Y_h}\,.
\end{align}
It follows that the dihedral group $D_2$ is a subgroup of both the symmetry group $\mc{G}$ and the point group $\bar{\gr}$ of the metasurface. 
We shall show that in this scenario, the real subwavelength scattering resonances for dielectric metasurfaces can exist above the light line, which are the eigenvalues of the Maxwell operator $\mc{M}_\ep(0)$ embedded in the continuous spectrum $[0, \infty)$ and hence named as \emph{embedded eigenvalues}. The associated Bloch modes are the so-called \emph{bound states in the continuum} \cite{koshelev2018asymmetric,zhen2014topological,yuan2021parametric}. 

Our argument is based on a regularized variational formulation in terms of the magnetic field. We follow the discussions in \cite{bao2000scattering,bao1997variational,hazard1996solution} and note that $\ww > 0$ is a real scattering resonance if and only if there is a divergence-free periodic magnetic field $H$ solving the equation: 
\begin{align} \label{eq1}
&\curl (\ep^{-1} \curl H) - \na (\ep^{-1} \ddiv H) - \ww^2 H = 0\,,
\end{align}
with the Rayleigh-Bloch expansion \eqref{eq:RBexpansion}. Here 
$- \na (\ep^{-1} \ddiv H)$ is a regularization term added for the coercivity of the associated sesquilinear form. As we have assumed that the magnetic permeability $\mu$ is constant on the whole space, the magnetic field $H$ is expected to have the $H^1$-regularity. Similarly to \eqref{eq:desa}, to incorporate the outgoing radiation condition into the variational form, we need the following DtN operator:
\begin{align} \label{def:dtn_normal}
    \mc{T} f := \sum_{q \in \Lambda^*} i \beta_q f_q e^{i q \dd x}: H_p^{1/2}(\Sigma_{\pm h}) \to  H^{-1/2}_p(\Sigma_{\pm h})\,,
\end{align}
where $f_q$, $q \in \Lad^*$, are the Fourier coefficients of $f$, i.e., $f = \sum_{q \in \Lad^*} f_q e^{i q \dd x}$, and $\beta_q$ are defined in \eqref{def:bq} with $\alpha = 0$. For a vector field $f \in {\bf H}_p^{1/2}(\Sigma_{\pm h})$, we define $\mc{T} f$ by acting $\mc{T}$ on each component of $f$. Then the outgoing condition for $H$ can be modelled by the  transparent boundary condition:
\begin{align} \label{bc:transp}
    \frac{\p}{\p \nu} H = \mc{T} H \q \text{on}\ \Sigma = \Sigma_h \cup \Sigma_{-h}\,.
\end{align}
We also introduce the following  pseudo-differential operator:
\begin{align} \label{def:boundary_opt}
    \mc{R} = (\p_1, \p_2 ,\mc{T}): H_p^{1/2}(\Sigma_{\pm h}) \to  {\bf H}_p^{-1/2}(\Sigma_{\pm h})\,.
\end{align}
For convenience, we shall simply write $\na_T$ for $(\p_1, \p_2)$
in what follows. By a standard integration by parts with the boundary condition \eqref{bc:transp}, we obtain the variational problem for \eqref{eq1}:
find $(\ww, H) \in \R\backslash\{0\} \t {\bf H}_p^1(Y_h)$ such that 
\begin{equation} \label{eq:quadra_form_eig}
     A_{\tau,\ww}(\vp, H) = \ww^2 (\vp, H)_{Y_h}\, \q \forall \vp \in {\bf H}_p^1(Y_h)\,,
\end{equation}
where 
\begin{equation} \label{eq:quadra_form}
    A_{\tau,\ww}(\vp, H): = (\curl \vp, \ep^{-1} \curl H)_{Y_h} + (\ddiv \vp, \ep^{-1} \ddiv H)_{Y_h} + \l \vp , \nu \t (\mc{R} \t H)\r_{\Sigma} -  \l \nu \dd \vp, \mc{R} \dd H \r_{\Sigma}\,.
\end{equation}
Moreover, a direct computation gives 
\begin{align}  \label{eq4}
    \l \vp , \nu \t (\mc{R} \t H)\r_{\Sigma} -  \l \nu \dd \vp, \mc{R} \dd H \r_{\Sigma} = - \l \vp, \mc{T} H \r_{\Sigma} - \l \nu \dd \vp, \na_T \dd H \r_{\Sigma} - \l \na_T \dd \vp, \nu \dd H \r_{\Sigma}\,.
\end{align}
Therefore, to show the existence of embedded eigenvalues, it suffices to prove that there exists a real subwavelength eigenfrequency $\ww$ for \eqref{eq:quadra_form_eig} with a divergence-free eigenfunction.

Let us now focus on the eigenvalue problem \eqref{eq:quadra_form_eig}. Motivated by the analysis in \cite{bonnet1994guided}, we introduce the operators $\mc{T}^{\mc{P}}$ and $\mc{T}^{\mc{E}}$ by considering the sum in \eqref{def:dtn_normal} over $\Lad^*_{\mc{P}}$ and $ \Lad^*_{\mc{E}}$, respectively. Then the operators $\mc{R}^{\mc{P}}$ and $\mc{R}^{\mc{E}}$ can be defined as in \eqref{def:boundary_opt} correspondingly. By using $\mc{R}^\mc{E}$ and the identity \eqref{eq4}, we define the sesquilinear form on ${\bf H}_p^1(Y_h)$:
\begin{equation} \label{eq:quadra_form_real}
    A^{\mc{E}}_{\tau,\ww}(\vp, H): = (\curl \vp, \ep^{-1} \curl H)_{Y_h} + (\ddiv \vp, \ep^{-1} \ddiv H)_{Y_h} - \l \vp, \mc{T}^{\mc{E}} H \r_{\Sigma} - \l \nu \dd \vp, \na_T \dd H \r_{\Sigma} - \l \na_T \dd \vp, \nu \dd H \r_{\Sigma}\,,
\end{equation}
which is actually the real part of $A_{\tau,\ww}$, namely, $ \re \A_{\tau, \ww}(H, H) = A_{\tau,\ww}^{\mc{E}}(H,H)$. By a similar argument as in the proof of Proposition \ref{lem:expdecay}, we see that the eigenfunction $H$ for \eqref{eq:quadra_form_eig} associated with a real resonant frequency satisfies $H_q = 0$ on $\Sigma$ for $q \in \Lad_{\mc{P}}^*$ and hence 
exponentially decays as $x_3 \to \infty$. Moreover, we have the following lemma.  

\begin{lemma} \label{lem1}
 The magnetic field $H \in {\bf H}_p^1(Y_h)$ solves \eqref{eq:quadra_form_eig} for some $\ww > 0$ if and only if it satisfies
 \begin{align}\label{eq2}
       A^{\mc{E}}_{\tau,\ww}(\vp, H) = \ww^2 (\vp, H)_{Y_h}\,, \q \forall \vp \in {\bf H}_p^1(Y_h)\,,
\end{align}
with 
\begin{align}\label{eq3}
    H_q = 0 \q \text{on}\ \Sigma_{\pm h}\, \q \text{for} \ q \in \Lambda_*^{\mc{P}}\,.
\end{align}
\end{lemma}
\begin{remark}
    The variational problem \eqref{eq2} in the above lemma is equivalent to the equation \eqref{eq1} with the transparent boundary condition $\p_\nu H = \mc{T}^{\mc{E}} H$ on $\Sigma$, which only allows the exponentially decaying modes of the field $H$ passing through the boundary $\Sigma$ and going into the far field. The second condition \eqref{eq3} means that if the eigenpair of \eqref{eq2} is also an eigenpair of \eqref{eq:quadra_form_eig}, then the propagating part of $H$ indeed vanishes. 
\end{remark}

As suggested by Lemma \ref{lem1}, we shall prove the existence of real resonances for problem \eqref{eq:quadra_form_eig} in the subwavelength regime by the following two steps. \mb{First, we consider the} eigenvalue problem \eqref{eq2} and show that the real subwavelength resonant frequency does exist when the contrast $\tau$ is large enough. Then, with the help of the tools developed in Section \ref{sec:sym_band_multi}, we show that under the symmetry assumption \eqref{assp:sym}, we can find an associated eigenfunction of \eqref{eq2} satisfying the condition \eqref{eq3}.

To analyze the eigenvalue problem \eqref{eq2}, we first note that the sesquilinear form $A^{\mc{E}}_{\tau,\ww} (\cdot, \cdot)$ is conjugate-symmetric, i.e., $A^{\mc{E}}_{\tau,\ww} (H, \vp) = \overline{A^{\mc{E}}_{\tau,\ww} (\vp, H)}$. We next investigate its coercivity and dependence on the parameters $\tau$ and $\ww$ in detail. We compute, by \eqref{eq:quadra_form_real}, for $H \in {\bf H}_p^1(Y_h)$, 
\begin{align} \label{eq6}
    A^{\mc{E}}_{\tau,\ww}(H,H) = \int_{Y_h} \ep^{-1} |\curl H|^2 + \ep^{-1} |\ddiv H|^2\ dx
    - \l H, \mc{T}^{\mc{E}} H \r_{\Sigma} - 2\, \re \l \nu \dd H, \na_T \dd H \r_{\Sigma}\,,
\end{align}
which, along with the following identity, 
\begin{align} \label{eq5}
\int_{Y_h} |\curl H|^2 + |\ddiv H|^2 \ dx =     \int_{Y_h} |\na H|^2 \ dx + 2\, \re \l \nu \dd H, \na_T \dd H \r_{\Sigma}\,,
\end{align}
readily gives 
\begin{align} \label{eq7}
      A^{\mc{E}}_{\tau,\ww}(H,H) =& \int_{Y_h} (\frac{1}{\ep} - \frac{1}{1 + \tau}) |\curl H|^2 + (\frac{1}{\ep} - \frac{1}{1 + \tau}) |\ddiv H|^2 \ dx
    - \l H, \mc{T}^{\mc{E}} H \r_{\Sigma} - 2\, \re \l \nu \dd H, \na_T \dd H \r_{\Sigma} \notag \\
      &  + \frac{1}{1 + \tau} \int_{Y_h} |\na H|^2 \ dx + \frac{2}{1 + \tau}\, \re \l \nu \dd H, \na_T \dd H \r_{\Sigma}\,. 
\end{align}
Recall that we are mainly interested in the high contrast and subwavelength regime. We follow \cite{ammari2020mathematical} to introduce a new real variable:
\begin{align} \label{def:new}
\h{\ww} = \sqrt{\tau} \ww\,,
\end{align}
and assume that $\h{\ww}$ is in some compact set, so that $\ww$ is in the subwavelength regime and there holds
\begin{align}\label{eqlattice}
\Lad_{\mc{E}}^* = \Lad^*\backslash \{0\} \,.
\end{align}
By definition of $\mc{T}^{\mc{E}}$, we have, for $\ww > 0$, 
\begin{equation} \label{eq:key_est_1}
 - \l H, \mc{T}^{\mc{E}} H\r_{\Sigma_{\pm h}} = \sum_{q \in \Lambda_\mc{E}^*} \sqrt{|q|^2 - \ww^2} |H_q|^2  \ge 0\,. 
\end{equation}
Then, thanks to \eqref{def:new}, we can derive, for $\tau$ large enough and $\h{\ww}$ in a compact set $K$ of $\R$, 
\begin{align} \label{eq10}
    & - \l H, \mc{T}^{\mc{E}} H \r_{\Sigma_{\pm h}} - 2\, \re \l \nu \dd H, \na_T \dd H \r_{\Sigma} + \frac{2}{1 + \tau}\, \re \l \nu \dd H, \na_T \dd H \r_{\Sigma} \notag \\
    \ge &  \sum_{q \in \Lambda_\mc{E}^*} \big(\sqrt{|q|^2 - \tau^{-1} \h{\ww}^2} - |q| \big)|H_q|^2 + \frac{1}{1 + \tau} |q||H_q|^2 \notag \\
     \ge &  \sum_{q \in \Lambda_\mc{E}^*} - c_0  \frac{\h{\ww}^2}{\tau} \frac{1}{|q|} |H_q|^2 + \frac{1}{1 + \tau} |q||H_q|^2 \notag \\
     \ge & \sum_{q \in \Lambda_\mc{E}^*} \frac{1}{1 + \tau} |q||H_q|^2 - c_0  \frac{\h{\ww}^2}{\tau} \frac{1}{\min_{q \in \Lad_{\mc{E}^*}} |q|} \norm{H}_{\Sigma} \,, 
\end{align}
where the first inequality follows from \eqref{eq:key_est_1} and the following estimate:
\begin{align*}
    2\, |\re \l \nu \dd H, \na_T \dd H \r_{\Sigma_{\pm h}}| \le  \sum_{q \in \Lambda_\mc{E}^*} |q| |H_q|^2\,,
\end{align*}
and the second inequality is from the following asymptotics uniformly in $q \in \Lad_{\mc{E}}^*$ and $\h{\ww} \in K$, as $\tau \to \infty$, 
\begin{align*}
|q| - \sqrt{|q|^2 - \tau^{-1} \h{\ww}^2} & = |q| - |q|\sqrt{1 - |q|^{-2}\tau^{-1}\h{\ww}^2} = O\big(|q|^{-1}\tau^{-1}\h{\ww}^2 \big)\,.
\end{align*}
Here the generic constant $c_0$ in \eqref{eq10} is independent of $\tau \gg 1$, $q \in \Lad_{\mc{E}}^*$ and $\h{\ww} \in K$. We recall the weighted trace inequality: for any given $\eta > 0$, there holds 
\begin{align} \label{eq:trace}
    \norm{u}_{\Sigma} \le \eta \norm{\na u}_{Y_h} + C(\eta) \norm{u}_{Y_h}\,,\q \forall u \in {\bf H}_p^1(Y_h)\,,
\end{align}
with a constant $C(\eta)$ satisfying $C(\eta) = O(1)$ as $\eta \to \infty$ and $C(\eta) = O(\eta^{-1})$ as $\eta \to 0$. Then it readily follows from the formula \eqref{eq7}, the estimate \eqref{eq10}, and the inequality \eqref{eq:trace} with
$$
\eta = \frac{1}{2} \frac{\tau}{1 + \tau} \frac{\min_{q \in \Lad_{\mc{E}^*}}|q|}{c_0 \h{\ww}^2}
$$
that there exists a constant $C$ independent of $\tau \gg 1$ and $\h{\ww} \in K$ such that 
\begin{align} \label{eq12}
    A_{\tau,\sqrt{\tau}^{-1}\h{\ww}}^{\mc{E}}(H, H) \ge \frac{1}{2} \frac{1}{1 + \tau} \int_{Y_h} |\na H|^2 \ dx - C \frac{\h{\ww}^2}{\tau} \int_{Y_h} |H|^2 \ dx \,,\q \forall H \in {\bf H}_p^1(Y_h)\,.
\end{align}
We define a new form on ${\bf H}_p^1(Y_h)$: 
\begin{align} \label{def:new_form}
     \w{A}_{\tau,\h{\ww}}^{\mc{E}}(\vp, H) =  \tau A^{\mc{E}}_{\tau,\sqrt{\tau}^{-1}\h{\ww}}(\vp, H)\,,
\end{align}
and consider the following eigenvalue problem equivalent to \eqref{eq2}:
\begin{align} \label{eq11}
    \tau A^{\mc{E}}_{\tau,\sqrt{\tau}^{-1}\h{\ww}}(\vp, H) =  \h{\ww}^2 (\vp, H)_{Y_h}\,,\q \forall \vp \in {\bf H}_p^1(Y_h)\,.
\end{align}
In view of Garding's inequality \eqref{eq12} and the compact embedding from ${\bf H}_p^1(Y_h)$ to ${\bf L}^2_p(Y_h)$, we have the following min-max principle: for $j = 0,1,2,\cdots$,
\begin{align} \label{eq:min_max}
\lad_j(\tau, \h{\ww}) =  \min_{\substack{V \subset {\bf H}_p^1(Y_h) \\ \dim V = j + 1}  } \max_{\substack{ H \in V \\ \norm{H}_{Y_h} = 1}} \w{A}_{\tau,\h{\ww}}^{\mc{E}}(H, H) \,,
\end{align}
recalling that $\w{A}_{\tau,\h{\ww}}^{\mc{E}}$ is a symmetric form. Then $\h{\ww}$ is an eigenfrequency for \eqref{eq11} if and only if for some $j$, there holds 
\begin{align} \label{eig_eq}
    \lad_j(\tau,\h{\ww}) = \h{\ww}^2\,.
\end{align}
We now summarize some basic important properties of  the form $\w{A}_{\tau,\h{\ww}}^{\mc{E}}$ and its eigenvalues $\lad_{j}(\tau,\h{\ww})$, according to what we have discussed above. 
\begin{proposition}\label{prop:basic_Prop}
\mb{Given any compact set $K$ of $\R$, for $\h{\ww} \in K$ and large enough $\tau > 0$, we have}
\begin{enumerate}
    \item The sesquilinear form $\w{A}_{\tau,\h{\ww}}^{\mc{E}}$ on ${\bf H}_p^1(Y_h)$ defined in \eqref{def:new_form} is conjugate-symmetric and bounded from below: \begin{align} \label{est:bound_below}
    \w{A}_{\tau,\h{\ww}}^{\mc{E}}(H,H) \ge \frac{1}{4} \norm{\na H}^2_{Y_h} - C \h{\ww}^2 \norm{H}_{Y_h}^2\,,\q \forall H \in {\bf H}_p^1(Y_h)\,,
\end{align} 
for a constant $C$ independent of $\tau \gg 1$ and $\h{\ww} \in K$. 
\item For each $j \ge 0$, the eigenvalue $\lambda_j(\tau,\h{\omega})$ in \eqref{eq:min_max} is continuous in $(\tau,\h{\ww})$ and decreasing in $\h{\ww}$.  
\end{enumerate}
In particular, if $\h{\ww} = 0$, the form $\w{A}_{\tau,0}^{\mc{E}}$ is symmetric and positive for any $\tau > 0$. Moreover, 
$\lad_0(\tau, 0) = 0$ holds with the eigenfunctions being constant vector fields, while for $j \ge 1$, there holds $\lad_j(\tau, 0) > 0$.
\end{proposition}

\begin{proof}
The estimate \eqref{est:bound_below} follows from \eqref{eq12} directly. The continuous dependence of $\lad_j(\tau,\h{\ww})$ on $\tau$ and $\h{\ww}$ is a consequence of the continuity of the form $\w{A}_{\tau,\h{\ww}}^{\mc{E}}$ in $\tau$ and $\h{\ww}$ and the standard perturbation theory. Note from \eqref{eq6} and \eqref{eq:key_est_1} that only the term $- \tau \l H, \mc{T}^{\mc{E}} H \r_{\Sigma}$ in $\w{A}_{\tau,\h{\ww}}^{\mc{E}}$ depends on the parameter $\h{\ww}$ and it is decreasing in $\h{\ww}$ for a fixed $H \in {\bf H}_p^1(Y_h)$. It then follows from the representation \eqref{eq:min_max} that the eigenvalue $\lambda_j(\tau,\h{\omega})$ is decreasing in $\h{\ww}$. We now consider the last statement. By the same estimates as in \eqref{eq10} and \eqref{eq12}, we see that for any $\tau > 0$, there holds
\begin{align*}
      \w{A}_{\tau,0}^{\mc{E}}(H,H) \ge \frac{1}{2} \frac{\tau}{1 + \tau} \norm{\na H}_{Y_h}^2\,,
\end{align*}
which readily gives $\lad_j \ge 0$ for any $j$, and $\w{A}_{\tau,0}^{\mc{E}}(H,H)  = 0$ if and only if $H$ is a constant vector field. This means that $\lad_0(\tau, 0) = 0$ with constant eigenfunctions. The proof is complete. 
\end{proof}

By Proposition \ref{prop:basic_Prop} above, we see that for a given compact set $K$ of $\R$, there exists a large enough $\tau_0$ such that for each $j$, $\lad_j(\tau,\h{\ww})$ in \eqref{eq:min_max} is well defined for $\tau \ge \tau_0$ and $\h{\ww} \in K$. Moreover, noting that $\lad_j(\tau,\h{\ww})$ is a decreasing and continuous curve in $\h{\ww} \in K$ with $\lad_j(\tau, 0) \ge 0$, one may expect that the equation \eqref{eig_eq} admits a unique solution, which then further gives a desired eigenvalue to the problem \eqref{eq11}. In order to rigorously show the existence of a solution to \eqref{eig_eq} in some given compact set $K$, we next investigate the limiting behavior of $\lad_j(\tau,\h{\ww})$ as $\tau \to \infty$.

For this, we first show that if $\h{\ww}$ is suitably small, then $\w{A}_{\tau,\h{\ww}}^{\mc{E}}(H,H)$ is increasing in $\tau$. By \eqref{eq7}, we have
\begin{align} \label{lim_eq1}
     \w{A}^{\mc{E}}_{\tau,\h{\ww}}(H,H) = &   \int_{Y_h \backslash D} \frac{\tau^2}{1 + \tau} \Big(|\curl H|^2 +  |\ddiv H|^2\Big) \ dx
    + \frac{\tau}{1 + \tau} \int_{Y_h} |\na H|^2 \ dx  \notag \\
      &  - \tau \l H, \mc{T}^{\mc{E}} H \r_{\Sigma} - 2 \tau\, \re \l \nu \dd H, \na_T \dd H \r_{\Sigma} + \frac{2 \tau}{1 + \tau}\, \re \l \nu \dd H, \na_T \dd H \r_{\Sigma}\,. 
\end{align}
Since both $\tau^2/(1 + \tau)$ and $\tau/(1 + \tau)$ are increasing functions in $\tau$, we only need to consider the boundary integrals in \eqref{lim_eq1} for the monotonicity of $\w{A}^{\mc{E}}_{\tau,\h{\ww}}$. We compute, by the Fourier expansion of $H = (H', H_3)$ on $\Sigma$, 
\begin{align} \label{eq:bry}
   & - \tau \l H, \mc{T}^{\mc{E}} H \r_{\Sigma_{\pm h}} - 2 \tau\, \re \l \nu \dd H, \na_T \dd H \r_{\Sigma_{\pm h}} + \frac{2\tau}{1 + \tau}\, \re \l \nu \dd H, \na_T \dd H \r_{\Sigma_{\pm h}}\\
= & \sum_{q \in \Lambda_\mc{E}^*} \tau \sqrt{|q|^2 - \tau^{-1} \h{\ww}^2} |H_q|^2 - 2 \big(\tau - \frac{\tau}{1 + \tau}\big)\, \re \big(i \bar{H}_{3,q} q \dd H'_{q}\big)\,.\notag
\end{align}
Consequently, by a simple Cauchy's inequality: $2 \, \re (i \bar{H}_{3,q} q \dd H'_{q}) \le |q| |H_q|^2$, it suffices to show that, for any $q \in \Lad_{\mc{E}}^*$,   
\begin{align} \label{eq8}
    \w{\tau} \sqrt{|q|^2 - \w{\tau}^{-1} \h{\ww}^2} -  \tau \sqrt{|q|^2 - \tau^{-1} \h{\ww}^2} \ge  \Big(\big(\w{\tau} - \frac{\w{\tau}}{1 + \w{\tau}}\big) - \big(\tau - \frac{\tau}{1 + \tau}\big)\Big)|q|\,,\q \forall \w{\tau} > \tau \gg 1\,.
\end{align}
For this, a direct computation gives 
\begin{align} \label{lim_eq2}
     \w{\tau} \sqrt{|q|^2 - \w{\tau}^{-1} \h{\ww}^2} -  \tau \sqrt{|q|^2 - \tau^{-1} \h{\ww}^2} = \frac{(\w{\tau}^2 - \tau^2)|q|^2 - (\w{\tau} - \tau) \h{\ww}^2   }{ \w{\tau} \sqrt{|q|^2 - \w{\tau}^{-1} \h{\ww}^2} +  \tau \sqrt{|q|^2 - \tau^{-1} \h{\ww}^2}}\,,
\end{align}
and 
\begin{align} \label{lim_eq3}
    \big(\w{\tau} - \frac{\w{\tau}}{1 + \w{\tau}}\big) - \big(\tau - \frac{\tau}{1 + \tau}\big) = (\w{\tau} - \tau) - \frac{\w{\tau} - \tau}{(1+\w{\tau})(1+\tau)}\,.
\end{align}
With the above formulas \eqref{lim_eq2} and \eqref{lim_eq3}, by Taylor expansion:
\begin{align*}
    &\w{\tau} \sqrt{1 - \w{\tau}^{-1} |q|^{-2} \h{\ww}^2} +  \tau \sqrt{1- \tau^{-1}|q|^{-2} \h{\ww}^2} \\
    = & \w{\tau} \Big(1 - \frac{1}{2} \w{\tau}^{-1} |q|^{-2} \h{\ww}^2 + O(\w{\tau}^{-2}|q|^{-4}\h{\ww}^4)\Big) +  \tau \Big(1- \frac{1}{2} \tau^{-1}|q|^{-2} \h{\ww}^2  + O(\tau^{-2}|q|^{-4}\h{\ww}^4) \Big) \\
    = & \w{\tau} + \tau - |q|^{-2}\h{\ww}^2 + O(\w{\tau}^{-1}|q|^{-4}\h{\ww}^4)  +  O(\tau^{-1}|q|^{-4}\h{\ww}^4)\,,
\end{align*}
we estimate, as $\w{\tau} > \tau \to \infty$,  
\small
\begin{align} \label{eq9}
    \frac{ \w{\tau} \sqrt{|q|^2 - \w{\tau}^{-1} \h{\ww}^2} -  \tau \sqrt{|q|^2 - \tau^{-1} \h{\ww}^2}}{(\w{\tau} - \frac{\w{\tau}}{1 + \w{\tau}}) - (\tau - \frac{\tau}{1 + \tau}) } &\ge \frac{(\w{\tau} + \tau)|q|^2 -  \h{\ww}^2   }{ \w{\tau} \sqrt{|q|^2 - \w{\tau}^{-1} \h{\ww}^2} +  \tau \sqrt{|q|^2 - \tau^{-1} \h{\ww}^2}} \Big(1 + \frac{1}{(1+\w{\tau})(1+\tau)}\Big) \notag \\
    & \ge \frac{(\w{\tau} + \tau)|q| -  |q|^{-1}\h{\ww}^2   }{ \w{\tau} \sqrt{1 - \w{\tau}^{-1} |q|^{-2} \h{\ww}^2} +  \tau \sqrt{1- \tau^{-1}|q|^{-2} \h{\ww}^2}} \Big(1 + \frac{1}{(1+\w{\tau})(1+\tau)}\Big) \notag \\
     & \ge |q| \frac{(\w{\tau} + \tau) -  |q|^{-2}\h{\ww}^2   }{\w{\tau} + \tau - |q|^{-2}\h{\ww}^2 + O(\w{\tau}^{-1}|q|^{-4}\h{\ww}^4)  +  O(\tau^{-1}|q|^{-4}\h{\ww}^4)} \Big(1 + \frac{1}{(1+\w{\tau})(1+\tau)}\Big)  \notag \\
     & \ge |q| \frac{1}{1 + O(\w{\tau}^{-2}|q|^{-4}\h{\ww}^4)  +  O((\tau\w{\tau})^{-1}|q|^{-4}\h{\ww}^4)} \Big(1 + O((\tau \w{\tau})^{-1}) \Big)\,,
\end{align}
\normalsize
where the generic constants involved in the big $O$-terms  are independent of $\tau \gg 1$, $q \in \Lad_{\mc{E}}^*$, and $\h{\ww}$ in a compact set $K$. Therefore, it is clear that if $K$ is suitably small, there holds 
\begin{align*}
     \frac{1}{1 + O(\w{\tau}^{-2}|q|^{-4}\h{\ww}^4)  +  O((\tau\w{\tau})^{-1}|q|^{-4}\h{\ww}^4)} \Big(1 + O((\tau \w{\tau})^{-1}) \Big) > 1 \,,\q \text{as}\  \w{\tau}  > \tau \to \infty\,,
\end{align*}
and then the inequality \eqref{eq8} follows. Moreover, we can conclude that the following lemma holds. 
\begin{lemma} \label{lem2}
\mb{There exists a compact neighborhood $K_* \subset \R$ of $0$ such that} for $\h{\ww} \in K_*$ and $H \in {\bf H}^1_p(Y_h)$, 
$\w{A}_{\tau,\h{\ww}}^{\mc{E}}(H,H)$ is increasing in $\tau$ when $\tau$ is large enough. 
\end{lemma}

We next follow the arguments in  \cite{simon1978canonical, hempel2000spectral} to define a limiting form of $\w{A}_{\tau,\h{\ww}}^{\mc{E}}$ with $\h{\ww} \in K_*$ as $\tau \to \infty$, by using the monotonicity of $\w{A}_{\tau,\h{\ww}}^{\mc{E}}$ in $\tau$. We regard 
the sesquilinear form $\w{A}_{\tau,\h{\ww}}^{\mc{E}}$ as a  densely defined and closed unbounded form on ${\bf L}_p^2(Y_h)$ with domain $\dom (\w{A}_{\tau,\h{\ww}}^{\mc{E}}) = {\bf H}^1_p(Y_h)$, which is bounded from below by \eqref{est:bound_below}. By the first representation theorem \cite[Theorem VI-2.1]{kato2013perturbation}, there is a unique  
self-adjoint operator $\mbb{A}_{\tau,\h{\ww}}^{\mc{E}} $ on ${\bf L}_p^2(Y_h)$ such that 
\begin{align*}
    \w{A}^{\mc{E}}_{\tau,\h{\ww}}(\vp, H) = (\vp, \mbb{A}_{\tau,\h{\ww}}^{\mc{E}} H)_{Y_h}\,, \q H \in \dom(\mbb{A}_{\tau,\h{\ww}}^{\mc{E}})\,,\ \vp \in \dom (\w{A}_{\tau,\h{\ww}}^{\mc{E}})\,,
\end{align*}
with $\dom(\mbb{A}_{\tau,\h{\ww}}^{\mc{E}}) $ dense in $\dom (\w{A}_{\tau,\h{\ww}}^{\mc{E}})$ with respect to the $\norm{\dd}_{{\bf H}^1(Y_h)}$-norm. Thanks to Lemma \ref{lem2}, we now define a closed limiting quadratic form on ${\bf L}^2_p(Y_h)$ by 
\begin{equation} \label{def:limit_form}
    \w{A}^{\mc{E}}_{\infty,\h{\ww}} (H,H) =  \sup_{\tau > 0} \w{A}_{\tau,\h{\ww}}^{\mc{E}}(H,H) = \lim_{\tau \to +\infty} \w{A}_{\tau,\h{\ww}}^{\mc{E}}(H,H) \,, 
\end{equation}
with the domain:
\begin{equation*}
  \dom (\w{A}^{\mc{E}}_{\infty,\h{\ww}}) := \big\{H \in {\bf H}_p^1(Y_h) \,;\ \sup_{\tau > 0} \w{A}_{\tau,\h{\ww}}^{\mc{E}}(H,H) < \infty \big\}\,;
\end{equation*}
see \cite[Theorem VIII-3.13a]{kato2013perturbation}. Similarly, there exists a unique self-adjoint operator $\mbb{A}_{\infty,\h{\ww}}^{\mc{E}}$ on ${\bf L}^2_p(Y_h)$, corresponding to the form $\w{A}_{\infty,\h{\ww}}^{\mc{E}}$. It is clear that both the operators $\mbb{A}_{\tau,\h{\ww}}^{\mc{E}}$ and $\mbb{A}_{\infty,\h{\ww}}^{\mc{E}}$ have compact resolvents and thus have purely discrete spectrum. 
We denote by $\{\mu_j(\h{\ww})\}_{j = 0}^\infty$ the eigenvalues of the operator $\mbb{A}_{\infty,\h{\ww}}^{\mc{E}}$, recalling that the eigenvalues of $\mbb{A}_{\tau,\h{\ww}}^{\mc{E}}$
have been characterized by \eqref{eq:min_max}. 
The following result is a corollary of \cite[Theorem VIII-3.15]{kato2013perturbation}, and helps us to estimate the solution of \eqref{eig_eq} (see Corollary \ref{coro:est} below).

\begin{proposition} \label{prop:conver_ladj}
Let $K_*$ be the compact set given in Lemma \ref{lem2}. It holds that for $j \ge 0$, 
\begin{align} \label{eq:pointconver} 
    \lad_{j}(\tau,\h{\ww}) \nearrow \mu_j(\h{\ww})\,, \q \text{as}\ \tau \to \infty\,,
\end{align}
pointwisely for $\h{\ww} \in K_*$, and $\mu_j(\h{\ww})$ is a decreasing function on $K_*$. 
\end{proposition} 

\begin{proof}
The monotonicity of the functions $\lad_j(\tau,\h{\ww})$ in $\tau$ follows from Lemma \ref{lem2} and the min-max formula \eqref{eq:min_max}, while \cite[Theorem VIII-3.15]{kato2013perturbation} guarantees the pointwise convergence of $\lad_j(\tau,\h{\ww})$ to $\mu_j(\h{\ww})$ as $\tau \to \infty$,  which, by definition, also gives that $\mu_j$ is decreasing in $\h{\ww}$. 
\end{proof}

\mb{We remark that the remaining analysis of this section only needs the monotone convergence \eqref{eq:pointconver} of $\lad_j(\tau,\h{\ww})$ in $\tau$ at the point $\h{\ww} = 0$. For notational simplicity, we write $\{\mu_j\}_{j \ge 0}$ for the eigenvalues $\{\mu_j(0)\}_{j \ge 0}$ of the form $\w{A}^{\mc{E}}_{\infty,\h{\ww}}|_{\h{\ww} = 0}$ defined in \eqref{def:limit_form} with $\h{\ww} = 0$.}

\begin{corollary}\label{coro:est} 
For each $j \ge 0$, when $\tau$ is large enough, the equation \eqref{eig_eq} has a unique solution $0 \le \h{\ww}_j(\tau) \le \sqrt{\mu_j}$ and $\sqrt{\tau}^{\sss -1} \h{\ww}_j$ gives a eigenfrequency for \eqref{eq2}. 
\end{corollary}

\begin{proof}
    It suffices to note that the curve $\h{\ww}^2$ intersects with the constant curve $\mu_j$ at the point $\sqrt{\mu_j}$, and that for $\tau$ large enough, by Propositions \ref{prop:basic_Prop} and \ref{prop:conver_ladj}, the curve $\lad_j(\tau,\h{\ww})$ is well defined on $\h{\ww} \in [0,\sqrt{\mu_j}]$, and it is decreasing and continuous in $\h{\ww}$ with $\lad_j(\tau,\h{\ww}) \le \mu_j$ and $\lad_j(\tau,0) \ge 0$ (strictly positive for $j \ge 1$).  
\end{proof}

By Corollary \ref{coro:est} with \eqref{def:new_form} and \eqref{eq11}, we have proved the existence of subwavelength eigenfrequencies for the problem \eqref{eq2} in the high contrast regime. 
We proceed to show that there exist associated eigenfunctions $H$ such that the condition \eqref{eq3} holds. Recalling Lemma \ref{lem1}, this will complete the proof of the existence of subwavelength resonances for the eigenvalue problem \eqref{eq:quadra_form_eig}. We first note that in the subwavelength regime, there holds $\Lad^*_{\mc{P}} = \{0\}$, and hence the condition \eqref{eq3} reduces to
\begin{align} \label{cond:red}
    \int_{\Sigma_{\pm h}} H \ dx = 0\,.
\end{align}
We recall that the symmetry assumption \eqref{assp:sym} implies that the dihedral group $D_2$ is a subgroup of the
symmetry group $\mc{G}$ and the point group $\bar{\gr}$ of the configuration $\mc{D}$. It is also easy to see that $\mc{O}_g$ with $g \in D_2$, defined in \eqref{eq:quasiregular}, gives a unitary representation of $D_2$ on ${\bf L}^2_p(Y_h)$. It 
is well known that $D_2$ is an abelian group with $4$ one-dimensional irreducible representations that are completely determined by its characteristic table. Similarly to \eqref{eq:proj} and \eqref{def:proj}, by the projections $P_\si$ corresponding to the group $D_2$, the space ${\bf L}^2_p(Y_h)$ can be written as the orthogonal sum of the following four subspaces: for $i,j = 0, 1$,  
\begin{align} \label{def:spaces}
    M_{i,j} = \{f \in {\bf L}^2_p(Y_h)\,;\ &f_1(-x_1,x_2,x_3) = (-1)^{1 + i} f_1(x_1,x_2,x_3), \ f_1(x_1, -x_2,x_3) = (-1)^j f_1(x_1,x_2,x_3), \notag \\ &f_2(-x_1,x_2,x_3) = (-1)^i f_2(x_1,x_2,x_3),\ f_2(x_1,-x_2,x_3) = (-1)^{1 + j} f_2(x_1,x_2,x_3), \notag \\ &f_3(-x_1,x_2,x_3) = (-1)^i f_3(x_1,x_2,x_3),\ f_3(x_1, - x_2,x_3) = (-1)^j f_3(x_1,x_2,x_3) \, \}\,.
\end{align}
\begin{theorem} \label{thm:main_real_1}
For the large enough contrast $\tau$, the eigenvalue problem \eqref{eq:quadra_form_eig} admits subwavelength eigenfrequencies $\ww = O(\sqrt{\tau}^{-1})$ as $\tau \to \infty$. 
\end{theorem}
\begin{proof}
We consider the subspaces ${\bf H}_{p,0}^1(Y_h) := (M_{0,0} \oplus M_{0,1} \oplus M_{1,0}) \cap {\bf H}_{p}^1(Y_h)$ and ${\bf H}_{p,1}^1(Y_h) := M_{1,1} \cap {\bf H}_{p}^1(Y_h)$. 
By a lengthy but direct computation, one can check that the above two spaces ${\bf H}_{p,0}^1(Y_h)$ and ${\bf H}_{p,1}^1(Y_h)$ are orthogonal with respect to both the sesquilinear form $A_{\tau,\ww}$ in \eqref{eq:quadra_form} and the ${\bf H}_{p}^1$-inner product. Therefore, if a field $H \in {\bf H}_{p, 1}^1(Y_h)$ satisfies the variational equation \eqref{eq:quadra_form_eig} for all $\vp \in {\bf H}_{p, 1}^1(Y_h)$, then \eqref{eq:quadra_form_eig} holds for all $ \vp \in {\bf H}_{p}^1(Y_h)$. It suffices to consider the problem \eqref{eq:quadra_form_eig} on the subspace ${\bf H}_{p, 1}^1(Y_h)$. We denote by $\lad_j^{\sss (1)}(\tau,\ww)$ the eigenvalues of the form $A_{\tau,\ww}^{\mc{E}}$ restricted on ${\bf H}_{p,1}^1(Y_h)$. By similar arguments as above, we can conclude that for each $j$, when $\tau$ is large enough, there exists $\ww_j^{\sss (1)}(\tau)$ satisfying $\lad_j(\tau, \ww_j^{\sss (1)}) =  (\ww_j^{\sss (1)})^2$ and $ \ww_j(\tau)^{\sss} = O(\sqrt{\tau}^{\sss -1})$ as $\tau \to \infty$, 
and it holds that for some $H \in {\bf H}_{p,1}^1(Y_h)$, the equation \eqref{eq2} holds for any $\vp \in  {\bf H}_{p,1}^1(Y_h)$;
see Corollary \ref{coro:est}. The proof is completed by noting that the condition \eqref{cond:red} holds for any  $H \in {\bf H}_{p,1}^1(Y_h)$ due to the symmetry, and that 
Lemma \ref{lem1} still holds when we restrict the space to ${\bf H}_{p,1}^1(Y_h)$. 
\end{proof}

It is clear that the set of eigenfrequencies of \eqref{eq:quadra_form_eig} is a subset of the eigenfrequencies of \eqref{eq2}. For simplicity, in the remaining of this section, we focus on the first nonzero positive eigenfrequency  for the eigenvalue problem \eqref{eq:quadra_form_eig}:
\begin{align} \label{eq:resonance}
\ww_*(\tau) := \sqrt{\tau}^{\sss -1}\h{\ww}_{j_0}(\tau)\,,\q \text{for some}\  j_0 > 0\,, 
\end{align}
where $\h{\ww}_{j_0}$ is given as in Corollary \ref{coro:est}. 
To show that $\ww_*(\tau)$ is a desired real scattering resonance embedded in the continuous spectrum for the problem \eqref{model}, we only need to prove that there exists an associated divergence-free eigenfunction $H_{j_0}$ of \eqref{eq:quadra_form_eig}. For this, we define $u = \ep^{-1} \ddiv H$ for $H \in {\bf H}_p^1(Y_h)$ satisfying the equation \eqref{eq1} with \eqref{bc:transp}. It was shown in \cite[Theorem 4.3]{bao1997variational} that $u$ satisfies 
\begin{align} \label{eq:scalar}
 \Delta u + \ww^2 \ep u = 0 \q \text{on}\ Y_h\,,\q  \frac{\p}{\p \nu} u = \mc{T} u \q \text{on}\ \Sigma = \Sigma_h \cup \Sigma_{-h}\,.
\end{align}
If $\ddiv H \neq 0$, then $(\ww,u)$ is clearly an eigen-pair for the above eigenvalue problem \eqref{eq:scalar}. In Appendix \ref{app:b}, we characterize the limiting behavior of the resonances of \eqref{eq:scalar} 
following the analysis in \cite{ammari2020mathematical,ammari2019subwavelength}, which readily provides a sufficient condition for the existence of the embedded eigenvalue of the Maxwell operator $\mc{M}_{\ep}(0)$. 

\begin{theorem} \label{main_real_2}
\mb{In the case of $\alpha = 0$,} let $\ww_*(\tau) = \sqrt{\tau}^{-1}\h{\ww}_{j_0}(\tau)$ be the lowest real resonance \eqref{eq:resonance} of the problem \eqref{eq:quadra_form_eig}. If 
there holds 
\begin{align} \label{assp_existence}
\sqrt{\mu_{j_0}} < c_0 := \max\{ \sqrt{\eta_0}^{-1}, \sqrt{\gamma_0}\}\,,
\end{align}
then, for large enough $\tau$,  $\ww_*(\tau)$ is a subwavelength embedded eigenvalue for the scattering problem \eqref{model}, where the constant
$\mu_{j_0}$ is the $j_0$th eigenvalue of the form $\w{A}_{\infty,0}^{\mc{E}}$ \eqref{def:limit_form}, $\eta_0$ is the largest eigenvalue of the operator $\mc{K}_D$ in \eqref{eq:limopscal}, and the constant $\gamma_0$ is given in \eqref{constant_lower}.
\end{theorem}

\begin{proof}
We note from Proposition \ref{prop:limit_scalar} 
and Proposition \ref{prop_resonance}
that for any small $\delta > 0$, when the contrast $\tau$ is large enough, the equation \eqref{eq:scalar} admits only trivial solution for any $ \ww \in (0, \sqrt{\tau}^{-1}(c_0 - \d))$. 
By Corollary \ref{coro:est}, we also have $\ww_*(\tau) \le \sqrt{\tau}^{-1}\sqrt{\mu_{j_0}}$.  Then it follows from the assumption $\sqrt{\mu_{j_0}} < c_0$ that for sufficiently large $\tau$, the problem \eqref{eq:scalar} is well-posed at $\ww_*(\tau)$ and hence has only zero solution $u = 0$, which means that the associated $H$ is divergence-free. 
\end{proof}

\begin{remark}
Since the set of real resonances of \eqref{eq:scalar} is clearly a subset of all scattering resonances that satisfy the asymptotics \eqref{asym_scalar} in the subwavelength regime, a straightforward way to relax the assumption \eqref{assp_existence} above is to consider the constant $c_0: = \min\{\sqrt{\eta_j}^{-1}\,;\ \ww = \sqrt{\tau \eta_j}^{-1} + O(\tau^{-1}) \ \text{is a real subwavelength resonance} \}$. In addition, we remark that to justify \eqref{assp_existence} in some scenarios, one may need to quantitatively estimate the resonances of \eqref{eq:quadra_form_eig} and \eqref{eq:scalar}, which is beyond the scope of this work.  
\end{remark}

\section{Fano-type reflection anomaly} \label{sec:fano_scattering}
This section is devoted to the investigation of the Fano-type reflection anomaly for the dielectric 
metasurface with broken symmetry. To this end, we first note from the symmetry assumption \eqref{assp:sym} that $\ep(x,y,z)|_{Y_h} = \ep(-x, -y,z)|_{Y_h}$. We define the symmetric and the antisymmetric vector fields for the inversion operation $(x',x_3) \to (-x', x_3)$ via the projections \eqref{eq:proj}: 
\begin{equation}\label{eq:syml2}
    {\bf L}^2_{sym}(D) := \{f \in {\bf L}^2(D)\,;\ f_1(x',x_3) = f_1(-x',x_3),\, f_2(x',x_3) = f_2(-x',x_3),\, f_3(x',x_3) = - f_3(-x',x_3) \},
\end{equation}
and 
\begin{equation} \label{eq:asyml2}
 {\bf L}^2_{ant}(D) := \{f \in {\bf L}^2(D)\,;\ f_1(x',x_3) = - f_1(-x',x_3),\, f_2(x',x_3) = - f_2(-x',x_3),\, f_3(x',x_3) =  f_3(-x',x_3)   \}\,.
\end{equation}
One may observe that $ {\bf L}^2_{sym}(D) := {\bf L}^2(D) \cap (M_{0,1} \cup M_{1,0})$ and $ {\bf L}^2_{ant}(D) := {\bf L}^2(D) \cap (M_{0,0} \cup M_{1,1})$, where $M_{i,j}$ are given in \eqref{def:spaces}. We write ${\bf H}_{ant}(\ddiv0,D) := {\bf L}^2_{ant}(D) \cap {\bf H}(\ddiv0,D)$ and ${\bf H}_{sym}(\ddiv0,D) := {\bf L}^2_{sym}(D) \cap {\bf H}(\ddiv0,D)$. Similarly, we denote by ${\bf H}_{0,ant}(\ddiv0,D)$ and ${\bf H}_{0,sym}(\ddiv0,D)$ the antisymmetric and symmetric parts of the space ${\bf H}_{0}(\ddiv0,D)$.

Let $(E^i, H^i)$ be the incident plane wave: 
\begin{equation} \label{eq:indplawave} 
    E^i = \ei e^{i\ww \di \dd x}, \quad H^i = \frac{1}{i\ww}\curl E^i = \di \t \ei e^{i \ww \di \dd x}\,,  \q \ww > 0\,,
\end{equation}
where $\di \in \S$ is the incident direction with $d_3 > 0$, and $\ei \in \S$ is the polarization vector with $\ei \dd \di = 0$. By Fredholm alternative, the scattering problem \eqref{model} with $\alpha = \ww d'$ is solvable for any incident plane wave with $\ww > 0$ \cite{schmidt2004electromagnetic}, although the solution may not be unique. Note that the subwavelength incident frequency, i.e., $\ww \ll 1$, is located in the first radiation continuum $|\alpha| < \ww < \inf_{q \in \Lad^*\backslash \{0\}} |\alpha + q|$, so that the scattered wave $E^s$ consists of a single propagating mode in the far field. In what follows, we write $f \sim g$ for two continuous functions (vector fields) on $\R$ if there holds 
$|f(x) - g(x)| = O(e^{- C |x_3|})$ as $x_3 \to \pm \infty$ for some $C > 0$. Then the total electric field $E$ is of the form: 
\begin{align} \label{eq:total_prop}
    E \sim \left\{ 
    \begin{aligned}
&      \po^t e^{i \ww \di_+ \dd x} &&\  x_3 \to \infty\,, \\
& \ei e^{i \ww \di_+ \dd x} + \po^r e^{i \ww \di_- \dd x}&& \ x_3 \to -\infty\,,
    \end{aligned}
    \right.
\end{align}
where $\di_{\pm} = (d',\pm d_3)$, and $\po^t$ and $\po^r$ are the reflection and transmission  polarization vectors, respectively. Moreover, the energy conservation gives \cite{dobson1994variational}
\begin{align*}
    |\po^t|^2 + |\po^r|^2 = |\po^i|^2 = 1\,.
\end{align*}
We next derive the volume integral representations for $\po^t$ and $\po^r$. We recall the quasi-periodic Green’s function \eqref{eq:qpgreen} with $\alpha = \ww d'$: 
\begin{align} \label{eq:green_incidence}
      G^{\ww d', \ww}(x) 
      & = \frac{i}{2} \frac{e^{i \ww d' \dd x'} e^{i\ww d_3|x_3 |}}{\ww d_3}  + \frac{1}{2}\sum_{q \in \Lad^*\backslash \{0\}}  \frac{e^{i(\ww d' + q)\dd x'} e^{-\sqrt{|\ww d' + q|^2-\ww^2}|x_3 |}}{\sqrt{|\ww d' + q|^2-\ww^2}}\,.
\end{align}
which implies, as $|x_3| \to \infty$, 
\begin{align*}
     G^{\ww d', \ww}(x) 
      \sim \frac{i}{2} \frac{e^{i \ww d' \dd x'} e^{i\ww d_3|x_3 |}}{\ww d_3}\,.
\end{align*}
Then it follows from the Lippmann-Schwinger equation \eqref{eq:Lippmann-Schwinger} that as $x_3 \to \pm \infty$,
\begin{align}  \label{eq:total_prop_2}
    E \sim \ei e^{i \ww \di_+ \dd x} + \po_{\pm}(\di,D) e^{i \ww \di_\pm \dd x}\,,
\end{align}
where $E = (1 - \tau \mc{T}_D^{\ww d', \ww})^{-1}[E^i] = \tau^{-1} \A_\tau(\ww d',\ww)^{-1}[E^i]$ (cf.\,\eqref{def:lipschi_operator}) and 
\begin{align} \label{eq:transreflec_polar}
    \po_{\pm}(\di,D) = \tau \frac{i \ww}{2 d_3} ({\rm I} - \di_\pm \otimes \di_\pm) \int_D e^{-i\ww \di_{\pm}\dd y} E(y)\ dy\,.
\end{align}
Combining \eqref{eq:total_prop} and \eqref{eq:total_prop_2} readily gives $\po^t = \ei + \po_+$ and $\po^r = \po_-$.

In Section \ref{sec:exist_real}, we have considered the case of the normal incident wave (i.e., $\di = (0,0,1)$) and discussed the existence of embedded eigenvalues under the symmetry assumption for the metasurface \eqref{assp:sym}. In this case, we have the bound states in the continuum with the energy confined near the metasurface. When the symmetry is broken, the real resonance may be shifted into the complex plane and the resonant modes can propagate into the far field, which leads to the anomalous Fano scattering phenomena near the resonances.  Inspired by recent physical advances \cite{koshelev2018asymmetric,zhen2014topological,bulgakov2017topological}, we break the symmetry \eqref{assp:sym} \mb{by either perturbing} the domain $D$ by a symmetric field $V$, i.e.,
\begin{align} \label{def:defor_field}
V(x)\in  {\bf L}^2_{sym}(\Omega)\cap {\bf C}^\infty(\Omega)   
\end{align}
with ${\bf L}^2_{sym}(\Omega)$ defined as in  \eqref{eq:syml2}, \mb{or considering} the incident direction $\di \in \S$ perturbed from  the normal incidence:
\begin{align} \label{eq:incidece_direction}
 \di = (\epsilon \alpha_0, d_3)\,, \q \text{with}\   \epsilon \ll 1 \ \text{and}\  |\alpha_0| = 1\,,
\end{align}
where $\Omega$ is an open neighborhood of $\overline{D}$.
We denote the perturbed domain by
\begin{align} \label{eq:pert_domain}
D_t := (I + t V)(D)\,,\q |t|\ll 1\,.
\end{align}
The main aim of this section is to approximate the scattered polarization $\po_\pm(\di, D_t)$ for such a symmetry-broken configuration uniformly in the incident frequency $\ww$ near the real resonance, and analyze the reflection and transmission energy $|\po^r|^2$ and $|\po^t|^2$.

For the sake of simplicity and clarity, we introduce the following assumptions:
\begin{enumerate}[label=\textbf{C.\arabic*},ref=C.\arabic*]
    \item \label{1} In view of Theorems \ref{thm:main_real_1} and \ref{main_real_2}, we assume that $\ww_*$ is a subwavelength embedded eigenvalue of \eqref{model} for the normal incidence with an antisymmetric eigenfunction $E \in {\bf H}_{ant}(\ddiv0,D)$ satisfying $ \A_\tau(0,\ww_*)[E] = 0$. 
    \item \label{2} By Theorem \ref{thm:asym_to}, without loss of generality, we let $ \ww_* = \sqrt{\tau \lad_0}^{-1} + O(\tau^{-1})$ with $\lad_0$ being the largest eigenvalue of the operator $\pd \mc{K}_D^{0,0}\pd$. We assume that the eigenspace of $\pd \mc{K}_D^{0,0}\pd$ for $\lad_0$ is spanned by a real antisymmetric field $\vp_1 \in {\bf H}_{0,ant}(\ddiv0,D)$ and a real symmetric field $\vp_2 \in {\bf H}_{0,sym}(\ddiv0,D)$. 
    \item \label{3} Let $\h{\ww} = \sqrt{\tau} \ww > 0$ be the scaled incident frequency. We assume that $\h{\ww} - \sqrt{\tau} \ww_*$ is in a small compact neighborhood $U$ of the origin. We always denote $\h{\ww}_0 = \sqrt{\lad_0}^{\sss -1}$ in what follows. 
\end{enumerate}

\begin{remark}
    To justify the reasonability of the assumption \eqref{2}, we first note that since the kernel $G^{0,0}$ of the operator $\mc{K}_D^{0,0}$ is real, it suffices to consider the real eigenfunctions of $\pd \mc{K}_D^{0,0}\pd$. Moreover, it is easy to check that for $\vp \in {\bf H}_0(\ddiv0,D)$, we can decompose it as $\vp = \vp_{sym} + \vp_{ant}$ with $\vp_{sym} \in {\bf H}_{0,sym}(\ddiv0,D)$ and $\vp_{ant} \in {\bf H}_{0,ant}(\ddiv0,D)$, where both ${\bf H}_{0,ant}(\ddiv0,D)$ and ${\bf H}_{0,sym}(\ddiv0,D)$ are invariant subspaces of $\pd \mc{K}_D^{0,0}\pd$. It follows that if $\vp$ is an eigenfunction for $\pd \mc{K}_D^{0,0}\pd$ associated with the eigenvalue $\lad_0$, which is neither antisymmetric or symmetric, then $\vp_{sym}$ and $\vp_{ant}$ are linearly independent eigenfunctions for $\lad_0$. However, it is a difficult task to characterize the conditions when the eigenspace is of dimension two 
and exactly spanned by $\vp_{sym}$ and $\vp_{ant}$. From the generic simplicity \cite{teytel1999rare,chitour2016generic}, we expect that such a result holds generically, at least in the dilute regime. Here we choose to leave the detailed investigation for future work. We remark that this fact can be proved for the acoustic Minnaert resonance in terms of the capacitance matrix \cite{ammari2020exceptional,ammari2021bound}, while
in our case, due to the infinite-dimensional kernel of the leading-order operator of $\mc{T}_D^{\alpha,\ww}$ consisting of magnetostatic fields, there is essentially an infinite number of scattering resonances hybridizing with each other
so that the capacitance-matrix characterization is not available. We also would like to mention that an alternative approach to analyzing the Fano anomaly is to regard the reflection energy $|\po^r|^2$ as an analytic function in $\ww$ and investigate its behavior near the embedded eigenvalue by asymptotic analysis as in \cite[Sections 5.2 and 5.3]{shipman2010resonant}. In this way, the assumption \eqref{2} could be removed (but some additional assumptions may be involved).
\end{remark}

We shall design a structure in terms of the \mb{parameters $\alpha_0, V, t, \epsilon$ and $\tau$}, under the assumptions \eqref{1}-\eqref{3}, such that the reflection energy $|\po^r|^2$ presents a Fano-resonance shape near $\ww_*$. The analysis for the transmission energy is similar. Before we proceed, we collect some asymptotic formulas from \cite{ammari2020mathematical,ammari2021bound} for Green's functions and the integral operators, which apply to the general domain $D$ without any symmetry assumption.

We start with the Taylor expansion of $G^{\ww d', \ww}$ in  \eqref{eq:green_incidence}: as $\ww \to 0$, 
\begin{align} \label{eq:expgreengn}
    G^{\ww d', \ww}(x) = - \frac{1}{2 i \ww d_3} + G_0^{d'}(x) + \ww G_1^{d'}(x) + \sum_{n = 2}^\infty \ww^n G_n^{d'}(x)\,,
\end{align}
where the function $G_0^{d'}$ is defined by
\begin{align} \label{asym_1}
  G_0^{d'}(x): = G^{0,0}(x) - \frac{d'\dd x'}{2 d_3}\,,
\end{align}
with $G^{0,0}(x)$ being the periodic Green's function (i.e., $G^{\alpha,\ww}$ with $\alpha
= 0$ and $\ww = 0$),
and $G_{1}^{d'}$ is given by
\begin{align} \label{green_1}
G_{1}^{d'} &:= - \frac{i \big(d_3|x_3| + d'\dd x' \big)^2  }{4 d_3} - d' \dd g_1(x)\,,
\end{align}
with $g_1$ being a  purely imaginary vector-valued function independent of $d'$ and $\ww$ and satisfying 
\begin{align} \label{eq:symm_green}
    g_1(x',x_3) = g_1(x',-x_3) = - g_1(-x',x_3)\,.
\end{align}
Then, by \eqref{eq:expgreengn}, we have the asymptotics for the operators $\np_D^{\ww d',\ww}$ and $\mc{T}_D^{\ww d',\ww}$ in \eqref{def:vectorpotential} and \eqref{def:electricpotential}: for $\vp \in {\bf L}^2(D)$,
\begin{align*}
   \np_D^{\ww d',\ww}[\vp] = - \frac{1}{2 i \ww d_3} \l 1 , \vp\r_D + \sum_{n = 0}^\infty \ww^n \kb{n}^{d'}[\vp] \,, 
\end{align*}
and 
\begin{align} \label{eq:exp_td}
    \TT_D^{\ww d',\ww} = \sum_{n = 0}^\infty  \ww^n \tb{n}^{d'}\,,
\end{align}
where the series converge in the operator norm
for $\ww \ll 1$. Here the operators $\kb{n}^{d'}$, $n \ge 0$, are defined by 
\begin{equation} \label{asym_2}
    \kb{n}^{d'}[\vp] := \int_D G^{d'}_n(x-y) \vp(y)\ d y\,,
\end{equation}
with $G^{d'}_n$ given in \eqref{eq:expgreengn}, while the operators $\tb{n}^{d'}, n \ge 0$, are defined by 
\begin{align*}
    \TT_{D,0}^{d'}[\vp] = \na \ddiv  \kb{0}^{d'}[\vp]\,,\q \TT^{d'}_{D,1}[\vp] = - \frac{1}{2 i d_3} \l 1 , \vp\r_D + \na \ddiv \kb{1}^{d'}[\vp]\,, 
\end{align*}
and
\begin{align}
    & \tb{n}^{d'}[\vp] := \kb{n-2}^{d'}[\vp] + \na \ddiv \kb{n}^{d'}[\vp]\,, \quad n \ge 2\,. \label{eq:operexp}
\end{align}
In particular, we can derive, by \eqref{asym_1} and \eqref{asym_2},  
\begin{align*}
    \kb{0}^{d'}[\vp] = \np_D^{0,0}[\vp] - \frac{d'\dd x'}{2 d_3} \l 1, \vp\r + \frac{1}{2 d_3}  \int_D d' \dd y' \vp(y)\ dy \,,
\end{align*}
which gives 
\begin{align} \label{eq:lead_dd}
    \pd \kb{0}^{d'} \pd = \pd \mc{K}_D^{0,0} \pd\,,  
\end{align}
and 
\begin{align} \label{eq:lead_ww}
    \TT^{d'}_{D,0}[\vp] = \na \ddiv \kb{0}^{d'}[\vp] =  \na \ddiv \np_D^{0,0}[\vp] = \mc{T}_D^{0,0}[\vp]\,.
\end{align}
The following result on the basic properties of the operator $\mc{T}_D^{0,0}$ is adapted from \cite[Lemma 3.6]{ammari2020mathematical}. 
\begin{lemma} \label{lem:op1}
Both $W$ and ${\bf H}_0(\ddiv 0, D)$ are the invariant subspaces of the operator $\mc{T}_D^{0,0} : {\bf H}(\ddiv0, D) \to {\bf H}(\ddiv0, D)$  with 
$\ker  \mc{T}_D^{0,0}  = 
{\bf H}_0(\ddiv 0, D)
$. Moreover, we have $\sigma(\mc{T}_D^{0,0}|_W) \subset [-1,0)$ and $\sigma(\mc{T}_D^{0,0}) = \sigma(\mc{T}_D^{0,0}|_W) \cup\{0\}$. 
\end{lemma}
By integration by parts, it is easy to see that
\begin{equation} \label{rela:simply}
     (\phi,\tb{n}^{d'}[\vp])_{D} = (\phi, (\kb{n-2}^{d'} + \na \ddiv \kb{n}^{d'})[\vp])_{D} = (\phi,\np_{D,n-2}^{d'}[\vp])_{D}\,,
\end{equation}
for any $\phi \in {\bf L}^2(D)$, $\vp \in {\bf H}_0(\ddiv 0, D)$, or $\vp \in {\bf L}^2(D)$, $\phi \in {\bf H}_0(\ddiv 0, D)$. 
It readily gives 
\begin{equation} \label{rela:simply_2}
  \P_{\rm s} \kb{n-2}^{d'} \P_{\rm t} = \P_{\rm s} \tb{n}^{d'} \P_{\rm t}\,, \q \text{for}\ {\rm s,t = d, w} \ \text{except the case}\ {\rm s,t = w}\,.
\end{equation}
It follows from \eqref{eq:exp_td} and \eqref{rela:simply_2} that
\begin{align} \label{block_exp_1}
 \P_{\rm s} \TT_D^{\ww d',\ww} \P_{\rm t} = \sum_{n = 2}^\infty \ww^n \P_{\rm s} \kb{n-2}^{d'} \P_{\rm t}\,,  \q \text{for}\ {\rm s,t = d, w} \ \text{except the case}\ {\rm s,t = w}\,,
\end{align}
and 
\begin{align} \label{block_exp_2}
    \pw \TT_D^{\ww d',\ww} \pw = \sum_{n = 0}^\infty  \ww^n  \pw \tb{n}^{d'} \pw \,.
\end{align}
We are now ready to give the first-order approximation for the subwavelength resonances $\ww$ (i.e., the poles of the meromorphic function $\A_{\tau}(\ww d', \ww)^{-1}$ near the origin) in terms of the high contrast $\tau$. For convenience, we define 
\begin{align} \label{def:kdd}
\mbb{K}_D := \pd \mc{K}_D^{0,0} \pd : \hzz \to \hzz\,.
\end{align}

\begin{theorem} \label{asym:resonance}
Let $D$ be a general smooth domain. When the contrast $\tau$ is large enough,  the scattering resonances for the problem \eqref{model} with the plane wave incidence \eqref{eq:indplawave} exist in the subwavelength regime. Moreover, for any subwavelength resonance $\ww$, as $\tau \to \infty$, there exists some eigenvalue $\lad_i$ 
of the operator $\mbb{K}_D$ with the orthogonal eigenfunctions $\{\vp_j\}_{j = 1}^{n_i} \in \hzz$ such that the following asymptotic holds, 
\begin{align}\label{eq:first_order_app}
    \ww = \frac{1}{\sqrt{\tau \lad_i}} -  \frac{c_{i,j}}{2 \tau \lad_i^2} + O(\tau^{-\frac{3}{2}})\,,
\end{align}
where $c_{i,j}$, $1 \le j \le n_i$, are the eigenvalues of the matrix $C$ defined by 
\begin{align}\label{eq:first_matrix}
    C_{kl} = (\vp_k, \P_d \kb{1}^{d'} \P_d [\vp_l])_D\,, \q \text{for}\ 1 \le k,l \le n_i\,.
\end{align}
\end{theorem}

\begin{proof}
By the generalized Rouch\'{e} theorem \cite{gohberg1990classes,gokhberg1971operator} and a similar argument as in \cite[Section 3.3]{ammari2020mathematical}, we have the existence of the subwavelength resonances $\ww = O(\sqrt{\tau}^{\sss -1})$ in the high contrast regime. To derive the asymptotic formula for a subwavelength resonance $\ww$, we consider the equation 
\begin{align} \label{auxeq:reso_asym}
 \A_{\tau}(\ww d', \ww)[E] = 0 \q \text{with}\ E \in {\bf H}(\ddiv0, D)\ \text{and}\ \norm{E}_D = 1\,.
\end{align}
Acting the projection $\pw$ on the above equation gives, by \eqref{eq:lead_ww}, \eqref{block_exp_1}, and \eqref{block_exp_2},
\begin{align*}
    \Big(\mc{T}_D^{0,0} + O\big(\sqrt{\tau}^{-1}\big)\Big)[\pw E] = O\left(\tau^{-1}\right)\,,
\end{align*}
which, by the invertibility of $\mc{T}_D^{0,0}$ on $W$ in Lemma \ref{lem:op1}, implies $\pw E = O(\tau^{-1})$. We then act the projection $\pd$ on the equation \eqref{auxeq:reso_asym} and find, by \eqref{eq:lead_dd} and \eqref{block_exp_1}, 
\begin{align} \label{auxeq:reso_asym_2}
    (1 - \h{\ww}^2 \mbb{K}_D - \sqrt{\tau}^{-1} \h{\ww}^3 \P_d \kb{1}^{d'} \P_d + O(\tau^{-1}))[E] = O(\tau^{-1})\,,
\end{align}
where $\h{\ww} = \sqrt{\tau}\ww$ is the scaled resonance. It is easy to observe from \eqref{auxeq:reso_asym_2} that the leading-order approximation for $\h{\ww}$ is given by $\h{\ww} = \sqrt{\lad_i}^{\sss -1} + O(\sqrt{\tau}^{\sss -1})$ for some eigenvalue $\lad_i$ of $\mbb{K}_D$. For the higher-order approximation, letting $\{\vp_j\}_{j=1}^{n_i}$ be the orthogonal eigenfunctions of $\mbb{K}_D$ for $\lad_i$, by the standard perturbation theory with the Lyapunov-Schmidt reduction, we need to consider the zeros of the following determinant near $\h{\ww}_i: = \sqrt{\lad_i}^{\sss -1}$:
\begin{align} \label{auxeq:reso_asym_3}
    \det(I - \h{\ww}^2 \lad_i I - \sqrt{\tau}^{-1}\h{\ww}^3 C) = 0\,,
\end{align} 
with the $n_i \t n_i$ matrix $C$ given in \eqref{eq:first_matrix}. It is easy to see that linearizing the above equation \eqref{auxeq:reso_asym_3} at $\h{\ww}_i$ gives
\begin{align*}
     \det(- 2(\h{\ww} - \h{\ww}_i) \h{\ww}_i^{-1} I - \sqrt{\tau}^{-1}\h{\ww}_i^3 C) = 0\,.
\end{align*}
Then, the desired approximation \eqref{eq:first_order_app} for the resonance follows. 
\end{proof}

\subsection{Shape derivatives and asymptotics} \label{sec:shape_deriva}
In order to analyze the scattering effect by the symmetry-broken metasurface, it is necessary to understand the dependence of the subwavelength resonance $\ww$ on the domain $D$. For this, in view of \eqref{eq:first_order_app}, we will compute the shape derivatives of the eigenvalues $\lad_i$ of the operator $\mbb{K}_D$, which is also of independent interest. 

Suppose that $\Omega$ is a smooth convex bounded domain such that $\overline{D} \subset \Omega$ and let $V \in {\bf C}^\infty(\Omega)$ be a feasible deformation field. We define the map $\Phi_t:= I  + t V$ for $t \ge 0$. 
Then $\Phi_0(D) = D$ and $\Phi_t$ gives a family of $C^\infty$-diffeomorphisms for small enough $|t|$. As in \eqref{eq:pert_domain}, we denote by $D_t = \Phi_t(D)$ the deformed domain. Moreover, for any diffeomorphism $\Phi$ on $D$, we introduce the pullback of a function $f$ on $\Phi(D)$: $\Phi^* f = f \circ \Phi$, and the pushforward of a function $g$ on $D$: $\Phi_{*} g = g \circ \Phi^{-1}$. Note that $\Phi_*$ is a smooth bijective map between ${\bf L}^2(D_t)$ and ${\bf L}^2(D)$. With these notions, the eigenvalue problem on the deformed domain $D_t$:
\begin{align*}
\mbb{K}_{D_t}[E] = \lad E\,, \q \text{on}\ {\bf L}^2(D_t)\,,   
\end{align*}
can be equivalently reformulated on the reference domain $D$: 
\begin{align} \label{eq:transformed_eig}
    \w{\mbb{K}}_{D,t}[E] = \lad E\,,\q \text{on}\ {\bf L}^2(D)\,, \q \text{with}\ \w{\mbb{K}}_{D,t} := \Phi_t^* \mbb{K}_{D_t} \Phi_{t,*}\,.
\end{align}
It is clear that the operator $\w{\mbb{K}}_{D,t}$ is the composition of 
\begin{align} \label{def:deformed_op}
    \P_D^t:= \Phi_t^* \P_{{\rm d},D_t} \Phi_{t,*} \q \text{and}\q \w{\mc{K}}_D^t := \Phi_t^* \mc{K}_{D_t}^{0,0} \Phi_{t,*}\,,
\end{align}
namely, $\w{\mbb{K}}_{D,t} = \P_D^t  \w{\mc{K}}_D^t \P_D^t$, where $\P_{{\rm d}, D_t}: {\bf L}^2(D_t) \to {\bf H}_0(\ddiv0,D_t)$ is the projection defined as in \eqref{def:proj_helm} (we add the subscript $D_t$ to emphasize its dependence on the domain).
The main results of this section are as follows. We postpone their proofs to Appendix \ref{app:c} for the sake of clarity.

\begin{lemma} \label{prop:differnetiablity}
The operators $ \P_D^t$, $\w{\mc{K}}_D^t: {\bf L}^2(D) \to {\bf L}^2(D)$, defined in \eqref{def:deformed_op}, are differentiable at $t = 0$ with the derivatives:  
\begin{align} \label{der_proj}
    \frac{d}{d t}\Big|_{t = 0}\P_D^t[E] = (\na V)^T \P_{{\rm d},D}^\perp E - \na w\,, 
\end{align} 
with $w \in H^1(D)$ being a solution to 
\begin{align} \label{der_proj_2}
    (\na w, \na \vp)_D = ( ( - \na V +  \ddiv V) \P_{{\rm d},D} E +  (\na V)^T \P_{{\rm d},D}^{\perp} E, \na \vp)_D \,,\q \forall \vp \in H^1(D)\,,
\end{align}
and 
\begin{align} \label{eq:der_integral}
    \frac{d}{d t}\Big|_{t = 0} \w{\mc{K}}_D^t[E] 
    & = \int_{D} - \ddiv_y \big[G^{0,0}(x, y)  (V(x)-V(y))\big] E(y)\  d y\,.
\end{align}
Here, the projection $\P_{{\rm d},D}^{\perp}$ is the orthogonal complement of $\P_{{\rm d},D}$. 
\end{lemma}

\begin{proposition} \label{lem:shapedev}
Suppose that $\lad_0$ is an eigenvalue of the operator $\mbb{K}_D$ with eigenfunctions $\{\vp_j\}_{j = 1}^m \subset {\bf L}^2(D)$. Then, given a deformation field $V \in {\bf C}^\infty(\Omega)$, there exists $0 < t_0 \ll 1$ such that for any $ |t| \le t_0$, the operator $\mbb{K}_{D_t}$ on the perturbed domain $D_t = \Phi_t(D)$ has exactly $m$ eigenvalues near $\lad_0$, denoted by $\{\lad_{0,j}(t)\}_{j = 1}^m$, that are continuously differentiable in $t$ with the derivatives
$ \dot{\lad}_{0,j}(0)$ 
being eigenvalues of the $m \t m$ Hermitian matrix $\h{K}$ defined by
\begin{align} \label{matrix_K}
    \h{K}_{ij} :=  \frac{d}{d t}\Big|_{t = 0} K_{ij}(t) = \lad_0(\n \dd V \vp_i, \vp_j)_{\p D}\,,
\end{align}
where
\begin{align*}
    K_{ij}(t) := (\vp_i, \w{\mbb{K}}_{D,t} \vp_j)_D \,.
\end{align*}


\end{proposition}

With Theorem \ref{asym:resonance} and Proposition \ref{lem:shapedev} above, we can easily derive the asymptotic expansions of the subwavelength resonances for the symmetry-broken geometry, which will be useful in the next section. For ease of exposition, we first introduce the following matrices under the assumptions \eqref{1}-\eqref{3} with the incident direction $\di$ given in \eqref{eq:incidece_direction}. 
Let the fields $\vp_1 \in {\bf H}_{0,ant}(\ddiv0,D)$ and $\vp_2 \in {\bf H}_{0,sym}(\ddiv0,D)$ be given in the assumption \eqref{2}, and $V$ be a deformation field \eqref{def:defor_field}. By a slight abuse of notation, we define the $2 \t 2$ symmetric matrix ${\rm C}^0$ by
\begin{align} \label{def:matrix_1}
{\rm C}^0_{ij} := (\nu \dd V \vp_i, \vp_j)_{\p D}\,,    
\end{align}
and the $2 \t 2$ matrix ${\rm C}^{1,\alpha_0}$ as in \eqref{eq:first_matrix} by
\begin{align} \label{def:matrix_2}
{\rm C}^{1,\alpha_0}_{ij} := (\vp_i, \P_d \kb{1}^{\epsilon \alpha_0} \P_d [\vp_j])_D\,.    
\end{align}
Note from \eqref{assp:sym} that the unit outer normal vector $\nu$ to $\p D$ is antisymmetric in the sense of \eqref{eq:asyml2}. Thanks to the symmetry properties of $V$ and $\vp_i$, we have 
\begin{align} \label{matrix_1}
    {\rm C}^0 = \mm 0 & c_0 \\ c_0 & 0 \nn\,,\q \text{with}\ c_0 := (\n\dd V \vp_1, \vp_2)_{\p D}\,.
\end{align}
Moreover, recalling the kernel $ G_{1}^{\epsilon \alpha_0}$ \eqref{green_1} for the operator $\P_d \kb{1}^{\epsilon \alpha_0} \P_d$, we write it as 
\begin{align*}
    G_{1}^{\epsilon \alpha_0}(x) = iK_1(x) + i  \epsilon^2 K_2(x) + i \epsilon  K_3(x)\,,
\end{align*}
with 
\begin{align*}
     K_1(x) = - \frac{ d_3|x_3|^2}{4}\,, \q  K_2(x) = -  \frac{ (\alpha_0 \dd x')^2}{4 d_3} \,,\q  K_3(x) = - \alpha_0\dd \left(\frac{ |x_3| x'}{2} - i  g_1(x) \right) \,.
\end{align*}
We remark that all $K_i$ are real functions, since $g_1$ is purely imaginary. Again, by the symmetry of $\vp_i$ and $g_1$ in \eqref{eq:symm_green}, it follows that 
\begin{align} \label{matrix_2}
    {\rm C}^{1,\alpha_0} = i \mm c_{1,1} & 0 \\ 0 & c_{1,2}\nn + i \epsilon^2 \mm c_{2,1} & 0 \\ 0 & c_{2,2} \nn + i \epsilon \mm 0 & c_3 \\ - c_3 & 0 \nn\,,
\end{align}
with real numbers:
\begin{align} \label{element_1}
    c_{s,l} = \int_D \int_D K_s(x-y) \vp_l(x) \dd \vp_l(y) \ dx dy\,, \q s= 1,2\,, \ l =1, 2\,,
\end{align}
and 
\begin{align} \label{element_2}
c_3 = \int_D \int_D K_3(x-y) \vp_2(x) \dd \vp_1(y)\ dx dy\,.
\end{align}
We claim that $c_{1,1}$ is zero and $c_{1,2} \ge 0$. Indeed, if the assumptions \eqref{1} and \eqref{2} hold, Theorem \ref{asym:resonance} for the case of the normal incidence and the symmetric domain $D$ shows that there are two characteristic values $\ww_1$ and $\ww_2$ of $\mc{A}_{\tau}(0,\ww)$ near $\sqrt{\tau \lad_0}^{\sss -1}$ with asymptotics: for $j = 1,2$, as $\tau \to \infty$, 
\begin{align}\label{eq:first_order_app_2}
    \ww_j = \frac{1}{\sqrt{\tau \lad_0}} - \frac{i c_{1,j}}{2 \tau \lad_0^2} + O(\tau^{-\frac{3}{2}})\,,
\end{align}
where $c_{1,j}$, $j = 1,2$, is defined in \eqref{element_1}. 
Noting that the resonance $\ww_1 = \ww_*$ corresponds to the antisymmetric function, the assumption \eqref{1} enforces the vanishing of the imaginary part of $\ww_1$, that is, $c_{1,1} = 0$ and 
\begin{align} \label{aym_wstar}
    \ww_* = \frac{1}{\sqrt{\tau \lad_0}} + O(\tau^{-\frac{3}{2}})\,.
\end{align}
Moreover, it is known that the scattering resonances exist in the lower half-plane \cite{ammari2020mathematical,dyatlov2019mathematical}, which implies $c_{1,2} \ge 0$.

\begin{corollary} \label{coro:appro_reso}
Given a deformation field $V$ in \eqref{def:defor_field}, let $\di$ and $D_t$ be the perturbed incident direction and domain defined in \eqref{eq:incidece_direction} and \eqref{eq:pert_domain}. Under the assumptions \eqref{1} and \eqref{2}, there exist two subwavelength resonances $\ww_1$ and $\ww_2$  near $\sqrt{\tau \lad_0}^{\sss -1}$ for the scattering problem \eqref{model} with $\alpha = \epsilon \alpha_0$, which satisfy the asymptotic expansions: for $j = 1,2$,
\begin{align*}
\ww_j = \frac{1}{\sqrt{\tau}} \left(\frac{1}{\sqrt{\lad_0}} - \frac{\mu_{0,j}}{2}\right) + O(\tau^{-\frac{3}{2}} + |t|^2 \tau^{-\frac{1}{2}})\,, \q \text{as}\ \tau \to \infty\,, \ |t| \to 0\,,
\end{align*}
with $\mu_{0,j}$ being the eigenvalues of the following matrix:
\begin{align} \label{matrix_1st}
 {\rm M} =   t \lad_0^{-\frac{1}{2}}  \mm 0 & c_0 \\ c_0 & 0 \nn + \tau^{-\frac{1}{2}} \lad_0^{-2} \left(i \mm 0 & 0 \\ 0 & c_{1,2}\nn + i \epsilon \mm 0 & c_3 \\ - c_3 & 0 \nn\right) + O(\epsilon^2 \tau^{-\frac{1}{2}})\,,
\end{align}
where $c_{1,2} \ge 0$ and $c_0, c_3 \in \R $ are given in \eqref{matrix_1}, \eqref{element_1} and \eqref{element_2}. 
\end{corollary}

\begin{proof}
We only provide a sketch of the proof, as it is almost the same as those of Theorem \ref{asym:resonance} and Proposition \ref{lem:shapedev}. Let $\h{\ww} = \sqrt{\tau}\ww$ be the scaled resonance. We consider the characteristic values of $(1 - \tau \mc{T}_{D_t}^{\ww d', \ww})[E] = 0$ near $\sqrt{\tau}^{\sss -1}\h{\ww}_0$
with $\norm{E}_{D} = 1$ and $\h{\ww}_0 = \sqrt{\lad_0}^{\sss -1}$.  Similarly to the proof of Theorem \ref{asym:resonance}, there holds $\pw E = O(\tau^{-1})$, and it suffices to solve, after mapping the problem to the reference domain $D$ as in \eqref{eq:transformed_eig},  
\begin{align} \label{auxeq:reso_asym_22}
    (1 - \h{\ww}^2 \w{\mbb{K}}_{D,t} - \sqrt{\tau}^{-1} \h{\ww}^3 \P_D^t \Phi_t^* \kb{1}^{d'} \Phi_{t,*}  \P_D^t )[\Phi_t^* E] = O(\tau^{-1})\,.
\end{align}
We linearize the above equation \eqref{auxeq:reso_asym_22} at $\h{\ww} = \h{\ww}_0$ and $t = 0$, and find that, up to an error term of $O(\tau^{-1}) + O(|t|^2)$, the scaled resonance $\h{\ww}$ near $\h{\ww}_0$ satisfies 
\begin{align*}
    \det \left( - 2(\h{\ww} - \h{\ww}_0)\h{\ww}_0^{-1} {\rm I} - t {\rm C}^0 - \sqrt{\tau}^{-1}\h{\ww}_0^3 {\rm C}^{1,\alpha_0} \right) = 0\,, 
\end{align*}
which readily completes the proof, by a direct computation. 
\end{proof}

\subsection{Analysis of the reflection energy}
In this section, we investigate the Fano-type transmission and reflection anomalies. For this, we will first derive the leading-order approximation for the scattered polarization vectors $\po_{\pm}$ defined in \eqref{eq:transreflec_polar}. Let $\vp_1$ and $\vp_2$ be the antisymmetric and symmetric functions in the assumption \eqref{2}, respectively. We introduce the following two $3 \t 3$ real matrices for later use:
\begin{align}\label{matrix_3}
    {\rm A}^{ant} = \int_D x \otimes \vp_1 \ dx\,, \q {\rm A}^{sym} = \int_D x \otimes \vp_2 \ dx\,.
\end{align}
By symmetry, it is easy to see that ${\rm A}^{ant}$ and ${\rm A}^{sym}$ are of the block form:
\begin{align*}
    {\rm A}^{ant} = \mm {\bf A}_1 & 0 \\ 0 & a_2\nn \,, \q  {\rm A}^{sym} = \mm 0 & {\bf b}_1 \\ {\bf b}_2^T & 0 \nn\,,
\end{align*}
where ${\bf b}_1, {\bf b}_2 \in \R^2$ and ${\bf A}_1$ is the $2 \t 2$ matrix. The following proposition characterizes the solution of the volume integral equation for the incident frequency $\ww$ near the embedded eigenvalue $\ww_*$:
\begin{align}\label{target}
    (1 - \tau \mc{T}_{D_t}^{\ww d', \ww})[E] = E^i\,.
\end{align}

\begin{theorem} \label{thm:solappro}
Suppose that the assumptions \eqref{1}-\eqref{3} hold. Define $\h{\ww}_{0,j} := \h{\ww}_0 - \mu_{0,j}/2$ for $j = 1, 2$ with $\mu_{0,j}$ given in Corollary \ref{coro:appro_reso}. There exists a constant $C_A > 0$ such that when $|t| + \tau^{-\frac{1}{2}} \le C_A |\h{\ww} - \h{\ww}_0|$, the solution $E \in {\bf H}(\ddiv0, D_t)$ to the equation \eqref{target} has the asymptotic expansion uniformly in $\h{\ww}$ near $\h{\ww}_0$: 
\begin{align}  \label{main_solest}
 \left\{
    \begin{aligned}
    & \Phi_t^* \P_{{\rm w}, D_t} E = - \tau^{-1} (\mc{T}_D^{0,0})^{-1}[\po^i] + O\Big(\tau^{-1}|t| + \frac{\tau^{-\frac{3}{2}}}{|\h{\ww} - \h{\ww}_0|}\Big)\,, \\
    & \Phi_t^* \P_{{\rm d}, D_t} E = q_1 \vp_1 + q_2 \vp_2 +  O\Big(\tau^{-\frac{1}{2}} + \frac{\tau^{-\frac{1}{2}}|t| + \tau^{-1} }{|\h{\ww} - \h{\ww}_0|}\Big)\,,
    \end{aligned} \right.  \q \text{as}\ \tau \to \infty\,,\ |t| \to 0\,,
\end{align}
with $q_1, q_2 \in \C$ given by 
\begin{align} \label{solq1q2}
    \mm q_1 \\ q_2\nn = \frac{1 - \h{\ww}^2 \lad_0}{h}  \left(i \tau^{-\frac{1}{2}} \h{\ww} \mm f_1 \\ f_2 \nn + O\Big(\tau^{-1} + \tau^{-\frac{1}{2}} |t| + \frac{\tau^{-\frac{1}{2}}(|t|^2 + \tau^{-1}) }{|\h{\ww} - \h{\ww}_0|}\Big) \right)\,,
\end{align}
where $h$ is an analytic function in $\h{\ww}$ defined by 
\begin{align} \label{app_det}
    h(\h{\ww}) = r(\h{\ww})(\h{\ww} - \h{\ww}_{0,1} + O(\tau^{-1} + |t|^2))(\h{\ww} - \h{\ww}_{0,2} + O(\tau^{-1} + |t|^2))\,,
\end{align}
with $r(\h{\ww})$ being invertible and analytic in $\h{\ww}$ with $|r(\h{\ww})| = 1$, 
and the real numbers $f_1$ and $f_2$ are given by 
\begin{align} \label{eq:rhs}
    f_1: = \di \dd {\rm A}^{ant} \po^i \,,\q f_2 : = \di \dd {\rm A}^{sym} \po^i\,.
\end{align}
\end{theorem}

\begin{proof}
Noting the incident frequency $\ww = O(\sqrt{\tau}^{\sss -1})$ and $\int_{D_t} \vp \ dx = 0$  for any $\vp \in {\bf H}_0(\ddiv0, D_t)$, we have the following asymptotic expansions for the incident wave:
\begin{align} \label{exp:incident}
    \P_{{\rm w},D_t} E^i = \po^i + O(\tau^{-\frac{1}{2}})\,, \q \P_{{\rm d}, D_t} E^i & = i \ww \P_{{\rm d}, D_t}[\po^i \di \dd x] + O(\tau^{-1})\,,
\end{align}
which implies, by the pullback $\Phi_t^*$,
\begin{align*}
    \Phi_t^*\P_{{\rm d}, D_t} E^i = i \ww \P_{{\rm d}, D} [\po^i \di \dd x] + O(\tau^{-\frac{1}{2}}|t| + \tau^{-1})\,.
\end{align*}
It follows that, by \eqref{matrix_3} and \eqref{eq:rhs},
\begin{align} \label{auxest_sol3}
    (\vp_j, \Phi_t^*\P_{{\rm d}, D_t} E^i)_{D} & =   i \ww (\vp_j, \po^i \di \dd x)_D + O(\tau^{-\frac{1}{2}}|t| + \tau^{-1}) \notag \\
    & = i \ww f_j + O(\tau^{-\frac{1}{2}}|t| + \tau^{-1})\,.
\end{align}
Following the analysis in \cite[Section 4]{ammari2020mathematical}, we reformulate the equation \eqref{target} on the reference domain $D$ with a matrix form, by the projections and the pullback $\Phi_t^*$:
\begin{align} \label{matrix_eq}
    \mbb{A}(t, \tau,\ww) \mm \Phi_t^* \P_{{\rm d}, D_t} E \\ \Phi_t^* \P_{{\rm w}, D_t} E  \nn = \mm \Phi_t^* \P_{{\rm d}, D_t} E^i \\ \tau^{-1} \Phi_t^* \P_{{\rm w}, D_t} E^i \nn\,, 
\end{align}
with 
\begin{align*}
    \mbb{A}(t, \tau,\ww): = \mm 1 - \tau \Phi_t^* \P_{{\rm d}, D_t} \mc{T}_{D_t}^{\ww d', \ww} \P_{{\rm d}, D_t} \Phi_{t,*} &  - \tau  \Phi_t^* \P_{{\rm d}, D_t} \mc{T}_{D_t}^{\ww d', \ww} \P_{{\rm w}, D_t} \Phi_{t,*} \\
    - \Phi_t^* \P_{{\rm w}, D_t} \mc{T}_{D_t}^{\ww d', \ww} \P_{{\rm d}, D_t} \Phi_{t,*}  & \tau^{-1} - \Phi_t^* \P_{{\rm w}, D_t} \mc{T}_{D_t}^{\ww d', \ww} \P_{{\rm w}, D_t} \Phi_{t,*} \nn\,.
\end{align*}
Similarly to Lemma \ref{prop:differnetiablity}, for $n \ge 1$ and ${\rm s,t} = {\rm d, w}$, the operator $\Phi_t^* \P_{{\rm s}, D_t} \mc{T}_{D_t, n}^{d'} \P_{{\rm t}, D_t} \Phi_{t,*}$ is differentiable at $t = 0$, and there holds 
\begin{align*}
 \Phi_t^* \P_{{\rm s}, D_t} \mc{T}_{D_t, n}^{d'} \P_{{\rm t}, D_t} \Phi_{t,*} =    \P_{{\rm s}, D} \mc{T}_{D_t, n}^{d'} \P_{{\rm t}, D} + O(|t|)\,,
\end{align*}
where the error term is measured in the operator norm. Then, by expansions \eqref{block_exp_1} and \eqref{block_exp_2}, we find 
\begin{align} \label{asym_matrix}
     \mbb{A}(t, \tau,\ww) = \mm 1 -  \h{\ww}^2 \mbb{K}_D  & - \h{\ww}^2 \pd \mc{K}_{D,0}^{d'} \pw \\
     0 & - \mc{T}_D^{0,0} 
      \nn + \mm O(|t| + \tau^{-\frac{1}{2}}) & O(|t| + \tau^{-\frac{1}{2}}) \\ 
      O(\tau^{-1}) & O(|t| + \tau^{-\frac{1}{2}})\nn \,.
\end{align}
Hence, for $\h{\ww}$ near $\h{\ww}_0$, the Neumann series expansion shows that when $|t| + \tau^{-\frac{1}{2}} = O(|\h{\ww} - \h{\ww}_0|)$ is suitably small, there holds $\norm{\mbb{A}(t, \tau,\ww)} \lesssim |\h{\ww} - \h{\ww}_0|^{-1}$. We also see from \eqref{exp:incident} that $   [\P_{{\rm d},D_t} E^i, \tau^{-1} \P_{{\rm w}, D_t} E^i] = O(\sqrt{\tau}^{\sss -1})$. It follows that the solution $E \in {\bf H}_0(\ddiv 0, D_t)$ to \eqref{target} satisfies
\begin{align} \label{rough_est}
    E = O(|\h{\ww} - \h{\ww}_0|^{-1} \tau^{-\frac{1}{2}})\,.
\end{align}

To obtain \eqref{main_solest}, we consider the second component of the equation \eqref{matrix_eq}. By estimates \eqref{exp:incident}, \eqref{asym_matrix} and \eqref{rough_est}, we have  
\begin{align} \label{asym_pw}
    \Phi_t^* \P_{{\rm w}, D_t} E = - \tau^{-1} (\mc{T}_D^{0,0})^{-1}[\po^i] + O\Big(\tau^{-1}|t| + \frac{\tau^{-\frac{3}{2}}}{|\h{\ww} - \h{\ww}_0|}\Big)\,.
\end{align}
Then, we consider the first component of \eqref{matrix_eq}. By the assumption \eqref{2}, without loss of generality, let the field $\Phi_t^* \P_{{\rm d}, D_t} E \in {\bf L}^2(D)$ have the following ansatz: 
\begin{align} \label{ansatz_1}
     \Phi_t^* \P_{{\rm w}, D_t} E = q_1 \vp_1 + q_2 \vp_2 + \phi\,,\q \text{with}\ \phi \perp \vp_i\,,
\end{align}
where $\phi \in {\bf L}^2(D)$ and $q_1, q_2 \in \C$. Substituting \eqref{ansatz_1} into \eqref{matrix_eq}, and using the asymptotic expansions \eqref{block_exp_1} and \eqref{exp:incident} and the estimate \eqref{asym_pw}, we can derive 
\begin{align}  \label{eqpd_1}
 &\left(1 - \h{\ww}^2 \w{\mbb{K}}_{D,t} - \tau^{-\frac{1}{2}} \h{\ww}^3 \P_D^t \Phi_t^* \kb{1}^{d'} \Phi_{t,*}  \P_D^t + O(\tau^{-1}) \right)[q_1 \vp_1 + q_2 \vp_2 + \phi] \notag \\
= & i \ww \P_{{\rm d}, D} [\po^i \di \dd x] + O(\tau^{-\frac{1}{2}}|t| + \tau^{-1}) + O\Big( \frac{\tau^{-\frac{3}{2}}}{|\h{\ww} - \h{\ww}_0|}\Big)\,.
\end{align}
We take the inner product between the above equation \eqref{eqpd_1} and $\phi$, and obtain 
\begin{align} \label{auxest_sol1}
    \left|\left(\phi, (1 - \h{\ww}^2 \mbb{K}_{D} + O(|t|) + O(\tau^{-\frac{1}{2}}))[q_1 \vp_1 + q_2 \vp_2 + \phi]\right)_D\right| \lesssim \Big(\tau^{-\frac{1}{2}} +  \frac{\tau^{-\frac{3}{2}}}{|\h{\ww} - \h{\ww}_0|}\Big)\norm{\phi}_D\,.
\end{align} 
Noting from \eqref{rough_est} that $\norm{q_1 \vp_1 + q_2 \vp_2 + \phi}_D = O(|\h{\ww} - \h{\ww}_0|^{-1} \tau^{-\frac{1}{2}})$, and from the ansatz \eqref{ansatz_1} and the assumption \eqref{2} that, for $\h{\ww}$ near $\sqrt{\lad_0}^{\sss -1}$, 
\begin{align*}
  \norm{\phi}_D^2 \lesssim  \left|\left(\phi, (1 - \h{\ww}^2 \mbb{K}_{D}))[q_1 \vp_1 + q_2 \vp_2 + \phi]\right)_D\right|\,,
\end{align*}
we estimate 
\begin{align*}
    0 <  \norm{\phi}_D^2 - O\Big( \frac{\tau^{-\frac{1}{2}}(|t| + \tau^{-\frac{1}{2}})  }{|\h{\ww} - \h{\ww}_0|}\Big)\norm{\phi}_D \lesssim \left|\left(\phi, \big(1 - \h{\ww}^2 \mbb{K}_{D} + O\big(|t| + \tau^{-\frac{1}{2}}\big) \big)[q_1 \vp_1 + q_2 \vp_2 + \phi]\right)_D\right|\,,
\end{align*}
which, along with \eqref{auxest_sol1}, gives 
\begin{align} \label{est_phi}
   \phi = O\Big(\tau^{-\frac{1}{2}} + \frac{\tau^{-\frac{1}{2}}|t| + \tau^{-1} }{|\h{\ww} - \h{\ww}_0|}\Big)\,.
\end{align}
We proceed to take the inner product between \eqref{eqpd_1} and $\vp_i$. We first have,  by Proposition \ref{lem:shapedev},
\begin{align*}
 \big(\vp_i,    (1 - \h{\ww}^2 \w{\mbb{K}}_{D,t})[\vp_j] \big)_D = \big(I - \h{\ww}^2 \lad_0 (I + t {\rm C}^{0}))_{ij} + O(|t|^2)\,.
\end{align*}
Recalling the definition of ${\rm C}^{1,\alpha_0}$ in \eqref{def:matrix_2}, there holds 
\begin{align*}
   (\vp_i, \tau^{-\frac{1}{2}} \h{\ww}^3 \P_D^t \Phi_t^* \kb{1}^{d'} \Phi_{t,*}  \P_D^t [\vp_j])_D = \tau^{-\frac{1}{2}} \h{\ww}^3 {\rm C}^{1,\alpha_0}_{ij} + O(\tau^{-\frac{1}{2}}|t|)\,.
\end{align*}
It is also easy to see, by \eqref{est_phi},
\begin{align} \label{auxest_sol2}
  & \left(\vp_i, \big(1 - \h{\ww}^2 \w{\mbb{K}}_{D,t} - \tau^{-\frac{1}{2}} \h{\ww}^3 \P_D^t \Phi_t^* \kb{1}^{d'} \Phi_{t,*}  \P_D^t + O(\tau^{-1}) \big)[\phi]\right)_D \notag\\
 = & O\Big(\tau^{-1} + \tau^{-\frac{1}{2}} |t| + \frac{\tau^{-\frac{1}{2}}(|t|^2 + \tau^{-1}) }{|\h{\ww} - \h{\ww}_0|}\Big)\,.
\end{align}
Combining these facts together, it follows that $[q_1,q_2]$ satisfies the following matrix equation:
\begin{align} \label{appsol_matrix}
   & \Big(I - \h{\ww}^2 \lad_0 (I + t {\rm C}^{0}) - \tau^{-\frac{1}{2}}\h{\ww}^3 {\rm C}^{1,\alpha_0} + O(\tau^{-1} + |t|^2) \Big) \mm q_1 \\ q_2 \nn \notag \\
 = & i \tau^{-\frac{1}{2}} \h{\ww} \mm f_1 \\ f_2 \nn + O\Big(\tau^{-1} + \tau^{-\frac{1}{2}} |t| + \frac{\tau^{-\frac{1}{2}}(|t|^2 + \tau^{-1}) }{|\h{\ww} - \h{\ww}_0|}\Big)\,,
\end{align}
where the right side is from \eqref{auxest_sol3} and \eqref{auxest_sol2}.

We write $Z(\h{\ww})$ for the coefficient matrix of the equation \eqref{appsol_matrix}. By Rouch\'{e}'s theorem, we know that there are two zeros of the analytic function $\det(Z)$ near $\h{\ww}_0$. We denote them by $z_1, z_2 \in \C$ and have, for $\h{\ww} \in \C$ near $\h{\ww}_0$, 
\begin{align*}
\det(Z(\h{\ww})) = r(\h{\ww})(\h{\ww} - z_1)(\h{\ww} - z_2)\,,
\end{align*}
where $r(\h{\ww})$ is an invertible analytic function near $\h{\ww}_0$ with modulus of order one. Then, it follows from the analysis in Corollary \ref{coro:appro_reso} that $z_i$ can be approximated by $z_i = \h{\ww}_0 - \mu_{0,j}/2 + O(\tau^{-1} + |t|^2)$. Therefore, the formula \eqref{solq1q2} holds by the inverse formula for $2 \t 2$ matrices.  The proof is complete by estimates \eqref{asym_pw}, \eqref{ansatz_1}, and \eqref{est_phi}. 
\end{proof}

With the help of Theorem \ref{thm:solappro} above, we are now ready to derive the asymptotic expansion for the scattered polarization vectors \eqref{eq:transreflec_polar}. The estimate \eqref{main_solest} has shown that, for the solution $E$ to \eqref{target},
\begin{align*}
 & \P_{{\rm w}, D_t} E =  O\Big(\tau^{-1} + \frac{\tau^{-\frac{3}{2}}}{|\h{\ww} - \h{\ww}_0|}\Big)\,,\q 
     \P_{{\rm d}, D_t} E = O\Big(\frac{\tau^{-\frac{1}{2}}}{|\h{\ww} - \h{\ww}_0|}\Big)\,.
\end{align*}
By the above estimate and Taylor expansion with $\ww = \sqrt{\tau}^{-1} \h{\ww}$, we find 
\begin{align*}
     \po_{\pm}(\di,D_t) & = \tau^{\frac{1}{2}} \frac{i \h{\ww}}{2 d_3} \big({\rm I} - \di_\pm \otimes \di_\pm \big) \int_{D_t} \Big(\P_{{\rm w}, D_t} E - i \tau^{-\frac{1}{2}}\h{\ww}\di_{\pm}\dd y \P_{{\rm d}, D_t} E    \Big)\ dy + O\Big( \frac{\tau^{-1}}{|\h{\ww} - \h{\ww}_0|} \Big) \\
     & =  \frac{\h{\ww}^2}{2 d_3} \big({\rm I} - \di_\pm \otimes \di_\pm \big) \int_{D_t} \di_{\pm}\dd y \P_{{\rm d}, D_t} E\ dy + O\Big(\tau^{-\frac{1}{2}} + \frac{\tau^{-1}}{|\h{\ww} - \h{\ww}_0|} \Big)\,.
\end{align*}
Then, by a change of variables and \eqref{main_solest}, as well as \eqref{matrix_3},  we further calculate
\begin{align} \label{cal_1}
     \po_{\pm}(\di,D_t)
     & =  \frac{\h{\ww}^2}{2 d_3} \big({\rm I} - \di_\pm \otimes \di_\pm \big) \int_{D} (\di_{\pm}\dd \Phi_t)  (\Phi_t^* \P_{{\rm d}, D_t} E)\ d \Phi_t + O\Big(\tau^{-\frac{1}{2}} + \frac{\tau^{-1}}{|\h{\ww} - \h{\ww}_0|} \Big) \notag \\
      & =  \frac{\h{\ww}^2}{2 d_3} \big({\rm I} - \di_\pm \otimes \di_\pm \big) \int_{D} \di_{\pm}\dd y  (q_1 \vp_1 + q_2 \vp_2)  \ dy + O\Big(\tau^{-\frac{1}{2}} + \frac{\tau^{-\frac{1}{2}}(\tau^{-\frac{1}{2}} + |t|)}{|\h{\ww} - \h{\ww}_0|} \Big) \notag \\
        & =  \frac{\h{\ww}^2}{2 d_3} \big(q_1  {\bf g}^{\pm}_1 + q_2  {\bf g}^{\pm}_2 \big) + O\Big(\tau^{-\frac{1}{2}} + \frac{\tau^{-\frac{1}{2}}(\tau^{-\frac{1}{2}} + |t|)}{|\h{\ww} - \h{\ww}_0|} \Big)\,,
\end{align}
where 
\begin{align*}
    {\bf g}^{\pm}_1 = ({\rm A}^{ant})^T \di_{\pm} - \l \di_{\pm}, {\rm A}^{ant} \di_{\pm} \r \di_{\pm}\,,
\end{align*}
and 
\begin{align*}
    {\bf g}^{\pm}_2 = ({\rm A}^{sym})^T \di_{\pm} - \l \di_{\pm}, {\rm A}^{sym} \di_{\pm} \r \di_{\pm}\,.
\end{align*}
Recall that the perturbed incident direction $\di$ is of the form \eqref{eq:incidece_direction}. A direct computation gives 
\begin{align} \label{cal_2}
     {\bf g}^{\pm}_1 = \epsilon \mm ({\bf A}_1^T - a_2) \alpha_0 \\ 0 \nn + O(\epsilon^2)\,,\q {\bf g}^{\pm}_2 = \mm \pm {\bf b}_2 \\ - \epsilon {\bf b}_2 \dd \alpha_0 \nn + O(\epsilon^2)\,,\q \text{as}\ \epsilon \to 0\,.
\end{align}
Similarly, thanks to $\po \dd \di = 0$, we have
\begin{align} \label{cal_3}
    f_1 = O(\epsilon)\,,\q f_2 = O(1)\,.
\end{align}
For simplicity, we introduce the small parameter:
\begin{align*}
    \d = |t| + \tau^{-\frac{1}{2}}\,,
\end{align*}
and consider the regime:
\begin{align} \label{regime_fano}
    \d = o(|\h{\ww} - \h{\ww}_0|)\,,\q \text{as}\ \h{\ww} \to \h{\ww}_0\,.
\end{align}
By the estimate \eqref{main_solest} and the expansions \eqref{cal_1}, \eqref{cal_2} and \eqref{cal_3}, we readily have 
\begin{align} \label{app_polar}
       \po_{\pm}(\di,D_t)
         =  \frac{ i \tau^{-\frac{1}{2}} \h{\ww}^3}{2 d_3} \frac{1 - \h{\ww}^2 \lad_0}{h}  & f_2 \mm \pm {\bf b}_2 \\ - \epsilon {\bf b}_2 \dd \alpha_0 \nn  + O(\d + \epsilon^2)\,,
\end{align}
since $ (1 - \h{\ww}^2 \lad_0)/h = O(|\h{\ww} - \h{\ww}_0|^{-1})$ holds for $\h{\ww}$ near $\h{\ww}_0$. As a corollary, we can approximate the reflection energy $|\po^r|^2$ as follows. 
\begin{corollary} \label{coro:approx_energy}
Suppose that the assumptions in Theorem \ref{thm:solappro} hold. In the regime \eqref{regime_fano},  the reflection energy $|\po^r|^2$ has the following asymptotics: for $\h{\ww} \in \R$ near $\h{\ww}_0$, 
\begin{align*}
    |\po^r|^2 = \tau^{-1}\frac{|1 - \h{\ww}^2 \lad_0|^2}{|h|^2} p_0 + o(\d + \epsilon^2)\,, \q  \text{as}\ \d, \epsilon \to 0\,,
\end{align*}
with the real function $p_0$ in $\h{\ww}$:
\begin{align*}
    p_0(\h{\ww}) := \frac{|\h{\ww}|^6}{4(1 - \epsilon^2)} f_2^2 \big( |{\bf b}_2|^2 + \epsilon^2 |{\bf b}_2 \dd \alpha_0|^2 \big),
\end{align*}
where $h$ and $f_2$ are given in \eqref{app_det} and \eqref{eq:rhs}, respectively. 
\end{corollary}

We next show our main claim that for certain parameters, the Fano-type resonance can happen for the symmetry-broken metasurface when the incident frequency $\ww$ is near the embedded eigenvalue $\ww_*$. \mb{Recalling the classical Fano formula \cite{shipman2010resonant} (see also Remark \ref{remfano} below)} and the scaling $\ww = \sqrt{\tau}^{\sss - 1}\h{\ww}$, it suffices to prove that for some fixed large $\tau$ and small $t, \epsilon$, there holds
\begin{align} \label{fano_ansa}
   \tau^{-1} \frac{|1 - \h{\ww}^2 \lad_0|^2}{|h|^2}  \approx C \frac{|\theta + (\ww - \ww_*)|^2}{\gamma + |\ww - \ww_*|^2}\,, \q \text{as} \ \ww \to \ww_*\,,
\end{align}
for some constant $C$ and parameters $\theta \in \R$ and $\gamma > 0$. We first observe from Corollary \ref{coro:appro_reso} that, for $j = 1, 2$,
\begin{align} \label{cal_4}
    \h{\ww}_{0,j} = \frac{1}{\sqrt{\lad_0}} - i \eta_j\,,
\end{align}
where $\eta_j \ge 0$ are given by  
\begin{align*}
    \eta_{1} = \frac{i \tau^{-\frac{1}{2}}\lad_0^{-2} c_{1,2} - \sqrt{- \tau^{-1}\lad_0^{-4} c_{1,2}^2 + 4 (t^2 \lad_0^{-1} c_0^2 + \tau^{-1}\lad_0^{-4} \epsilon^2 c_3^2)}}{4} + O(\epsilon^2 \tau^{-\frac{1}{2}})\,,
\end{align*}
and 
\begin{align*}
    \eta_{2} = \frac{i \tau^{-\frac{1}{2}}\lad_0^{-2} c_{1,2} + \sqrt{- \tau^{-1}\lad_0^{-4} c_{1,2}^2 + 4 (t^2 \lad_0^{-1} c_0^2 + \tau^{-1}\lad_0^{-4} \epsilon^2 c_3^2)}}{4} + O(\epsilon^2 \tau^{-\frac{1}{2}})\,.
\end{align*}
Hence, \mb{when the symmetry is broken by purely perturbing the normal incidence (i.e., $t = 0$)}, there holds
\begin{align} \label{cal_5}
    \eta_1 = O(\tau^{-\frac{1}{2}}\epsilon^2)\,, \q \eta_2 = O(\tau^{-\frac{1}{2}})\,,
\end{align}
that is, $\h{\ww}_{0,1}$ has a much smaller imaginary part than the one of $\h{\ww}_{0,2}$. By definition \eqref{app_det} of $h$, we can compute 
\begin{align} \label{cal_6}
\tau^{-1} \frac{|1 - \h{\ww}^2 \lad_0|^2}{|h|^2} =  \frac{\lad_0^2  |(\h{\ww}_0 + \h{\ww})|^2}{|r(\h{\ww})|^2} \frac{|(\ww - \ww_*) + (\ww_* - \tau^{-\frac{1}{2}} \h{\ww}_0)|^2}{|\h{\ww} - \h{\ww}_{0,1} + O(\tau^{-1})| ^2 |\h{\ww} - \h{\ww}_{0,2} + O(\tau^{-1})|^2}\,.
\end{align}
Then, by estimates \eqref{aym_wstar} and \eqref{cal_4}, we see 
\begin{align*}
\tau^{-\frac{1}{2}}|\h{\ww} - \h{\ww}_{0,i} + O(\tau^{-1})| = |\ww - \ww_* + i \tau^{-\frac{1}{2}} \eta_i + O(\tau^{-\frac{3}{2}})|\,,
\end{align*}
which implies, by $\eta_2 = O(\tau^{-\frac{1}{2}})$ in \eqref{cal_5},
\begin{align} \label{cal_7}
    &\frac{1}{|\h{\ww} - \h{\ww}_{0,1} + O(\tau^{-1})| ^2 |\h{\ww} - \h{\ww}_{0,2} + O(\tau^{-1})|^2} \notag \\
    = & \frac{\tau^{-2}}{|\ww - \ww_* + i \tau^{-\frac{1}{2}} \eta_1 + O(\tau^{-\frac{3}{2}})|^2 |\ww - \ww_* + i \tau^{-\frac{1}{2}} \eta_2 + O(\tau^{-\frac{3}{2}})|^2} \notag \\
     = & \frac{1}{|\ww - \ww_* + i \tau^{-\frac{1}{2}} \eta_1 + O(\tau^{-\frac{3}{2}})|^2 |\tau(\ww - \ww_*) + i \tau^{\frac{1}{2}} \eta_2 + O(\tau^{-\frac{1}{2}})|^2}  \notag \\
     = & O(1) \frac{1}{|\ww - \ww_* + i \tau^{-\frac{1}{2}} \eta_1 + O(\tau^{-\frac{3}{2}})|^2}   \qq \qq 
      \text{as} \ \ww \to \ww_*\,.
\end{align}
Therefore, it follows from \eqref{cal_6} and \eqref{cal_7} that for real $\ww$ approximating $\ww_*$,
\begin{align} \label{fano_express}
    \tau^{-1} \frac{|1 - \h{\ww}^2 \lad_0|^2}{|h|^2} & = O(1) \frac{|(\ww - \ww_*) + (\ww_* - \tau^{-\frac{1}{2}} \h{\ww}_0)|^2}{|\ww - \ww_* + i \tau^{-\frac{1}{2}} \eta_1 + O(\tau^{-\frac{3}{2}})|^2} \notag \\
    & \approx  O(1) \frac{|(\ww - \ww_*) + (\ww_* - \tau^{-\frac{1}{2}} \h{\ww}_0)|^2}{|\ww - \ww_*|^2 + \tau^{-1}\eta_1^2}\,,
\end{align}
if $\tau^{-\frac{1}{2}} = o(\epsilon^2)$ holds (so that $O(\tau^{-\frac{3}{2}})$ is a higher-order term with respect to $\tau^{-\frac{1}{2}}\eta_1$). \mb{The formula \eqref{fano_express} exactly matches the Fano ansatz in \eqref{fano_ansa} with $\theta = \ww_* - \tau^{-\frac{1}{2}} \h{\ww}_0$ and $\gamma = \tau^{-1}\eta_1^2$. By Remark \ref{remfano}, this corresponds to the classical Fano resonance formula \eqref{eqfano} with $\Gamma = 2 (\tau^{-\frac{1}{2}} \eta_1)$ and $q = (\ww_* - \tau^{-\frac{1}{2}} \h{\ww}_0)/(\tau^{-\frac{1}{2}} \eta_1) = o(1)$ (from \eqref{aym_wstar} and \eqref{cal_5}, as well as $\tau^{-\frac{1}{2}} = o(\epsilon^2)$), which predicts the desired Fano-type double-spiked anomaly around the embedded eigenvalue $\ww_*$ with width $O(\tau^{-\frac{1}{2}} \eta_1)$. 
Similarly, we can break the symmetry of the metasurface by purely perturbing the shape of nanoparticles (i.e., $\epsilon = 0$). We assume the parameter $t = \tau^{-\frac{5}{8}}$ for simplicity and then find $\eta_1 = o(\tau^{-\frac{1}{2}})$ and $\eta_2 = O(\tau^{-\frac{1}{2}})$ as in \eqref{cal_5}. It is easy to see that in this regime, all the derivations \eqref{cal_6}--\eqref{fano_express} above still hold and we can have the Fano-type scattering anomalies in the same manner.}

\begin{remark}
\mb{In general, one could expect the occurrence of Fano resonance near the embedded eigenvalue when the symmetry of the system is destroyed. The broken symmetries by perturbing the particle configuration or the incident direction were considered in \cite{ammari2021bound} and \cite{lin2020mathematical}, respectively. Here, for the all-dielectric metasurface, we have justified that both of these symmetry-breaking scenarios could lead to a Fano resonance and hence generalized the previous results \cite{ammari2021bound,lin2020mathematical}. In addition, it is certainly possible to combine these two kinds of perturbations together to achieve the Fano resonance as observed in \cite{koshelev2018asymmetric}. For example, we can let $t = \tau^{-1} \epsilon^2$ and then it is easy to see that formulas \eqref{cal_5} and \eqref{fano_express} still hold.}
\end{remark}

\begin{remark}\label{remfano}
\mb{For completeness and reader's convenience, we briefly discuss the mathematical and physical significance of the Fano resonance ansatz \eqref{fano_ansa}. As mentioned above, Fano resonance corresponds to an asymmetric resonance peak (double-spiked anomaly) in the wave transmission or reflection spectra. It was Ugo Fano who suggested the first theoretical explanation of such a phenomenon by the following formula:
\begin{equation} \label{eqfano}
    F_q(e) = C \frac{|q+e|^2}{1 + e^2}, \q e \in \R
\end{equation}
where $q \in \R$ serves as a parameter that controls the shape of the resonance peak and dip; $C$ is a normalization constant such that $\max_{e} F_q(e) = 1$; the variable $e$ represents the energy difference:
\begin{equation*}
    e = \frac{\ww - \ww_*}{\Gamma/2}\,.
\end{equation*}
Here $\ww_* > 0$ denotes a resonance and $\Gamma > 0$ is the characteristic width. It is easy to see that the ansatz \eqref{fano_ansa} considered above is equivalent to \eqref{eqfano} by setting $\theta = q \Gamma/2$ and $\gamma = \Gamma^2/4$. We plot the formula \eqref{eqfano} with various values of the parameter $q$ in Figure \ref{fig:fano}, which clearly shows how it characterizes the Fano-type asymmetric line shape. We also emphasize that when $q$ varies from $+ \infty$ to $0$, 
the line shape transitions from a symmetric Lorentzian profile ($F_q(e) \sim 1/(1+e^2)$) to an inverted one ($F_q(e) \sim e^2/(1+e^2)$). We refer interested readers to \cite{shipman2010resonant,miroshnichenko2010fano} for more details on the physics of Fano resonances.}

\begin{figure}[!htbp]
    \centering
    \includegraphics[width=0.4\textwidth]{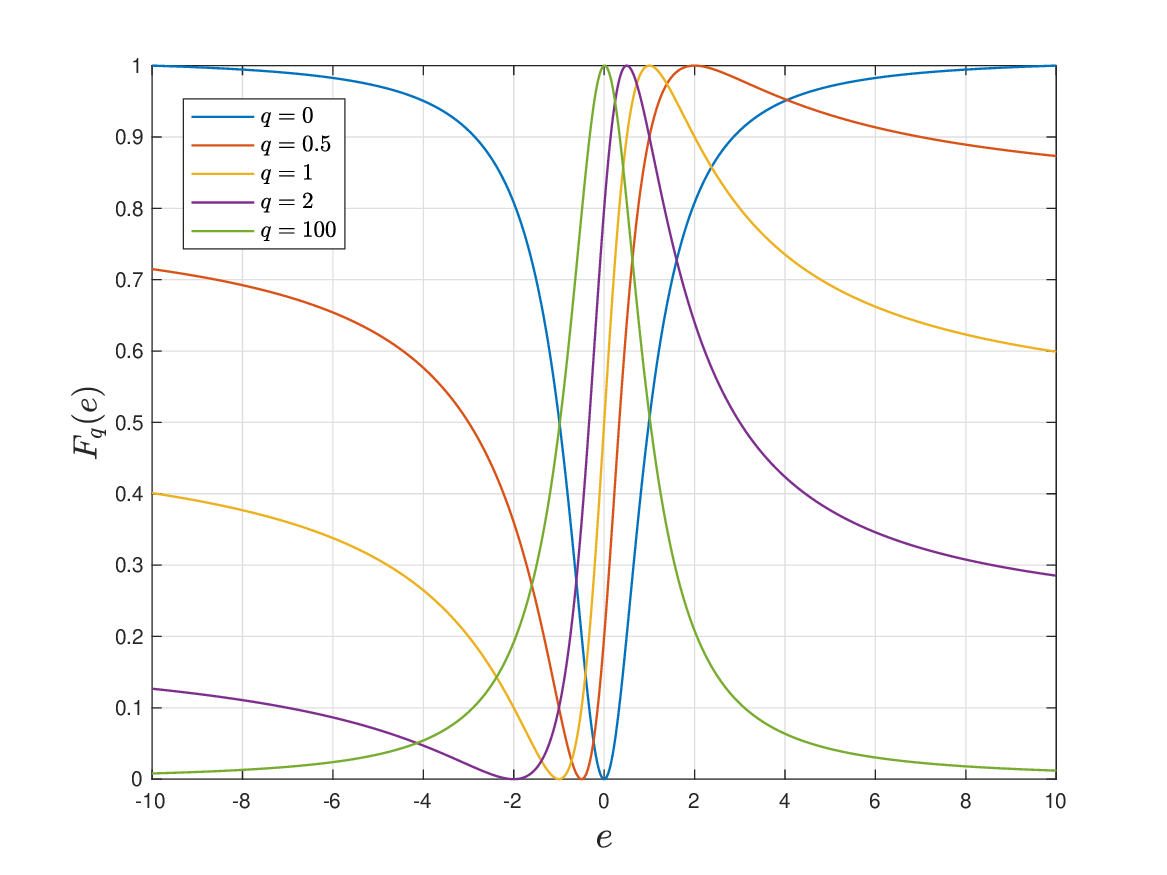} 
    \caption{\mb{Fano resonance formula \eqref{eqfano} with various $q$ \cite[Figure 5]{shipman2010resonant}.}}
    \label{fig:fano}
\end{figure}

\end{remark}

\section{Concluding remarks and discussions} \label{sec:conanddis}
In this work, we have given a comprehensive investigation of the resonant scattering by all-dielectric metasurfaces of periodically distributed nanoparticles with high refractive indices. We have characterized the essential spectrum of the Maxwell operator associated with the periodic scattering problem. We have shown that the interested real scattering resonances are the simple poles of the EM scattering resolvent with the corresponding resonant modes being the bound states exponentially decaying away from the metasurface. We have found that the real scattering resonances always exist below the essential spectrum of the Maxwell operator. We have also revealed some important relations between the resonances and the symmetry of the metasurface, which helps us to prove that in the high contrast regime, under the normal incidence, the subwavelength embedded eigenvalues exist for the symmetric dielectric metasurfaces. The resonant states corresponding to the embedded eigenvalues are the desired bound states in the continuum. To connect the BICs and the Fano resonances, we break the in-plane symmetry of the metasurface by perturbing the geometry of the dielectric nanoparticles and the incident directions (i.e., the Bloch wave vectors). We have derived the asymptotic expansions of the subwavelength resonances with respect to the high contrast and the shape perturbation. Furthermore, we have quantitatively approximated the reflection and transmission polarization vectors near the resonances and characterized the Fano-type asymmetric line shape in terms of the well-known Fano formula. Therefore, this work has provided a solid mathematical theory for the Fano resonance phenomenon physically observed in \cite{koshelev2018asymmetric,li2019symmetry}. 

There has also been increasing interest in understanding the topological properties of the BICs and the associated Fano resonances. In particular, we can define the topological charge by the winding number of the scattered polarization vectors, and the BICs are noting than the vortex centers of the polarization vectors in the momentum space \cite{zhen2014topological}. We also note that the BICs have a natural connection with the polarization singularities, which is a fundamental concept in the field of singular optics \cite{dennis2001topological,soskin2001singular,chen2019singularities,yoda2020generation}. It would be interesting to explore these concepts in our framework and further investigate the robustness of BICs and the Fano-type resonances.

\titleformat{\section}{\bfseries}{\appendixname~\thesection .}{0.5em}{}
\titleformat{\subsection}{\normalfont\itshape}{\thesubsection.}{0.5em}{}

\appendices

\section{Auxiliary lemmas} \label{app:a}
In this section, we establish some useful lemmas for proving Theorem \ref{thm:ess_discre}.  Letting $\ep = 1 + \tau \chi_D$ with $\tau > 0$, we consider the following elliptic equation on the unbounded domain $Y_\infty$:
\begin{align} \label{auxellieq}
    - \ddiv (\ep \na p) = f\,,
\end{align}
with the quasi-periodic boundary condition. For the well-posedness of \eqref{auxellieq}, we need the weighted Sobolev space: 
\begin{align*}
    H^{1,-1}_\alpha(Y_\infty):=\{u\,;\ (1+ |x_3|^2)^{-1/2} u \in L^2_\alpha(Y_\infty)\,,\ \na u \in L^2_\alpha(Y_\infty)\}\,.
\end{align*}
Note that $H^{1,-1}_\alpha(Y_\infty)$ with $\alpha = 0$ (i.e., the periodic case) includes the constant functions.

\begin{lemma}\label{lem:app1}
For any $f \in (H^{1,-1}_\alpha(Y_\infty))^*$, 
the equation \eqref{auxellieq} has a unique solution in $H^{1,-1}_\alpha(Y_\infty)/\R$,
continuously depending on the data $f$.
\end{lemma}

The proof of the above lemma easily follows from \cite[Theorem 2.5.14]{nedelec2001acoustic}. We next discuss how the solution of \eqref{auxellieq} depends on the contrast $\tau$. 

\begin{lemma} \label{auxlem1}
Suppose that $E \in {\bf L}_\alpha^2(Y_\infty)$ satisfies $\ddiv(\ep E) = 0$ and $\ddiv (\chi_D E) \neq 0$.  Define the coefficient $\ep' = \ep + \eta \chi_D$ for $\eta > 0$, and let $p \in H^{1,-1}_\alpha(Y_\infty)/\R$ be the unique solution to 
 \begin{align} \label{eq:shiftp}
     \ddiv (\ep'(E + \na p)) = 0\,.
 \end{align} 
Then there holds 
\begin{align*}
  \norm{E}^2_{\ep, Y_\infty} +
\norm{\na p}^2_{\ep, Y_\infty} \norm{E}_{D}/\norm{\na p}_D \le  \norm{E + \na p}^2_{\ep', Y_\infty}\,.
\end{align*}
\end{lemma}
\begin{proof}
    By the equation \eqref{eq:shiftp} and $\ddiv \ep E = 0$, it is clear that 
    \begin{align*}
      -  \ddiv (\ep' \na p) = \eta \ddiv (\chi_D E)\,,
    \end{align*}
    which implies $\na p \neq 0$ and
    \begin{align} \label{auxeqapp1}
        \norm{\na p}^2_{\ep', Y_\infty}  = - \eta (\na p, E)_{D} = - \eta (E, \na p)_{D}.
    \end{align}
    Then, a simple use of Cauchy's inequality gives 
    \begin{align} \label{auxeqapp2}
        \norm{\na p}^2_{\ep, Y_\infty} \le \eta (\norm{E}_{D} - \norm{\na p}_{D}) \norm{\na p}_{D}\,.
    \end{align}
 We can directly estimate, with the help of \eqref{auxeqapp1},
    \begin{align} \label{auxeqapp3}
        \norm{E + \na p}^2_{\ep', Y_\infty} -  \norm{E}^2_{\ep, Y_\infty} & = \eta  \norm{E}^2_{D} + \norm{\na p}^2_{\ep', Y_\infty} + \eta ( E, \na p)_{D} + \eta (\na p, E)_{D} \notag \\ 
        & = \eta  \norm{E}^2_{D} - \norm{\na p}^2_{\ep', Y_\infty}   \notag\\
        & \ge \eta  \norm{E}_{D} (\norm{E}_D - \norm{\na p}_D)\,.
    \end{align}
    The proof is complete by \eqref{auxeqapp2} and \eqref{auxeqapp3}: 
    \begin{align*}
      \norm{E + \na p}^2_{\ep', Y_\infty} -  \norm{E}^2_{\ep, Y_\infty}
&\ge   \norm{\na p}^2_{\ep, Y_\infty} \norm{E}_{D}/\norm{\na p}_D \ge 0\,.   \qedhere
    \end{align*}  
\end{proof}

\section{Analysis of the scalar eigenvalue problem~\eqref{eq:scalar}} \label{app:b}
In this section, we will formally discuss the high-contrast limits of the subwavelength resonances and the associated resonant modes of \eqref{eq:scalar}, and show the existence of real subwavelength resonances with a lower bound
under the symmetry assumption of the coefficient (which is implied by \eqref{assp:sym}):
\begin{align*}
    \ep(x,y,z)|_{Y_h} = \ep(-x,-y,z)|_{Y_h}\,.
\end{align*}
The rigorous analysis can be performed in the same manner as in \cite{shipman2010resonant,ammari2020mathematical}. We recall the periodic Green's function $G^{0,\ww}$ defined in \eqref{eq:qpgreen} (which clearly satisfies the transparent boundary condition $\frac{\p}{\p \nu}G^{0,\ww} = \mc{T} G^{0,\ww}$ on $\Sigma$) and the associated volume integral operator $\mc{K}_D^{0,\ww}$ in \eqref{def:vectorpotential}. By the asymptotic expansion \eqref{eq:expgreengn} of $G^{0,\ww}$: 
\begin{align*}
G^{0,\ww} = \h{G} + O(\ww)\,, \q \text{with} \ \h{G}: = \frac{i}{2 \ww} + G^{0,0}\,,
\end{align*}
we have $\mc{K}_D^{0,\ww}[\varphi] = \h{\mc{K}}_D[\varphi] + O(\ww)$, where the operator $\h{\mc{K}}_D$ is defined by 
\begin{align} \label{eqlimit}
    \h{\mc{K}}_D[\varphi] = \frac{i}{2 \ww} \l 1 ,  \vp\r_D + \mc{K}_D^{0,0}[\varphi]\,.
\end{align}
Thanks to the integral operator $\mc{K}_D^{0,\ww}$, the eigenvalue problem \eqref{eq:scalar} can be reformulated as 
\begin{align} \label{eqscal_lp}
    u = \ww^2 \tau   \mc{K}_D^{0,\ww}[u]\,.
\end{align}
To deal with the $O(\ww^{-1})$ singularity in \eqref{eqlimit}, we define the space $L^2_0(D) = \{u \in L^2(D)\,;\ \l 1 , u\r_D = 0\}$ and write the functions $u \in L^2(D)$ as
$
 u = \l 1_D, u\r_D + \bar{u},
$
where $1_D := 1/|D|$ and $\bar{u} \in L^2_0(D)$. Such decomposition is unique and orthogonal. Then the problem \eqref{eqscal_lp} is equivalent to
\begin{align*}
    \mm 
     \l 1_D, u\r_D \\
    \bar{u}
    \nn -
    \ww^2\tau \mm
     \l 1_D,   \mc{K}_D^{0,\ww}[1]\r_D &   \l 1_D,   \mc{K}_D^{0,\ww}[\dd]\r_D \\
     \mc{K}_D^{0,\ww}[1] -  \l 1_D,   \mc{K}_D^{0,\ww}[1]\r_D  &  \mc{K}_D^{0,\ww}[\dd] - \l 1_D,   \mc{K}_D^{0,\ww}[\dd]\r_D
    \nn  \mm 
    \l 1_D, u\r_D \\
    \bar{u}
    \nn = 0\,.
\end{align*}
By the asymptotic expansion \eqref{eqlimit} and $\h{\mc{K}}_D[1] = \frac{i|D|}{2 \ww} + \mc{K}_D^{0,0}[1]$, we have  
\begin{align} \label{eq:asym_form_eig}
 \mm 1 - \ww \tau \frac{i |D|}{2} + O(\ww^2 \tau)
    & O(\ww^2 \tau) \\  -
    \ww^2\tau
\phi_1 + O(\ww^3 \tau)  & 1 -
    \ww^2\tau \mc{K}_D + O(\ww^3 \tau) 
    \nn   \mm 
    \l 1_D, u\r_D \\
    \bar{u}
    \nn = 0\,,
\end{align}
where $\phi_1 := \mc{K}_D^{0,0}[1] -  \l 1_D,  \mc{K}_D^{0,0}[1]\r_D$ and
\begin{align} \label{eq:limopscal}
\mc{K}_D[\dd]:= \mc{K}_D^{0,0}[\dd] - \l 1_D,  \mc{K}_D^{0,0}[\dd] \r_D: L_0^2(D) \to L^2_0(D)\,.
\end{align}
Note from the coefficient matrix in \eqref{eq:asym_form_eig} that in the high contrast regime, the subwavelength resonances of \eqref{eq:scalar} can happen in two cases: $\ww = O(\tau^{-1})$  and $\ww = O(\sqrt{\tau}^{-1})$. In the first case, we have the asymptotics: 
\begin{align*}
 \mm 1 - \ww \tau \frac{i |D|}{2} + O(\ww^2 \tau)
    & O(\ww^2 \tau) \\  -
    \ww^2\tau
\phi_1 + O(\ww^3 \tau)  & 1 -
    \ww^2\tau \mc{K}_D + O(\ww^3 \tau) 
    \nn   = \mm 1 - \ww \tau \frac{i |D|}{2} 
    & 0 \\ 0 & 1 
    \nn  + O(\tau^{-1})\,, 
\end{align*} 
which readily gives a subwavelength resonance $\ww = \frac{2}{i \tau |D|} + O(\tau^{-2})$ with the associated resonant mode being almost constant: $u = 1 + O(\tau^{-1})$ on $D$. In the other case: $\ww  = O(\sqrt{\tau}^{-1})$, we first divide the first row of \eqref{eq:asym_form_eig} by $\sqrt{\tau}^{-1}$ and then find 
\begin{align} \label{eq:asymptsec}
 \mm 1 - \ww \sqrt{\tau} \frac{i |D|}{2} + O(\ww^2 \sqrt{\tau})
    & O(\ww^2 \sqrt{\tau}) \\  -
    \ww^2\tau
\phi_1 + O(\ww^3 \tau)   & 1 -
    \ww^2 \tau \mc{K}_D + O(\ww^3 \tau) 
    \nn  = \mm 1 - \ww \sqrt{\tau} \frac{i |D|}{2} 
    & 0 \\  -
    \ww^2\tau
\phi_1 & 1 -
    \ww^2 \tau \mc{K}_D
    \nn  + O(\sqrt{\tau}^{-1})\,.
\end{align} 
Suppose that the compact self-adjoint operator $\mc{K}_D$ admits the eigen-decomposition: 
\begin{align*}
 \mc{K}_D[u] = \sum_{j = 0}^\infty \eta_j (v_j, u)_D  v_j\,,\q \text{for}\ u \in L^2_0(D)\,.   
\end{align*}
Then, by the asymptotics \eqref{eq:asymptsec}, we can have another class of subwavelength resonances $\ww_j = \sqrt{\eta_j \tau}^{-1} + O(\tau^{-1})$ with the resonant modes $u_j = v_j + O(\sqrt{\tau}^{-1})$ on $D$. We now summarize the above discussion as follows.

\begin{proposition} \label{prop:limit_scalar}
The subwavelength resonances exist for the problem \eqref{eq:scalar} in the high contrast regime. Moreover, for any subwavelength resonance $\ww$, there holds either 
\begin{align} \label{asym_scalar}
 \ww = \frac{2}{i \tau |D|} + O(\tau^{-2})\,, \q \text{or}\q \ww = \frac{1}{\sqrt{\eta_j \tau}} + O(\tau^{-1}) \q \text{for some}\ j \ge 0\,, 
\end{align}
where $\eta_j$ is an eigenvalue of the operator $\mc{K}_D$ on $L_0^2(D)$. 
\end{proposition}

We proceed to consider the existence of real subwavelength resonances, which has been discussed in \cite{bonnet1994guided,shipman2007guided}. We provide a sketch of arguments below for completeness.
With the help of the DtN operator \eqref{def:dtn_normal}, the variational formulation of \eqref{eq:scalar} reads as follows: find $(\ww,u) \in \R\backslash\{0\} \t H_p^1(Y_h)$ such that
\begin{align} \label{scalar_form}
    b_{\ww}(\vp,u) = \ww^2 (\vp, \ep u)_{Y_h}\,, \q \forall \vp \in H^1_p(Y_h)\,,
\end{align}
where $b_{\ww}(\dd,\dd)$ is the sesquilinear form on $H^1_p(Y_h)$:
\begin{align}  \label{scalar_form_0}
b_{\ww}(\vp,u) = (\na \vp, \na u)_{Y_h} - \l \vp, \mc{T} u \r_{\Sigma}\,.
\end{align}
Let $\mc{T}^{\mc{P}}$ and $\mc{T}^{\mc{E}}$ be defined as in Section \ref{sec:exist_real} and $b^{\mc{E}}_\ww(\dd,\dd)$ be defined by \eqref{scalar_form} with $\mc{T}$ replaced by $\mc{T}^{\mc{E}}$. Moreover, we introduce the symmetric  and antisymmetric functions:
\begin{align*}
{H}^1_{p,sym}(Y_h) = \{\vp \in H^1_p(Y_h)\,;\ \vp(-x',x_3) = \vp(x',x_3)\}\,,
\end{align*}
and 
\begin{align*}
    {H}^1_{p,ant}(Y_h) = \{\vp \in H^1_p(Y_h)\,;\ \vp(-x',x_3) = - \vp(x',x_3)\}\,,
\end{align*}
which are orthogonal with respect to $b^{\mc{E}}_\ww(\dd,\dd)$. Similarly to
Proposition \ref{prop:basic_Prop}, we can prove that $b^{\mc{E}}_\ww(\dd,\dd)$ is a symmetric positive form on ${H}^1_{p,ant}(Y_h)$ with eigenvalues $\lad^{ant}_j(\tau,\ww)$ decreasing in $\ww$, where $\lad^{ant}_j$ is given by the min-max principle:
\begin{align} \label{eq:min_max_scalar}
\lad^{ant}_j(\tau, \ww) =  \min_{\substack{V \subset { H}_{p,ant}^1(Y_h) \\ \dim V = j + 1}  } \max_{\substack{ H \in V }} \frac{b^{\mc{E}}_\ww(u,u)}{(u,\ep u)_{Y_h}}\,.
\end{align}
It follows that for each $j$, the equation $\lad^{ant}_j(\tau,\ww) = \ww^2$ admits a unique solution $\ww_j > 0$, and we can find $u_j \in {H}_{p,ant}^1(Y_h)$ such that 
\begin{align} \label{appeq_1}
    b_{\ww_j}^{\mc{E}}(\vp,u_j) = \ww_j^2 (\vp, \ep u)_D \,,\q \forall \vp  \in {H}_{p,ant}^1(Y_h)\,.
\end{align}
We claim that for $\tau$ large enough, $\ww_j$ is in the subwavelength regime. Indeed, by \eqref{eq:min_max_scalar}, we have
\begin{align*}
    \lad^{ant}_j(\tau, \ww) \le \min_{\substack{V \subset { H}_{0,ant}^1(Y_h) \\ \dim V = j}  } \max_{\substack{ H \in V }} \frac{b^{\mc{E}}_\ww(u,u)}{\tau (u,u)_D} = O(\tau^{-1})\,, 
\end{align*}
with ${H}_{0,ant}^1(Y_h):= \{u\in H_{p,ant}^1(Y_h)\,;\ u = 0 \ \text{on}\ \overline{Y_h\backslash D}\}$, which readily gives $\ww_j = O(\sqrt{\tau}^{-1})$. By the orthogonality between ${H}^1_{p,sym}(Y_h)$ and $ {H}^1_{p,ant}(Y_h)$, the equation \eqref{appeq_1} holds for any $\vp \in H^1_p(Y_h)$. In addition, thanks to $\ww_j \ll 1$ and $u_j \in {H}_{p,ant}^1(Y_h)$, we have $(u_j)_q = 0$ for $q \in \Lad^*_{\mc{P}} = \{0\}$ by symmetry, which gives $ b_{\ww_j}^{\mc{E}}(\vp,u_j) = b_{\ww_j}(\vp,u_j)$. Therefore, $\ww_j$ is the desired real subwavelength resonance for \eqref{eq:scalar}. The following result, which gives a lower bound for real subwavelength resonances, is a variant of \cite[Theorem 4.1]{shipman2007guided}.  Let $X$ be the subspace of $H^1_p(Y_h)$ defined by 
\begin{align*}
    X: = \{u \in H^1_p(Y_h)\,;\ u_q = 0 \ \text{for}\ q \in \Lad_{\mc{P}}^*\}\,.
\end{align*}

\begin{proposition}\label{prop_resonance}
Suppose that $\ww$ is a real subwavelength resonance of \eqref{eq:scalar} satisfying $\ww = O(\sqrt{\tau}^{-1})$ for large enough $\tau$. Then it holds that 
\begin{align*}
    \ww \ge \sqrt{\frac{\gamma_0}{\tau}} +  O\left(\frac{1}{\tau^{3/4}}\right)\,,
\end{align*}
where $\gamma_0 > 0$ is given by 
\begin{align} \label{constant_lower}
    \gamma_0 = \inf_{u \in H_p^1(Y_h)}\frac{b_0(u, u) }{(u, u)_{Y_h}}\,,
\end{align}
with $b_0(\dd,\dd)$ given in  \eqref{scalar_form_0} with $\ww = 0$.
\end{proposition}

\begin{proof}
Let $u$ be a resonant mode associated with the real resonance $\ww = O(\sqrt{\tau}^{-1})$. Then $u \in X$ satisfies 
\begin{align*}
    b_\ww(\vp, u) = \ww^2(\vp, \ep u)_{Y_h}\,,\q \forall \vp \in X\,.
\end{align*}
We now consider the following variational eigenvalue problem on $X$:
\begin{align} \label{scalar_homo}
   b_\ww(\vp, u) = \gamma (\vp, u)_{Y_h}\,,\q \forall \vp \in X\,.
\end{align}
Since $b_\ww(\dd,\dd)$ is a positive symmetric form on $X$, we let $\gamma_j(\ww) \ge 0$ be the eigenvalues of \eqref{scalar_homo}, which has the asymptotics: $\gamma_j(\ww) = \gamma_j(0) + O(\sqrt{\tau}^{-1})$, by the standard perturbation theory. Hence, we have
\begin{align*}
    \ww^2 = \frac{b_\ww(u, u) }{(u, \ep u)_{Y_h}} \ge \frac{1}{\tau} \inf_{u \in H_p^1(Y_h)}\frac{b_\ww(u, u) }{(u, u)_{Y_h}} \ge \frac{1}{\tau} \gamma_0(0) + O\left(\frac{1}{\tau^{3/2}}\right).
\end{align*}
The proof is complete by noting that $\gamma_0(0)$ is strictly positive. Indeed, if $\gamma_0(0) = 0$, then the eigenfunction must be constant which does not belong to the space $X$.
\end{proof}

\section{Proofs in Section \ref{sec:shape_deriva}} \label{app:c}

\subsection{Proof of Lemma \ref{prop:differnetiablity}}
Let us first compute the derivative of $\P_D^t$ at $t = 0$. By the construction of the Helmholtz decomposition in \cite{amrouche1998vector}, we have, for $E \in {\bf L}^2(D)$, 
\begin{align*}
    \P_D^t E = E - \Phi_t^* \na u\,,
\end{align*}
where, up to constants,  $u \in H^1(D_t)$ satisfies the variational equation:
\begin{align} \label{auxeq:variation}
    (\Phi_{t,*} E, \na \vp)_{\Phi_t(D)} = (\na u, \na \vp)_{\Phi_t(D)}\,,\q \forall \vp \in H^1(D_t)\,.
\end{align}
Then, we define $Q_t = (\na \Phi_t)^{-1}$, and find, by chain rule,
\begin{align} \label{auxeq3}
    \P_D^t  E = E - \Phi_t^* \na u = E - Q_t^T \na \Phi_t^* u\,,
\end{align}
and 
\begin{align} \label{auxeq11}
    \frac{d}{d t} \P_D^t E = - \frac{d}{dt} Q^T_t \na \Phi_t^*u - Q_t^T \na \frac{d}{d t} \Phi_t^*u\,.
\end{align}
It is direct to compute 
\begin{align} \label{eq:der_defor_invers}
    \frac{d}{d t} Q^T_t = \frac{d}{d t} (I + t \na V)^{-T} = - (\na V)^T\,.
\end{align}
Moreover, by change of variables, we can see from \eqref{auxeq:variation} that $\Phi_t^* u$ is the
solution to 
\begin{align} \label{auxeq4}
    (E, Q_t^T \na \vp J(\Phi_t))_{D} = (Q_tQ_t^T\na w, \na \vp J(\Phi_t))_{D}\,, \q \forall \vp \in H^1(D)\,,
\end{align}
where $J(\Phi_t) := \det((I + t \na V)$ is the Jacobian determinant. It is known that \cite{delfour2011shapes,henrot2018shape}
\begin{align} \label{eq:der_jaco_det}
\frac{d}{d t}\Big|_{t = 0} J(\Phi_t) = \frac{d}{d t}\Big|_{t = 0} \det (I + t \na V) = \ddiv V\,. 
\end{align}
We denote by $u^0$ the solution to \eqref{auxeq:variation} at $t = 0$, i.e., 
\begin{align} \label{eqforu1}
    (E, \na \vp)_D = (\na u^0, \na \vp)_D\,,\q \forall \vp \in H^1(D)\,,
\end{align}
and take the derivative of the variational equation \eqref{auxeq4} at $t = 0$ with \eqref{eq:der_defor_invers} and \eqref{eq:der_jaco_det},
 \begin{align*} 
        (E, ( - (\na V)^T +  \ddiv V) \na \vp)_{D}  = &  - (  (\na V + (\na V)^T - \ddiv V) \na u^0 , \na \vp )_{D} 
        +  (\na \frac{d}{d_t}\Big|_{t = 0} \Phi_{t}^* u, \na \vp )_{D}\,.
    \end{align*}
Then, by \eqref{auxeq11}, \eqref{eq:der_defor_invers} and \eqref{eqforu1} with $\na u^0 = \P_{{\rm d},D}^{\perp} E$, we can conclude the desired formula \eqref{der_proj}. 

We next compute the derivative of $\w{\mc{K}}_D^t$, which essentially follows from \cite[Proposition 3.2]{sakly2017shape}. We sketch the computation below. By definition and change of variables, we have 
\begin{align*}
     \w{\mc{K}}_D^t[E] & = \Phi_t^* \mc{K}_{D_t}^{0,0} \Phi_{t,*} [E] = \int_{\Phi_t(D)} G^{0,0}(\Phi_t(x), y)E(\Phi_t^{-1}(y))\ dy  \\
        & = \int_{D} G^{0,0}(\Phi_t(x), \Phi_t(y))E(y) J(\Phi_t)\  dy\,.
\end{align*}
Then, a direct calculation gives
\begin{align*}
    \frac{d}{d t}\Big|_{t = 0} \w{\mc{K}}_D^t[E] & = \int_{D} - \na_y G^{0,0}(x, y) \dd (V(x) - V(y)) E(y) \ d y - \int_{D} G^{0,0}(x,y) \ddiv_y (V(x)-V(y)) E(y) \ d y \\
    & = \int_{D} - \ddiv_y \big[G^{0,0}(x, y)  (V(x)-V(y))\big] E(y)\  d y\,. \qedhere
\end{align*}

\subsection{Proof of Proposition \ref{lem:shapedev}}
Similarly to Theorem \ref{asym:resonance}, the proposition is a direct consequence of the standard perturbation theory with the Lyapunov-Schmidt reduction; see also the monograph \cite{henry2005perturbation}. It suffices to compute the Hermitian matrix $\h{K}$. We start with the definition, by chain rule, 
\begin{align} \label{auxeq6}
    \frac{d}{d t}\Big|_{t = 0} K_{ij} = \Big(\vp_i, \frac{d}{d t}\Big|_{t = 0} \P_D^t \np^{0,0}_{D} [\vp_j] \Big)_{D} + \Big(\vp_i,  \np^{0,0}_{D} \frac{d}{d t}\Big|_{t = 0} \P_D^t  [\vp_j]\Big)_{D} + \Big(\vp_i, \frac{d}{d t}\Big|_{t = 0} \w{\mc{K}}_D^t [\vp_j]\Big)_{D}\,.
\end{align}
We observe from \cite[Proposition 5.1]{ammari2020mathematical} that the eigenfunction of $\mbb{K}_D$ is of $H^1$-regularity, which allows us to use the integration by parts for \eqref{eq:der_integral}: for $E \in {\bf H}^1(D)$,
\begin{align*} 
    \frac{d}{d t}\Big|_{t = 0} \w{\mc{K}}_D^t [E] = \int_{D}  G^{0,0}(x, y) \na E(y)(V(x)-V(y))\  d y - \int_{\p D} G^{0,0}(x , y)  (V(x) - V(y)) \dd \n(y) E(y)\ d y\,.
\end{align*}
It follows that the last term in \eqref{auxeq6} can be calculated as
\begin{align} \label{auxeq7}
    \Big(\vp_i,  \frac{d}{d t}\Big|_{t = 0} \w{\mc{K}}_D^t [\vp_j]\Big)_{D} = & - (\np^{0,0}_{D}[\vp_i], \na \vp_j V)_{D} + (\np^{0,0}_{D}[ \vp_i \otimes V], \na \vp_j)_{D} \notag \\
    & + (\np^{0,0}_{D}[\vp_i], \n \dd V \vp_j)_{\p D} - ( \np^{0,0}_{D}[\vp_i  \otimes V],  \vp_j \otimes \n)_{\p D}\,.
\end{align}
Thanks to $\mbb{K}_D[\vp_i] = \lad_0 \vp_i$, for the first and third term in \eqref{auxeq7}, there hold 
\begin{align}  \label{auxeq9}
    - (\np^{0,0}_{D}[\vp_i], \na \vp_j V)_{D} = - (\lad_0 \vp_i, \na \vp_j V)_{D} - (\P_{{\rm d},D}^\perp \np^{0,0}_{D}[\vp_i], \na \vp_j V)_{D}\,, 
\end{align}
     and 
\begin{align}  \label{auxeq10}
    (\np^{0,0}_{D}[\vp_i], \n \dd V \vp_j)_{\p D} =  (\lad_0 \vp_i, \n \dd V \vp_j)_{\p D} + (\P_{{\rm d},D}^\perp \np^{0,0}_{D}[\vp_i], \n \dd V \vp_j)_{\p D}\,.
\end{align}
For the second term in \eqref{auxeq7}, a simple computation gives, by changing the order of differentiation and integration, 
 \begin{align} \label{auxeq8}
        (\np^{0,0}_{D}[\vp_i \otimes V], \na \vp_j)_{D} = (\vp_i \otimes V, \na \np^{0,0}_{D}[\vp_j])_{D} + ( \np^{0,0}_{D}[\vp_i  \otimes V],  \vp_j \otimes \n)_{\p D}\,.
\end{align} 
Then, applying the integration by parts to the first term in the right-hand side of \eqref{auxeq8}, we find 
\begin{align} \label{auxeqq_cross}
    (\vp_i \otimes V, \na \np^{0,0}_{D}[\vp_j])_{D}
         =& - (\ddiv V \vp_i + \na \vp_i V,  (\P_{{\rm d},D} + \P_{{\rm d}, D}^\perp) \np^{0,0}_{D}[\vp_j])_{D} + (\n \dd V \vp_i ,  \np^{0,0}_{D}[\vp_j])_{\p D} \notag \\
         =&  - (\ddiv (\vp_i \otimes V),  \lad_0 \vp_j)_{D}  - (\ddiv V \vp_i + \na \vp_i V,   \P_{{\rm d}, D}^\perp \np^{0,0}_{D}[\vp_j])_{D} + (\n \dd V \vp_i ,  \np^{0,0}_{D}[\vp_j])_{\p D} \notag \\
         =& - (\ddiv V \vp_i + \na \vp_i V,   \P_{{\rm d}, D}^\perp \np^{0,0}_{D}[\vp_j])_{D} + (\n \dd V \vp_i ,  \np^{0,0}_{D}[\vp_j])_{\p D} \notag \\
          & + (\vp_i,  \lad_0 \na \vp_j V)_{D} - (\vp_i,  \lad_0  \vp_j \n \dd V)_{\p D}\,.
\end{align}
Combining the above calculations \eqref{auxeq7}-\eqref{auxeqq_cross}, we arrive at 
\begin{align} \label{auxeq12}
    \Big(\vp_i, \frac{d}{d t}\Big|_{t = 0} \w{\mc{K}}_D^t  [\vp_j]\Big)_{D}  = & - (\P_{{\rm d},D}^\perp \np^{0,0}_{D}[\vp_i], \na \vp_j V)_{D}  + (\P_{{\rm d},D}^\perp\np^{0,0}_{D}[\vp_i], \n \dd V \vp_j)_{\p D} + \lad_0 (\n \dd V \vp_i , \vp_j)_{\p D} \notag \\
  &  - (\ddiv V \vp_i + \na \vp_i V, \P_{{\rm d},D}^\perp \np^{0,0}_{D}[\vp_j])_{D} + (\n \dd V \vp_i , \P_{{\rm d},D}^\perp \np^{0,0}_{D}[\vp_j])_{\p D}\,.
\end{align}

We next compute the first two terms in \eqref{auxeq6}. It is easy to see from \eqref{der_proj} that 
\begin{align*}
    \P_{{\rm d},D}^\perp \na w =  \P_{{\rm d},D}^\perp ( - \na V +  \ddiv V) \P_{{\rm d},D} E +  \P_{{\rm d},D}^\perp (\na V)^T \P_{{\rm d},D}^{\perp} E\,,
\end{align*}
which implies 
\begin{align} \label{auxder_1}
    \P_{{\rm d},D}^\perp  \frac{d}{d t}\Big|_{t = 0}\P_D^t[E] =  \P_{{\rm d},D}^\perp (\na V - \ddiv V) \P_{{\rm d},D} E\,.
\end{align}
It is also clear from \eqref{der_proj} that 
\begin{align} \label{auxder_2}
    \P_{{\rm d},D}\frac{d}{d t}\Big|_{t = 0}\P_D^t  = \P_{{\rm d},D} (\na V)^T  \P_{{\rm d},D}^\perp\,.
\end{align}
By \eqref{auxder_1}, we readily have 
\begin{align} \label{auxeq13}
   \Big(\vp_i,  \np^{0,0}_{D} \frac{d}{d t}\Big|_{t = 0} \P_D^t  [\vp_j]\Big)_{D} & = \Big(\vp_i,  \np^{0,0}_{D} (\P_{{\rm d},D} + \P_{{\rm d},D}^\perp)  \frac{d}{d t}\Big|_{t = 0} \P_D^t  [\vp_j]\Big)_{D} \notag\\ & = \Big(\lad_0 \vp_i , \frac{d}{d t}\Big|_{t = 0} \P_D^t \vp_j\Big)_{D} + \Big(\P_{{\rm d},D}^\perp \np_{D}^{0,0} \vp_i,  (\na V - \ddiv V) \vp_j\Big)_{D} \notag \\
    & = \Big(\P_{{\rm d},D}^\perp \np_{D}^{0,0} \vp_i,  (\na V - \ddiv V) \vp_j\Big)_{D}\,,
 \end{align}
where the first term in the second equality vanishes due to \eqref{auxder_2}. Similarly, we derive
 \begin{align} \label{eq15}
     \Big(\vp_i, \frac{d}{d t}\Big|_{t = 0} \P_D^t \np^{0,0}_{D} [\vp_j] \Big)_{D}  =  \Big(\vp_i, (\na V)^T \P_{{\rm d},D}^\perp \np^{0,0}_{D} \vp_j\Big)_{D}\,.
 \end{align}
 Collecting \eqref{auxeq12}, \eqref{auxeq13} and \eqref{eq15}, we obtain
 \begin{align} \label{eq:matrix_shape}
     \frac{d}{dt}\Big|_{t = 0} K_{ij} = & (\P_{{\rm d},D}^\perp \np_{D}^{0,0} \vp_i,  (\na V - \ddiv V) \vp_j- \na \vp_j V)_{D} + (\P_{{\rm d},D}^\perp \np^{0,0}_{D}[\vp_i], \n \dd V \vp_j)_{\p D} \notag \\
     &  + ( (\na V - \ddiv V) \vp_i - \na \vp_i V, \P_{{\rm d},D}^\perp \np^{0,0}_{D}[\vp_j])_{D} + (\n \dd V \vp_i , \P_{{\rm d},D}^\perp \np^{0,0}_{D}[\vp_j])_{\p D} \notag \\
     & + \lad_0 (\n \dd V \vp_i , \vp_j)_{\p D}\,.
 \end{align}
To simplify the above formula, we write $\na p_i \in {\bf L}^2(D)$ for $\P_{{\rm d},D}^\perp \np_{D}^{0,0} \vp_i$ and then have
\begin{align*}
    & (\P_{{\rm d},D}^\perp \np_{D}^{0,0} \vp_i,  (\na V - \ddiv V) \vp_j- \na \vp_j V)_{D} \\
= & (\na p_i,  \curl (V \t \vp_j))_{D} \\ =& (\na_{\p D} p_i, \n \t (V \t \vp_j))_{\p D}  \\ = & (\na_{\p D} p_i, - \n \dd V \vp_j)_{\p D}\,,
\end{align*}
by integration by parts with $\n \dd \vp_j = 0$ on $\p D$. It follows that
\begin{align*}
   (\P_{{\rm d},D}^\perp \np_{D}^{0,0} \vp_i,  (\na V - \ddiv V) \vp_j- \na \vp_j V)_{D} + (\P_{{\rm d},D}^\perp\np^{0,0}_{D}[\vp_i], \n \dd V \vp_j)_{\p D} = 0\,.
\end{align*} 
Similar calculation yields the vanishing of the term $( (\na V - \ddiv V) \vp_i - \na \vp_i V, \P_{{\rm d},D}^\perp \np^{0,0}_{D}[\vp_j])_{D} + (\n \dd V \vp_i , \P_{{\rm d},D}^\perp \np^{0,0}_{D}[\vp_j])_{\p D}$. 
The proof is complete by \eqref{eq:matrix_shape}.



\if \commentflag{} = \ct
Note that $A_{\tau,\ww}^{\mc{E}}$, as an unbounded form on ${\bf L}_p^2(W_p)$ with $\dom (A_{\tau,\ww}^{\mc{E}}) = {\bf H}^1(Y_h)$, is densely defined, closed and positive. By the first representation theorem \cite[Theorem VI-2.1]{kato2013perturbation}, there is a unique  
self-adjoint operator $\mbb{A}_{\tau,\ww}^{\mc{E}} $ on ${\bf L}_p^2(W_p)$ such that 
\begin{align*}
    A^{\mc{E}}_{\tau,\ww}(\vp, H) = (\vp, \mbb{A}_{\tau,\ww}^{\mc{E}} H)\,, \q H \in \dom(\mbb{A}_{\tau,\ww}^{\mc{E}})\,,\ \vp \in \dom (A_{\tau,\ww}^{\mc{E}})\,,
\end{align*}
with $\dom(\mbb{A}_{\tau,\ww}^{\mc{E}}) $ dense in $\dom (A_{\tau,\ww}^{\mc{E}})$ with respect to the $\norm{\dd}_{{\bf H}^1(Y_h)}$-norm.  We consider the high contrast limit of the form $ A^{\mc{E}}_{\tau,\ww}$ by using the monotonicity of $A_{\tau,\ww}^{\mc{E}}$ in $\tau$, following the arguments in \cite{simon1978canonical, hempel2000spectral}. We need the following spaces: \small
\begin{align*}
    \w{{\bf L}}^2_p(Y_h) = \{H \in {\bf L}^2(Y_h)\,;\ H = {\bf 0}\ a.e.\ on \ Y_h \backslash D\}\q \text{and}\q 
    \w{{\bf H}}^1_p(Y_h) = \{H \in {\bf H}^1(Y_h)\,;\ H = {\bf 0}\ a.e.\ on \ Y_h \backslash D\}\,.
\end{align*}
\normalsize
We define a closed limiting quadratic form on $\w{{\bf L}}^2_p(Y_h) \oplus \C^3$ by 
\begin{equation*}
    A^{\mc{E}}_\infty (H,H) =  \sup_{\tau > 0} A_{\tau,\ww}^{\mc{E}}(H,H) = \lim_{\tau \to +\infty} A_{\tau,\ww}^{\mc{E}}(H,H) \,, 
\end{equation*}
with the domain:
\begin{equation*}
  \dom (A^{\mc{E}}_\infty) := \{H \in {\bf H}_p^1(Y_h) \,;\ \sup_{\tau > 0} A_{\tau,\ww}^{\mc{E}}(H,H) < \infty \}\,.
\end{equation*}
To characterize $\dom (A^{\mc{E}}_\infty)$, it suffices to note that by \eqref{eq:formhh} and \eqref{eq:key_est_1}, the condition $\sup_{\tau > 0} A_{\tau,\ww}^{\mc{E}}(H,H) < \infty$ holds if and only if for some constant $C > 0$,
\begin{align*}
     \tau (\na H, \na H)_{Y_h \backslash D} + \tau \sum_{q \in \Lambda_\mc{E}^*} \sqrt{|q|^2 - \ww^2} |H_q|^2 \le C\,,\q \forall \tau >0\,.
\end{align*}
It follows that $\na H = 0$ a.e. on $Y_h \backslash D$. Hence, we have 
\begin{align*}
    A^{\mc{E}}_\infty (H,H) = \norm{\na H}_D^2\,,\q \forall H \in  \dom (A^{\mc{E}}_\infty) =  \w{{\bf H}}^1_p(Y_h) \oplus \C^3\,.
\end{align*}
Moreover, there exists a unique self-adjoint positive operator $\mbb{A}_{\infty}^{\mc{E}}$ on $\w{{\bf L}}^2_p(Y_h) \oplus \C^3$ corresponding to the form $A_\infty^{\mc{E}}$. It is also clear that the operators $\mbb{A}_{\tau,\ww}^{\mc{E}}$ and $\mbb{A}_{\infty}^{\mc{E}}$ have compact resolvents and thus purely discrete spectrum.
Then the following result is a direct consequence of Theorem VIII-3.5 and Theorem VIII-3.15 in \cite{kato2013perturbation}. 

\begin{proposition} \label{prop:conver_ladj}
For all $\xi \in \C \backslash \R$, we have, as $\tau \to \infty$,  
\begin{align*}
   R_{\mbb{A}_{\tau,\ww}^{\mc{E}}}(\xi): = (\mbb{A}_{\tau,\ww}^{\mc{E}} - \xi)^{-1} \longrightarrow  R_{\mbb{A}_{\infty}^{\mc{E}}}(\xi): = (\mbb{A}_{\infty}^{\mc{E}} - \xi)^{-1} \q \text{in} \q \mc{L}({\bf L}^2_p(Y_h))\,.
\end{align*}
Denote by $\{\mu_j\}_{j = 0}^\infty$ the eigenvalues of $\mbb{A}_{\infty}^{\mc{E}}$. Then, it holds that for $j \ge 0$, 
\begin{align*}
    \lad_{j}(\tau,\ww) \nearrow \mu_j \q \text{as}\q \tau \to \infty\,,
\end{align*}
uniformly on any compact subset of $[0,c_0)$, where $\lad_j$ are the eigenvalues of the operator $\mbb{A}_{\tau,\ww}^{\mc{E}}$ characterized by \eqref{eq:min_max}. 
\end{proposition}

\begin{corollary}
Let $\ww_j(\tau)$ be defined by \eqref{def:eigen_fre} for $j \ge 0$. Then it holds that for each $j$, as $\tau \to \infty$, 
\begin{align} \label{eq:conver_1}
    \tau \ww_j^2(\tau) \to \mu_j\,,
\end{align}
and 
\begin{align} \label{eq:conver_2}
{\rm E}_{\mbb{A}_{\tau,\ww_j(\tau)}^{\mc{E}}}(\tau \ww_j^2(\tau)) \to  {\rm  E}_{\mbb{A}_{\infty}^{\mc{E}}}(\mu_j)\,,
\end{align}
in the strong operator topology. 
\end{corollary}

\begin{proof}
Note that $\lad_j(\tau,\ww) \le \mu_j$ for any $\tau > 0$ and $\ww \in [0, c_0)$, and $\tau \ww^2 > \mu_j$ if and only if $\ww > \sqrt{\tau^{-1}\mu_j}$. It follows from \eqref{def:eigen_fre} that $\ww_j(\tau) < c_0$ for large enough $\tau$, and thus the convergence \eqref{eq:conver_1} holds by Proposition \ref{prop:conver_ladj}. Clearly, the limit \eqref{eq:conver_1} yields $\ww_j(\tau) \to 0$, and we have $\mbb{A}_{\tau,\ww_j(\tau)}^{\mc{E}} \to \mbb{A}_{\infty}^{\mc{E}}$ in the strong resolvent sense as $\tau \to \infty$. Then a standard result in the perturbation theory \cite[Theorem 1.15]{kato2013perturbation} readily implies \eqref{eq:conver_2}. 
\end{proof}

\begin{proposition}
Suppose that $\sigma(\mbb{A}_{\infty}^{\mc{E}}) \bigcap \sigma(\mc{K}_D^a) = \emptyset$. Then for each $j$, when $\tau$ is large enough, there exists an eigenfunction $H_j(\tau)$ that is divergence-free.
\end{proposition}
 
\begin{proof}
Let $\w{H}_j \in \w{{\bf H}}^1_p(Y_h) \oplus \C^3$ be an eigenfunction of $\mbb{A}_\infty^{\mc{E}}$ for the eigenvalue $\mu_j$ with $\norm{\w{H}_j}_{Y_h} = 1$, i.e.,
\begin{align} \label{eq:limit_eig}
A_\infty^{\mc{E}}(\w{H}_j,\w{H}_j) = \mu_j (\w{H}_j, \w{H}_j)_{Y_h}\,.    
\end{align}
We define $H_j \in {\bf H}_p^1(Y_h)$ by, for $\tau > 0$,  
\begin{equation*}
    H_j(\tau) = \frac{{\rm E}_{\mbb{A}_{\tau,\ww_j(\tau)}^{\mc{E}}}(\tau \ww_j^2(\tau))[\w{H}_j] }{\norm{{\rm E}_{\mbb{A}_{\tau,\ww_j(\tau)}^{\mc{E}}}(\tau \ww_j^2(\tau))[\w{H}_j]}_{Y_h}}\,,
\end{equation*}
that satisfies \eqref{eig_probreal}. By the strong
convergence \eqref{eq:conver_2}, there holds $\norm{H_j(\tau) - \w{H}_j}_{Y_h} \to 0$ as $\tau \to \infty$. We claim that the field $H_j$ is divergence-free. For this, we consider $u_j(\tau) = \ep^{-1} \ddiv H_j(\tau)$. It is easy to show as in \cite[Theorem 4.3]{bao1997variational} that $u_j$ satisfies $\Delta u_j + \ww_j(\tau)^2 \ep u_j = 0$ with the boundary condition $\frac{\p}{\p \nu} u_j = \mc{T}^{\mc{E}} u_j$. If $\ddiv H_j \neq 0$, then $(\ww_j(\tau)^2, u_j)$ is an eigenpair for the problem. 
By Appendix \ref{app:a}, we know that $\tau \ww_j(\tau)^2 \to \d_{j'}$ for some $\d_{j'} \in \sigma(\mc{K}_D^a)$, which contradicts with the assumption $\sigma(\mbb{A}_{\infty}^{\mc{E}}) \bigcap \sigma(\mc{K}_D^a) = \emptyset$. The proof is complete. 
\end{proof} 

We are now ready to give the proof of Theorem \ref{mainthm:real}.
respectively. Similar notions apply to other function spaces of vector fields, e.g., ${\bf H}^1(Y_h)$.

\begin{proof}[Proof of Theorem \ref{mainthm:real}]
Let ${\bf H}_{p,sym}^1(Y_h)$ and ${\bf H}_{p, ant}^1(Y_h)$ be the
symmetric part and the antisymmetric part of ${\bf H}_{p}^1(Y_h)$, respectively. One can readily verify that these two function spaces are orthogonal with respect to the sesquilinear form $A_{\tau,\ww}$ in \eqref{eq:quadra_form} and the $L^2$-inner product. Therefore, if $H \in {\bf H}_{p, ant}^1(Y_h)$ satisfies \eqref{eq:quadra_form_eig} for all $\vp \in {\bf H}_{p, ant}^1(Y_h)$, then the equation \eqref{eq:quadra_form_eig} holds for all $ \vp \in {\bf H}_{p}^1(Y_h)$. We define $\lad_j^{ant}(\tau,\ww)$ as in \eqref{eq:min_max} by restricting the fields $H$ in ${\bf H}_{p,ant}^1(Y_h)$. By similar arguments as above, we have that for each $j$, when $\tau$ is large enough, there exists $\ww_j^{ant}$ such that $\lad_j(\tau, \ww_j^{ant}) = \tau (\ww_j^{ant})^2$ with $ \ww_j(\tau)^{ant} \to 0$ as $\tau \to 0$, and we can find an associated divergence-free field $H_j \in {\bf H}_{p,ant}^1(Y_h)$ satisfying $A_{\tau,\ww_j^{ant}}^{\mc{E}}(H_j,H_j) = \tau (\ww_j^{ant})^2 (H_j,H_j)_{Y_h}$. Then by Lemma \ref{}, to complete the proof, it suffices to show \eqref{}. 
\end{proof}
\fi



\begin{thebibliography}{10}

\bibitem{ammari2020equivalent}
H.~Ammari, D.~P. Challa, A.~P. Choudhury, and M.~Sini.
\newblock The equivalent media generated by bubbles of high contrasts:
  Volumetric metamaterials and metasurfaces.
\newblock {\em Multiscale Modeling \& Simulation}, 18(1):240--293, 2020.

\bibitem{ammari2019subwavelength}
H.~Ammari, A.~Dabrowski, B.~Fitzpatrick, P.~Millien, and M.~Sini.
\newblock Subwavelength resonant dielectric nanoparticles with high refractive
  indices.
\newblock {\em Mathematical Methods in the Applied Sciences},
  42(18):6567--6579, 2019.

\bibitem{ammari2021bound}
H.~Ammari, B.~Davies, E.~O. Hiltunen, H.~Lee, and S.~Yu.
\newblock Bound states in the continuum and fano resonances in subwavelength
  resonator arrays.
\newblock {\em Journal of Mathematical Physics}, 62(10):101506, 2021.

\bibitem{ammari2020exceptional}
H.~Ammari, B.~Davies, E.~O. Hiltunen, H.~Lee, and S.~Yu.
\newblock Exceptional points in parity--time-symmetric subwavelength
  metamaterials.
\newblock {\em SIAM Journal on Mathematical Analysis}, 54(6):6223--6253, 2022.

\bibitem{ammari2017mathematical}
H.~Ammari, B.~Fitzpatrick, D.~Gontier, H.~Lee, and H.~Zhang.
\newblock A mathematical and numerical framework for bubble meta-screens.
\newblock {\em SIAM Journal on Applied Mathematics}, 77(5):1827--1850, 2017.

\bibitem{ammari2018mathematical}
H.~Ammari, B.~Fitzpatrick, H.~Kang, M.~Ruiz, S.~Yu, and H.~Zhang.
\newblock {\em Mathematical and computational methods in photonics and
  phononics}, volume 235.
\newblock American Mathematical Soc., 2018.

\bibitem{ammari2020meta}
H.~Ammari, B.~Li, and J.~Zou.
\newblock Mathematical analysis of electromagnetic plasmonic metasurfaces.
\newblock {\em Multiscale Modeling \& Simulation}, 18(2):758--797, 2020.

\bibitem{ammari2020superresolution}
H.~Ammari, B.~Li, and J.~Zou.
\newblock Superresolution in recovering embedded electromagnetic sources in
  high contrast media.
\newblock {\em SIAM Journal on Imaging Sciences}, 13(3):1467--1510, 2020.

\bibitem{ammari2020mathematical}
H.~Ammari, B.~Li, and J.~Zou.
\newblock Mathematical analysis of electromagnetic scattering by dielectric
  nanoparticles with high refractive indices.
\newblock {\em Transactions of the American Mathematical Society},
  376(01):39--90, 2023.

\bibitem{amrouche1998vector}
C.~Amrouche, C.~Bernardi, M.~Dauge, and V.~Girault.
\newblock Vector potentials in three-dimensional non-smooth domains.
\newblock {\em Mathematical Methods in the Applied Sciences}, 21(9):823--864,
  1998.

\bibitem{artin2011algebra}
M.~Artin.
\newblock {\em Algebra}.
\newblock Pearson Prentice Hall, 2011.

\bibitem{bao1997variational}
G.~Bao.
\newblock Variational approximation of maxwell's equations in biperiodic
  structures.
\newblock {\em SIAM Journal on Applied Mathematics}, 57(2):364--381, 1997.

\bibitem{bao2000scattering}
G.~Bao and D.~Dobson.
\newblock On the scattering by a biperiodic structure.
\newblock {\em Proceedings of the American Mathematical Society},
  128(9):2715--2723, 2000.

\bibitem{bonnet1994guided}
A.-S. Bonnet-Bendhia and F.~Starling.
\newblock Guided waves by electromagnetic gratings and non-uniqueness examples
  for the diffraction problem.
\newblock {\em Mathematical Methods in the Applied Sciences}, 17(5):305--338,
  1994.

\bibitem{bulgakov2017bound}
E.~N. Bulgakov and D.~N. Maksimov.
\newblock Bound states in the continuum and polarization singularities in
  periodic arrays of dielectric rods.
\newblock {\em Physical Review A}, 96(6):063833, 2017.

\bibitem{bulgakov2017topological}
E.~N. Bulgakov and D.~N. Maksimov.
\newblock Topological bound states in the continuum in arrays of dielectric
  spheres.
\newblock {\em Physical review letters}, 118(26):267401, 2017.

\bibitem{bulgakov2018optical}
E.~N. Bulgakov and D.~N. Maksimov.
\newblock Optical response induced by bound states in the continuum in arrays
  of dielectric spheres.
\newblock {\em JOSA B}, 35(10):2443--2452, 2018.

\bibitem{bulgakov2017propagating}
E.~N. Bulgakov and A.~F. Sadreev.
\newblock Propagating bloch bound states with orbital angular momentum above
  the light line in the array of dielectric spheres.
\newblock {\em JOSA A}, 34(6):949--952, 2017.

\bibitem{bulgakov2019high}
E.~N. Bulgakov and A.~F. Sadreev.
\newblock High-q resonant modes in a finite array of dielectric particles.
\newblock {\em Physical Review A}, 99(3):033851, 2019.

\bibitem{cao2022electromagnetic}
X.~Cao, A.~Ghandriche, and M.~Sini.
\newblock The electromagnetic waves generated by dielectric nanoparticles.
\newblock {\em arXiv preprint arXiv:2209.02413}, 2022.

\bibitem{chen2016review}
H.-T. Chen, A.~J. Taylor, and N.~Yu.
\newblock A review of metasurfaces: physics and applications.
\newblock {\em Reports on progress in physics}, 79(7):076401, 2016.

\bibitem{chen2019singularities}
W.~Chen, Y.~Chen, and W.~Liu.
\newblock Singularities and poincar{\'e} indices of electromagnetic multipoles.
\newblock {\em Physical review letters}, 122(15):153907, 2019.

\bibitem{chitour2016generic}
Y.~Chitour, D.~Kateb, and R.~Long.
\newblock Generic properties of the spectrum of the stokes system with
  dirichlet boundary condition in r3.
\newblock {\em Annales de l'Institut Henri Poincar{\'e} C, Analyse non
  lin{\'e}aire}, 33(1):119--167, 2016.

\bibitem{costabel2010volume}
M.~Costabel, E.~Darrigrand, and E.~H. Kon\'{e}.
\newblock Volume and surface integral equations for electromagnetic scattering
  by a dielectric body.
\newblock {\em Journal of Computational and Applied Mathematics},
  234(6):1817--1825, 2010.

\bibitem{costabel2012essential}
M.~Costabel, E.~Darrigrand, and H.~Sakly.
\newblock The essential spectrum of the volume integral operator in
  electromagnetic scattering by a homogeneous body.
\newblock {\em Comptes Rendus Mathematique}, 350(3-4):193--197, 2012.

\bibitem{delfour2011shapes}
M.~C. Delfour and J.-P. Zol{\'e}sio.
\newblock {\em Shapes and geometries: metrics, analysis, differential calculus,
  and optimization}.
\newblock SIAM, 2011.

\bibitem{dennis2001topological}
M.~R. Dennis.
\newblock {\em Topological singularities in wave fields}.
\newblock PhD thesis, University of Bristol, 2001.

\bibitem{dobson1994variational}
D.~Dobson.
\newblock A variational method for electromagnetic diffraction in biperiodic
  structures.
\newblock {\em ESAIM: Mathematical Modelling and Numerical Analysis},
  28(4):419--439, 1994.

\bibitem{dyatlov2019mathematical}
S.~Dyatlov and M.~Zworski.
\newblock {\em Mathematical theory of scattering resonances}, volume 200.
\newblock American Mathematical Soc., 2019.

\bibitem{evlyukhin2012demonstration}
A.~B. Evlyukhin, S.~M. Novikov, U.~Zywietz, R.~L. Eriksen, C.~Reinhardt, S.~I.
  Bozhevolnyi, and B.~N. Chichkov.
\newblock Demonstration of magnetic dipole resonances of dielectric nanospheres
  in the visible region.
\newblock {\em Nano letters}, 12(7):3749--3755, 2012.

\bibitem{fano1961effects}
U.~Fano.
\newblock Effects of configuration interaction on intensities and phase shifts.
\newblock {\em Physical Review}, 124(6):1866, 1961.

\bibitem{garcia2011strong}
A.~Garc{\'\i}a-Etxarri, R.~G{\'o}mez-Medina, L.~S. Froufe-P{\'e}rez,
  C.~L{\'o}pez, L.~Chantada, F.~Scheffold, J.~Aizpurua, M.~Nieto-Vesperinas,
  and J.~J. S{\'a}enz.
\newblock Strong magnetic response of submicron silicon particles in the
  infrared.
\newblock {\em Optics express}, 19(6):4815--4826, 2011.

\bibitem{girault2012finite}
V.~Girault and P.~Raviart.
\newblock {\em Finite element methods for {N}avier-Stokes equations: theory and
  algorithms}, volume~5.
\newblock Springer Science \& Business Media, 2012.

\bibitem{gohberg1990classes}
I.~Gohberg, S.~Goldberg, and M.~A. Kaashoek.
\newblock {\em Classes of linear operators}, volume~49.
\newblock Birkh{\"a}user, 1990.

\bibitem{gokhberg1971operator}
I.~Gokhberg and E.~I. Sigal.
\newblock An operator generalization of the logarithmic residue theorem and the
  theorem of rouch{\'e}.
\newblock {\em Matematicheskii Sbornik}, 126(4):607--629, 1971.

\bibitem{hazard1996solution}
C.~Hazard and M.~Lenoir.
\newblock On the solution of time-harmonic scattering problems for maxwell’s
  equations.
\newblock {\em SIAM Journal on Mathematical Analysis}, 27(6):1597--1630, 1996.

\bibitem{hempel2000spectral}
R.~Hempel and K.~Lienau.
\newblock Spectral properties of periodic media in the large coupling limit:
  Properties of periodic media.
\newblock {\em Communications in Partial Differential Equations},
  25(7-8):1445--1470, 2000.

\bibitem{henrot2018shape}
A.~Henrot and M.~Pierre.
\newblock {\em Shape Variation and Optimization: A Geometrical Analysis},
  volume~28.
\newblock European Mathematical Society (EMS), 2018.

\bibitem{henry2005perturbation}
D.~Henry.
\newblock {\em Perturbation of the boundary in boundary-value problems of
  partial differential equations}.
\newblock Cambridge University Press, 2005.

\bibitem{hewitt2013abstract}
E.~Hewitt and K.~A. Ross.
\newblock {\em Abstract Harmonic Analysis: Volume II: Structure and Analysis
  for Compact Groups Analysis on Locally Compact Abelian Groups}, volume 152.
\newblock Springer, 2013.

\bibitem{hislop2012introduction}
P.~D. Hislop and I.~M. Sigal.
\newblock {\em Introduction to spectral theory: With applications to
  Schr{\"o}dinger operators}, volume 113.
\newblock Springer Science \& Business Media, 2012.

\bibitem{hsu2016bound}
C.~W. Hsu, B.~Zhen, A.~D. Stone, J.~D. Joannopoulos, and M.~Solja{\v{c}}i{\'c}.
\newblock Bound states in the continuum.
\newblock {\em Nature Reviews Materials}, 1(9):1--13, 2016.

\bibitem{huidobro2016graphene}
P.~A. Huidobro, M.~Kraft, S.~A. Maier, and J.~B. Pendry.
\newblock Graphene as a tunable anisotropic or isotropic plasmonic metasurface.
\newblock {\em ACS nano}, 10(5):5499--5506, 2016.

\bibitem{jahani2016all}
S.~Jahani and Z.~Jacob.
\newblock All-dielectric metamaterials.
\newblock {\em Nature nanotechnology}, 11(1):23--36, 2016.

\bibitem{kato2013perturbation}
T.~Kato.
\newblock {\em Perturbation theory for linear operators}, volume 132.
\newblock Springer Science \& Business Media, 2013.

\bibitem{kaxiras2019quantum}
E.~Kaxiras and J.~D. Joannopoulos.
\newblock {\em Quantum theory of materials}.
\newblock Cambridge university press, 2019.

\bibitem{kirsch2018limiting}
A.~Kirsch and A.~Lechleiter.
\newblock The limiting absorption principle and a radiation condition for the
  scattering by a periodic layer.
\newblock {\em SIAM Journal on Mathematical Analysis}, 50(3):2536--2565, 2018.

\bibitem{knapp2001representation}
A.~W. Knapp.
\newblock {\em Representation theory of semisimple groups: an overview based on
  examples}.
\newblock Princeton university press, 2001.

\bibitem{koshelev2018asymmetric}
K.~Koshelev, S.~Lepeshov, M.~Liu, A.~Bogdanov, and Y.~Kivshar.
\newblock Asymmetric metasurfaces with high-q resonances governed by bound
  states in the continuum.
\newblock {\em Physical review letters}, 121(19):193903, 2018.

\bibitem{krasnok2012all}
A.~E. Krasnok, A.~E. Miroshnichenko, P.~A. Belov, and Y.~S. Kivshar.
\newblock All-dielectric optical nanoantennas.
\newblock {\em Optics Express}, 20(18):20599--20604, 2012.

\bibitem{kuchment2001mathematics}
P.~Kuchment.
\newblock The mathematics of photonic crystals.
\newblock In {\em Mathematical modeling in optical science}, pages 207--272.
  SIAM, 2001.

\bibitem{kuchment2012floquet}
P.~Kuchment.
\newblock {\em Floquet theory for partial differential equations}, volume~60.
\newblock Birkh{\"a}user, 2012.

\bibitem{kuchment2016overview}
P.~Kuchment.
\newblock An overview of periodic elliptic operators.
\newblock {\em Bulletin of the American Mathematical Society}, 53(3):343--414,
  2016.

\bibitem{kuznetsov2016optically}
A.~I. Kuznetsov, A.~E. Miroshnichenko, M.~L. Brongersma, Y.~S. Kivshar, and
  B.~Luk’yanchuk.
\newblock Optically resonant dielectric nanostructures.
\newblock {\em Science}, 354(6314):aag2472, 2016.

\bibitem{leroy2015superabsorption}
V.~Leroy, A.~Strybulevych, M.~Lanoy, F.~Lemoult, A.~Tourin, and J.~H. Page.
\newblock Superabsorption of acoustic waves with bubble metascreens.
\newblock {\em Physical Review B}, 91(2):020301, 2015.

\bibitem{Li2019}
H.~Li, S.~Li, H.~Liu, and X.~Wang.
\newblock Analysis of electromagnetic scattering from plasmonic inclusions
  beyond the quasi-static approximation and applications.
\newblock {\em ESAIM: Mathematical Modelling and Numerical Analysis},
  53:1351--1371, 7 2019.

\bibitem{LLZ0527}
H.~Li, H.~Liu, and J.~Zou.
\newblock Minnaert resonances for bubbles in soft elastic materials.
\newblock {\em SIAM Journal on Applied Mathematics}, 82(1):119--141, 2022.

\bibitem{li2019symmetry}
S.~Li, C.~Zhou, T.~Liu, and S.~Xiao.
\newblock Symmetry-protected bound states in the continuum supported by
  all-dielectric metasurfaces.
\newblock {\em Physical Review A}, 100(6):063803, 2019.

\bibitem{limonov2017fano}
M.~F. Limonov, M.~V. Rybin, A.~N. Poddubny, and Y.~S. Kivshar.
\newblock Fano resonances in photonics.
\newblock {\em Nature Photonics}, 11(9):543--554, 2017.

\bibitem{lin2020mathematical}
J.~Lin, S.~P. Shipman, and H.~Zhang.
\newblock A mathematical theory for fano resonance in a periodic array of
  narrow slits.
\newblock {\em SIAM Journal on Applied Mathematics}, 80(5):2045--2070, 2020.

\bibitem{lin2021fano}
J.~Lin and H.~Zhang.
\newblock Fano resonance in metallic grating via strongly coupled subwavelength
  resonators.
\newblock {\em European Journal of Applied Mathematics}, 32(2):370--394, 2021.

\bibitem{luk2010fano}
B.~Luk'Yanchuk, N.~I. Zheludev, S.~A. Maier, N.~J. Halas, P.~Nordlander,
  H.~Giessen, and C.~T. Chong.
\newblock The fano resonance in plasmonic nanostructures and metamaterials.
\newblock {\em Nature materials}, 9(9):707--715, 2010.

\bibitem{luk2015optimum}
B.~S. Luk’yanchuk, N.~V. Voshchinnikov, R.~Paniagua-Dom{\'\i}nguez, and A.~I.
  Kuznetsov.
\newblock Optimum forward light scattering by spherical and spheroidal
  dielectric nanoparticles with high refractive index.
\newblock {\em ACS Photonics}, 2(7):993--999, 2015.

\bibitem{meklachi2018asymptotic}
T.~Meklachi, S.~Moskow, and J.~C. Schotland.
\newblock Asymptotic analysis of resonances of small volume high contrast
  linear and nonlinear scatterers.
\newblock {\em Journal of Mathematical Physics}, 59(8):083502, 2018.

\bibitem{miroshnichenko2010fano}
A.~E. Miroshnichenko, S.~Flach, and Y.~S. Kivshar.
\newblock Fano resonances in nanoscale structures.
\newblock {\em Reviews of Modern Physics}, 82(3):2257, 2010.

\bibitem{nedelec2001acoustic}
J.-C. N{\'e}d{\'e}lec.
\newblock {\em Acoustic and electromagnetic equations: integral representations
  for harmonic problems}, volume 144.
\newblock Springer, 2001.

\bibitem{nedelec1991integral}
J.-C. N\'{e}d\'{e}lec and F.~Starling.
\newblock Integral equation methods in a quasi-periodic diffraction problem for
  the time-harmonic maxwell’s equations.
\newblock {\em SIAM Journal on Mathematical Analysis}, 22(6):1679--1701, 1991.

\bibitem{sakly2017shape}
H.~Sakly.
\newblock Shape derivative of the volume integral operator in electromagnetic
  scattering by homogeneous bodies.
\newblock {\em Mathematical Methods in the Applied Sciences},
  40(18):7125--7138, 2017.

\bibitem{sarid2010modern}
D.~Sarid and W.~A. Challener.
\newblock {\em Modern introduction to surface plasmons: theory, Mathematica
  modeling, and applications}.
\newblock Cambridge University Press, 2010.

\bibitem{schmidt2004electromagnetic}
G.~Schmidt.
\newblock Electromagnetic scattering by periodic structures.
\newblock {\em Journal of Mathematical Sciences}, 124(6):5390--5406, 2004.

\bibitem{shipman2010resonant}
S.~Shipman.
\newblock Resonant scattering by open periodic waveguides.
\newblock {\em Progress in Computational Physics (PiCP)}, 1:7--49, 2010.

\bibitem{shipman2007guided}
S.~Shipman and D.~Volkov.
\newblock Guided modes in periodic slabs: existence and nonexistence.
\newblock {\em SIAM Journal on Applied Mathematics}, 67(3):687--713, 2007.

\bibitem{shipman2012total}
S.~P. Shipman and H.~Tu.
\newblock Total resonant transmission and reflection by periodic structures.
\newblock {\em SIAM Journal on Applied Mathematics}, 72(1):216--239, 2012.

\bibitem{shipman2005resonant}
S.~P. Shipman and S.~Venakides.
\newblock Resonant transmission near nonrobust periodic slab modes.
\newblock {\em Physical Review E}, 71(2):026611, 2005.

\bibitem{shipman2013resonant}
S.~P. Shipman and A.~T. Welters.
\newblock Resonant electromagnetic scattering in anisotropic layered media.
\newblock {\em Journal of Mathematical Physics}, 54(10):103511, 2013.

\bibitem{sidorenko2021observation}
M.~Sidorenko, O.~Sergaeva, Z.~Sadrieva, C.~Roques-Carmes, P.~Muraev,
  D.~Maksimov, and A.~Bogdanov.
\newblock Observation of an accidental bound state in the continuum in a chain
  of dielectric disks.
\newblock {\em Physical Review Applied}, 15(3):034041, 2021.

\bibitem{simon1975methods}
B.~Simon.
\newblock {\em Methods of Modern Mathematical Physics: Fourier Analysis,
  Self-Adjointness}.
\newblock Academic Press, New York, 1975.

\bibitem{simon1978canonical}
B.~Simon.
\newblock A canonical decomposition for quadratic forms with applications to
  monotone convergence theorems.
\newblock {\em Journal of Functional Analysis}, 28(3):377--385, 1978.

\bibitem{simon1978methods}
B.~Simon.
\newblock {\em Methods of Modern Mathematical Physics: Analysis of Operators}.
\newblock Academic Press, New York, 1978.

\bibitem{soskin2001singular}
M.~Soskin and M.~Vasnetsov.
\newblock Singular optics.
\newblock {\em Progress in optics}, 42(4):219--276, 2001.

\bibitem{staude2017metamaterial}
I.~Staude and J.~Schilling.
\newblock Metamaterial-inspired silicon nanophotonics.
\newblock {\em Nature Photonics}, 11(5):274, 2017.

\bibitem{stessin2011analyticity}
M.~Stessin, R.~Yang, and K.~Zhu.
\newblock Analyticity of a joint spectrum and a multivariable analytic
  {F}redholm theorem.
\newblock {\em The New York Journal of Mathematics [electronic only]},
  17:39--44, 2011.

\bibitem{sun2019electromagnetic}
S.~Sun, Q.~He, J.~Hao, S.~Xiao, and L.~Zhou.
\newblock Electromagnetic metasurfaces: physics and applications.
\newblock {\em Advances in Optics and Photonics}, 11(2):380--479, 2019.

\bibitem{MichaelTaylor2010s}
M.~Taylor.
\newblock Multidimensional {A}nalytic {F}redholm {T}heory.
\newblock
  \url{https://mtaylor.web.unc.edu/wp-content/uploads/sites/16915/2018/04/fred.pdf},
  2018.
\newblock [Online].

\bibitem{teytel1999rare}
M.~Teytel.
\newblock How rare are multiple eigenvalues?
\newblock {\em Communications on Pure and Applied Mathematics: A Journal Issued
  by the Courant Institute of Mathematical Sciences}, 52(8):917--934, 1999.

\bibitem{tzarouchis2018light}
D.~Tzarouchis and A.~Sihvola.
\newblock Light scattering by a dielectric sphere: perspectives on the {M}ie
  resonances.
\newblock {\em Applied Sciences}, 8(2):184, 2018.

\bibitem{weber1980local}
C.~Weber and P.~Werner.
\newblock A local compactness theorem for maxwell's equations.
\newblock {\em Mathematical Methods in the Applied Sciences}, 2(1):12--25,
  1980.

\bibitem{yoda2020generation}
T.~Yoda and M.~Notomi.
\newblock Generation and annihilation of topologically protected bound states
  in the continuum and circularly polarized states by symmetry breaking.
\newblock {\em Physical Review Letters}, 125(5):053902, 2020.

\bibitem{yuan2021parametric}
L.~Yuan and Y.~Y. Lu.
\newblock Parametric dependence of bound states in the continuum in periodic
  structures: vectorial cases.
\newblock {\em arXiv preprint arXiv:2105.13764}, 2021.

\bibitem{zhang2016advances}
L.~Zhang, S.~Mei, K.~Huang, and C.-W. Qiu.
\newblock Advances in full control of electromagnetic waves with metasurfaces.
\newblock {\em Advanced Optical Materials}, 4(6):818--833, 2016.

\bibitem{zhen2014topological}
B.~Zhen, C.~W. Hsu, L.~Lu, A.~D. Stone, and M.~Solja{\v{c}}i{\'c}.
\newblock Topological nature of optical bound states in the continuum.
\newblock {\em Physical review letters}, 113(25):257401, 2014.

\end{thebibliography}
\end{document}